%% file: main.tex
\documentclass[11pt]{article}

\usepackage[margin=1 in]{geometry}
\usepackage[utf8]{inputenc}
\usepackage[english]{babel}

\usepackage{framed}
\usepackage{enumerate}
\usepackage{tcolorbox}
\usepackage{amsmath}
\usepackage{amsthm}
\usepackage{appendix}

\usepackage{float}
\usepackage{listings}

\usepackage[linesnumbered,ruled,commentsnumbered]{algorithm2e}
\usepackage{algpseudocode}
\usepackage{framed}
\usepackage{newfloat}
\usepackage{caption}
\usepackage{tikz}
\usepackage{amsmath,amssymb}
\usepackage{graphicx}

\usepackage{lineno, blindtext}
\usepackage{soul}
\usepackage{ulem}
\usepackage{cancel}
\usepackage{subcaption}
\usepackage{wrapfig}
\usepackage{tocbibind}

\usepackage[square,numbers]{natbib}
\newcommand\CausalEC{\texttt{CausalEC}}

\newcommand\receive{\texttt{receive}}

\newcommand\remove[1]{}

\newtheorem{theorem}{Theorem}[section]
\newtheorem{definition}{Definition}
\newtheorem{example}{Example}

\newtheorem{lemma}{Lemma}[section]
\newtheorem{claim}{Claim}[lemma]

\title{CausalEC: A Causally Consistent Data Storage Algorithm based on Cross-Object Erasure Coding}
\author{
  Viveck R. Cadambe, Shihang Lyu\footnote{Shihang Lyu is currently a Software Engineer at Google.} \\
  Electrical Engineering Department, School of EECS\\
  Pennsylvania State University, USA\\
  \texttt{viveck@psu.edu, shihang.lyu@gmail.com}
}

\date{}
\begin{document}
\maketitle

\begin{abstract}
Current causally consistent data storage algorithms use  partial or full replication to ensure data access to clients over a distributed setting. We develop, for the first time, an \textit{erasure coding} based algorithm called \CausalEC~ that ensures causal consistency for a collection of read-write objects stored in a distributed set of nodes over an asynchronous message passing system. \CausalEC~ can use an arbitrary linear erasure code for data storage, and ensures liveness, fault-tolerance and storage properties prescribed by the erasure code. 

\CausalEC~~ retains a key benefit of previous replication-based algorithms - every write operation is ``local'', that is, a server performs only local actions before returning to a client that issued a write operation.  For servers that store certain objects in an uncoded manner, read operations to those objects also return locally. In general, a read operation to an object can be returned by a server on contacting a small subset of other servers so long as the underlying erasure code allows for the object to be decoded from that subset.  Unlike previous consistent erasure coding based algorithms, \CausalEC~ is compatible with \textit{cross-object} erasure coding, where nodes encode values across multiple objects. \CausalEC~ navigates the technical challenges of cross-object erasure coding, in particular, pertaining to re-encoding when writes update the values and ensuring that concurrent reads are served in a non-blocking manner during the transition to storing codeword symbols corresponding to the updated values.
\end{abstract}
\maketitle
\newpage

\section{Introduction}
\normalem


\input{Intro.tex}

\input{challenges.tex}

\section{The \CausalEC~ Algorithm}
\label{chap:algorithm}

\input{algorithm}

\section{Correctness and Performance}
\label{chap:proofs}
\input{proofs}

\section{Related Works}
\label{sec:related}

\input{related}

\input{conclusion}

  \newpage

  \bibliography{cadambe-references,Biblio-Database}

\appendix
\input{AppendixA}

\input{AppendixB}

\input{AppendixC}
\end{document}

%% file: Intro.tex
\label{sec:Introduction}

Consistent data storage services such as Amazon DynamoDB and Apache Cassandra are important components of modern cloud computing infrastructure. 
The focus of this paper is the design of low-latency cost-effective \textit{causally consistent read/write data stores}, that have received significant attention in recent research \cite{du2014gentlerain, du2013orbe, lloyd2011don, mahajan2011consistency, attiya2017limitations}. Since causal consistency can be maintained via protocols where server nodes respond to client operations after performing only \textit{local} read and write operations (i.e., without requiring responses from other servers), it incurs much lower latency as compared to stronger consistency criteria (such as linearizability \cite{Herlihy_Wing}).
Classically, causally consistent data stores require every object to be replicated at every server node \cite{ahamad1995causal,lloyd2011don,du2013orbe,du2014gentlerain}. The requirement of every server having a replica of every object can be prohibitively expensive for large data stores. To address this, there is much recent interest in causally consistent data stores based on \textit{partial replication} \cite{shen2015causal, xiang2019partially, xiang2020global, helary2006efficiency,bravo2017saturn, mehdi2017can, akkoorath2016cure}. In partially replicated data stores, a given server node stores only a subset of the objects. Partial replication enables a system designer to trade-off latency for storage cost, since it allows low latency local read/write operations for the clients that access the data objects stored at a nearby server, and provides a lower level of service for objects that are not stored at that server. 

Erasure coding is a generalization of  replication that offers much lower storage costs as compared to (partial) replication for the same degree of fault-tolerance. In erasure coding, each server stores a \emph{codeword symbol} that is a general function of the data, unlike replication where servers are restricted to store copies of partitions of the data. 
Motivated by the promise of reduced costs, several consistent data storage algorithms that use erasure coding have been developed \cite{Konwar_PODC2017, goodson2004efficient, dobre_powerstore, Hendricks, Dutta, GWGR, CT, cadambe2017coded, FAB,Uluyol_Pando2020, chen2017giza, wang2020craft,zare2021legostore}. However, the use of erasure coding in the context of \emph{causal} consistency is nearly non-existent in literature\footnote{The only exception we are aware of is \cite{lyu2018erasure}, which studies a very limited class of codes and data access patterns, and has very weak liveness properties (see Sec. \ref{sec:related})}. In fact, at the outset, erasure coding appears to be incompatible with the key benefit of causal consistency - the ability to keep latency low by serving several operations locally. Specifically, previous works\footnote{References \cite{eikel2014robust,lyu2018erasure,AJX} that apply cross-object erasure coding are discussed in Sec. \ref{sec:related}.}  (for other consistency criteria) apply erasure coding by partitioning an object value into fragments (called \textit{data} fragments), then encoding the fragments to  redundant \textit{parity} fragments, and storing each fragment on one server. Since no server stores an object value in its entirety, it is impossible to serve reads locally at any server. For this reason, erasure coding is well known to be a technique that can lower costs, but generally incurs a higher latency than replication schemes. 

In this paper, we overcome the apparent latency penalty of erasure coding via a new approach. Unlike previous approaches that partition and encode a single object value, we develop a novel causally consistent distributed algorithm that is compatible with  \textit{cross-object erasure coding}. As we show next, cross-object erasure coding can achieve  significantly lower latencies than even the best partial replication schemes for a given storage cost. Thereby, our work opens new, desirable, operating points on the cost-latency trade-offs for data store design.    

\subsection{Motivation for Cross-Object Erasure Coding}
\label{subsec:motivation}
\begin{figure*}[!tbp]
\centering
\resizebox{0.9 \textwidth}{!}{
\begin{subfigure}[!t]{0.44\linewidth}
\includegraphics[width=2.3in]{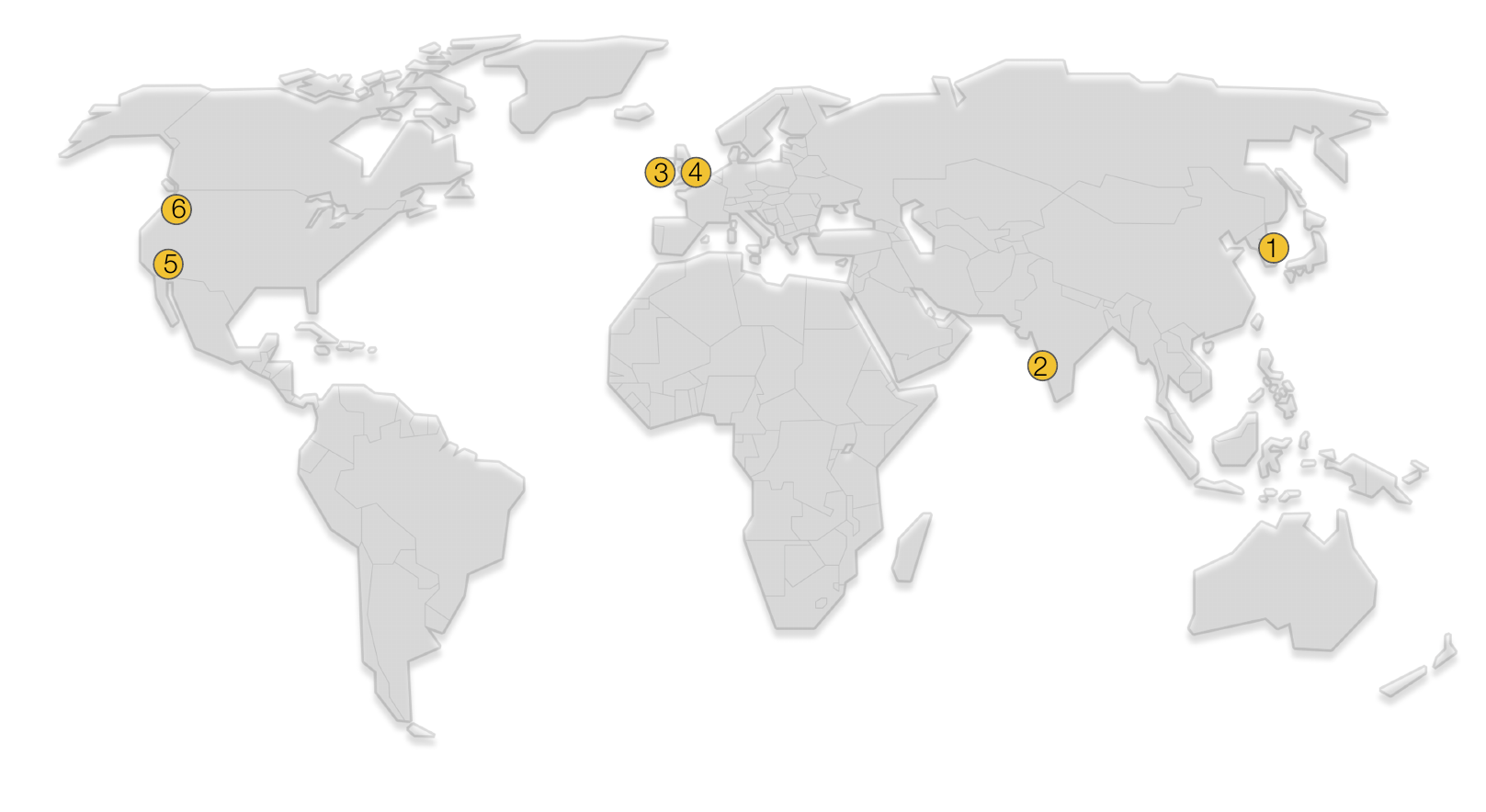}
\end{subfigure}
\begin{subfigure}[!t]{0.44\textwidth}\resizebox{ \textwidth}{!}{\begin{tabular}{|c|cccccc|}
  \hline
 Regions & Seoul & Mumbai & Ireland & London & N.  California & Oregon\\ 
  \hline
Seoul & 0 & 120 & 230 &  240  & 138 & 126\\ 
Mumbai &  120 & 0  & 121 & 113 & 228 & 220\\ 
Ireland & 230 & 121 & 0 & 13 & 138 & 126\\ 
London & 240 & 113 & 13& 0 & 146 & 137\\
N. California& 138 & 228 & 138 & 146 & 0 & 22 \\ 
Oregon& 146 & 220 & 126 & 137 & 22 & 0 \\
   \hline
\end{tabular}}
\end{subfigure}}
\caption{\footnotesize{Six DCs and their inter-DC round-trip-times (in ms) over AWS public cloud (obtained via \cite{cloudping} on Oct 2021).}}
\label{fig:AWS}
\vspace{-12pt}
\end{figure*}
 \begin{figure*}[!tbp]
\centering
\resizebox{0.9 \textwidth}{!}
{\begin{tabular}{|c|cccc|}
  \hline
 &  Worst-Case Latency & Average Latency & Communication Cost per read  & Communication cost per write\\ 
  \hline
Partial Replication & 228 ms & 88 ms & 3B/4 &  6B \\ Intra-Object Coding & 138 ms & 132 ms & 3B/4 &  6B/4   \\
Cross-Object Coding & 138 ms & 88 ms & 3B/4 &  12B   \\
   \hline
\end{tabular}}
\vspace{-3pt}
\caption{\footnotesize{Cost and latency comparisons between replication, intra-object coding and cross-object coding.}}
\label{table:cost_comparison}
\vspace{-10pt}
\end{figure*}

In cross-object erasure coding, the object value is not necessarily partitioned, but rather, each server stores a codeword symbol that is a function of the values of different objects. For example, the values of objects $X_1,X_2,X_3$ can be stored at $4$ servers as $x_1,x_2,x_3,x_1+x_2+x_3$\footnote{In the examples used to motivate our work in Section \ref{sec:Introduction}, we assume that object values come from a finite field, and $+$ denotes the addition operator over the field.}, where $x_i$ is the value of object $X_i,i\in\{1,2,3\}$. In this case, the first three servers can be utilized to provide local read operations respectively to $X_1,X_2,X_3$; such local reads would not be possible if $x_i$ was fragmented and encoded. To appreciate the promise of cross-object erasure coding, consider a hypothetical geo-distributed data store over $6$ data centers (DCs).  See Fig. \ref{fig:AWS} for a pictorial depiction of such a data store, along with inter-DC latencies measured as per Amazon AWS public cloud. Imagine that the data store contains $4M$ objects of $B$ bits each, and the storage capacity of each data center is $MB$ bits. Assume that the read requests to each of the $4M$ objects is spatially distributed in a uniform manner across the $6$ data centers (DCs) - that is, if the total average arrival rate to all the objects in the data store is $\lambda$ reads per second, then an average arrival rate at any one of the DCs to a specific object is $\frac{\lambda}{6 \times 4M}$. The performance of such data stores are commonly evaluated among three metrics (i) Tail/worst-case latencies, (ii) throughput, and (iii) communication costs. As no DC has the capacity to store all the objects, an inter-DC latency for some objects is inevitable. 

We assume that the latency is predictable and determined by the table in Fig. \ref{fig:AWS} - such assumptions are, in fact, commonly used for adaptive data placement in practical systems research \cite{wu2013spanstore,du2016pando,zare2021legostore}). Through a brute force search, we found that the worst-case latency for the {best} partial replication scheme where each DC stores at most $MB$ bits is $228 ms$. This optimal worst-case latency is achieved by partitioning the $4M$ objects into $4$ groups each, which we label as $\mathcal{X}_1, \mathcal{X}_2,\mathcal{X}_3,\mathcal{X}_4,$ and allocating the objects in $\mathcal{X}_{1}$ to Seoul and Ireland, $\mathcal{X}_{2}$ to Mumbai and London, and  $\mathcal{X}_{3},\mathcal{X}_{4}$ respectively to North California and Oregon. Now, consider the conventional ``intra-object'' approach to erasure coding utilized in \cite{cadambe2017coded, Konwar_ipdps2016, CT, Hendricks, dobre_powerstore, Dutta,androulaki2014erasure}, where each object value is partitioned into fragments and encoded using, say a Reed Solomon\footnote{Any maximum distance separable code can be used (see definition in \cite{roth2006introduction}).} code of length $6$ and dimension $4$. Each node stores $B/4$ bits per object for every object, and the codeword symbols in any $4$ nodes suffices to recover every object. Thus the latency for every request is the round-trip time to the third nearest neighbor. Consequently, the worst-case latency is $138ms,$ a whopping $90ms$ shaved off the replication scheme with the same storage cost! 

A drawback of the Reed Solomon coding approach is that every read request inevitably incurs a non-zero latency unlike the partial replication-based scheme; in fact, a minimum latency of $121$ms is incurred for all requests. In practice, this manifests itself as lower throughput for the erasure coding based data store. Due to Little's law  \cite{bertsekas2021data}\footnote{We ignore queuing delays due to excessive loading, so our analysis applies for well-provisioned systems.}, we use the \emph{average}  latency as a proportional estimate for the average throughput of our data store assuming  uniformly distributed loads. The average latency for the replication-based scheme is  $88.25ms,$ whereas the erasure coding scheme has an average latency of $132.5ms$. As per Little's law, the erasure coding based data store is likely to have a much lower throughput ($66\%$) of the replication-based scheme. 
Further, the replication-based scheme can adapt to spatially non-uniform workloads by placing objects at DCs with higher demand. 

The motivation for \emph{cross-object erasure coding} is that it can enjoy the worst-case performance of intra-object erasure coding, as well as the throughput and flexibility of partial replication. Consider a case where Seoul, Mumbai, Ireland, London, N. California and Oregon respectively store $\mathcal{X}_{1}+\mathcal{X}_{3}, \mathcal{X}_{2}+\mathcal{X}_{4},\mathcal{X}_{1},\mathcal{X}_{2},\mathcal{X}_{4},\mathcal{X}_{3}$. This scheme has a worst-case latency of $138ms$ and average latency of $87.5ms.$ From a communication viewpoint, partial replication, on average, incurs $3B/4$ bits of communication per read request on average, as $3$ in $4$ requests require communication with an external DC. The communication complexity of intra-object erasure coding is also $3B/4,$ as the three nearest neighbors send $B/4$ bits each for a request. However, the cross-object erasure coding scheme incurs an overhead of $3.33B/4 $ bits per request due to Mumbai and Seoul having to access $B$ bits from a remote DC for every request. These are summarized in Table \ref{table:cost_comparison} (see Sec. \ref{sec:communication} for a more refined characterization of communication costs, and Appendix \ref{app:write_cost} for a comparison with partial replication schemes of \cite{helary2006efficiency, xiang2019partially}). Evidently, a limitation of cross-object erasure coding is that the latency improvements incur an increased write communication cost. Notably, the latency benefits of cross-object coding in the context of Fig. \ref{fig:AWS} appears because the erasure code and data placement were carefully tuned\footnote{This code is not maximum distance separable, so it cannot be an instance of the Reed-Solomon code.} to the network round-trip times. The design of {cross-object} erasure codes that minimize average/worst-case latency for general topologies is an open problem. In this paper, we develop an algorithm that is compatible with an \emph{arbitrary} linear erasure code.

\subsection{Our Contribution: The \CausalEC ~ algorithm}

\label{subsec:challenges}

Our main technical contribution is \CausalEC, a causally consistent distributed algorithm for read/write data objects that uses erasure coding. \CausalEC~ is developed so that it can use \textit{any} linear erasure code including  cross-object coding techniques. \CausalEC~ satisfies the following properties:
{\begin{enumerate}[(I)]

\item A write to any object must return locally at every server node.

\item If the erasure code allows an object $X$ to be decodable from a subset $S$ of servers, then a read to object $X$ terminates with at most one round trip to the servers in $S$.

 \item In a fair execution with a finite number of writes, the amount of data stored at the servers is eventually equal to that prescribed by the erasure code.

\item In a fair execution with a finite number of writes, eventually, every read to a given object responds with the same value.
\end{enumerate}}

Property (I) preserves an essential aspect of causal consistency: that writes to any object at any server are local. Properties (II),(III) ensure that the algorithm inherits both the liveness properties and the storage cost of the underlying erasure code. In particular Property (II) implies that \CausalEC~inherits the fault-tolerance property of the erasure code\footnote{For example, if a maximum distance separable code such as Reed-Solomon code is used with length $N$ and dimension $k$, Property (II) readily implies that the algorithm tolerates $N-k$ failed nodes. For the same code, Property (III) ensures that each node stores a fraction of $1/k$ of the data.}.  Property (IV) is a well-known, desirable liveness property called \textit{eventual consistency} \cite{vogels2008eventually,attiya2017limitations}\footnote{Some references \cite{burckhardt2014principles,burckhardt2014replicated,attiya2017limitations} refer to requirement (IV) as quiescent consistency. In our formal statements, we show a stronger variant of this property called ``eventual visibility'' \cite{burckhardt2014replicated, burckhardt2014principles,attiya2017limitations}.}. The technical novelty of \CausalEC~lies in navigating the challenges of supporting cross-object erasure coding in its algorithm design, and in ensuring properties (I)-(IV) in an asynchronous environment with concurrent operations.   We emphasize that these challenges and our solutions to \CausalEC~ due to the requirement of supporting cross-object coding. In previous consistent erasure coding based algorithms - that do not support cross-object coding - it suffices to treat a codeword symbol as a ``blackbox'', and when an object is updated, a codeword symbol corresponding to the new object version simply replaces the older codeword symbol. Similarly, for these algorithms, reads are simply served by obtaining codeword symbols and using the decoding function that is concomitant with the error correcting code. On the other hand, in \CausalEC, a codeword symbol combines multiple object values - for example, the codeword symbol $x_1+x_2+x_3$ at a server combines objects $X_1,X_2,X_3$. So when one object, say $X_1$ is updated to a newer value, say $x_1'$, \CausalEC~uses certain properties of the encoding functions to transform the codeword symbol to $x_1'+x_2+x_3$. Because multiple objects are combined in a single codeword symbol and a read is only interested in a single object, ensuring successful decoding for serving read operations also requires new, different approaches in \CausalEC~ as compared with existing algorithms.


To obtain a deeper understanding of the main technical challenges solved by \CausalEC, it helps consider the following erasure code over $5$ nodes and $3$ objects:
$Y_1=X_1, Y_2=X_2, Y_3=X_3, Y_4=X_1+X_2+X_3, Y_5 = X_1+2X_2+X_3$, where $Y_i$ is stored in server $i, i \in \{1,2,3,4,5\}$.  Here, object values $X_1,X_2,X_3$ are elements over a finite field with odd characteristic. Note that the code can serve a read to object $X_1$ locally at node $1$ and a read to object $X_2$ locally at node $2$. A read request to server $5$ for object $X_2$ can be served by accessing the codeword symbol $Y_4$ at server $4$, and performing $Y_5-Y_4$ -  this read is not local and incurs a latency a round trip time between nodes $4$ and $5.$ We denote by $\mathcal{R}_{i},$ the \textit{recovery sets for object $X_i$}: $\mathcal{R}_{i}$ is a set of subsets of servers that suffice for reading $X_i.$ The minimal\footnote{If set $S$ is a recovery set for an object, so is any superset of $S$. We list the {minimal} recovery sets under the subset ordering.} recovery sets are:
$\mathcal{R}_{1}=\{\{1\},\{3,4,5\},\{2,3,4\},\{2,3,5\}\},$ $\mathcal{R}_{2}=\{\{2\},\{4,5\},\{1,3,4\},\{1,3,5\}\},$ $\mathcal{R}_{3}=\{\{3\},\{1,2,4\},\{1,2,5\},\{1,4,5\}\}\}.$ 

        Consider an execution $\beta$ where there are three writes to ${X}_{1}$, two writes to ${X}_{2},$ and two writes to $X_3.$ Let the nodes store $Y_1=X_1(2), Y_2=X_2(2), Y_3=X_3(1), Y_4=X_1(3)+X_2(1)+X_3(2), Y_5=X_1(2)+2X_2(1)+X_3(1)$, where $X_j(i)$ denotes the value of the $i$th write to object $X_j$\footnote{Here, we denote every version of an object by an integer index. The actual protocol we develop in Sec. \ref{chap:algorithm} implements an indexing of writes via vector timestamps.}. Due to asynchrony and the distributed nature, each node is not aware of the object versions at the other nodes. Suppose a read request for object $X_2$ arrives at node $5$, the server then sends queries to other servers since it cannot be served locally. Since $\{4,5\}$ is a recovery set for $X_2$, requirement (II) in Sec. \ref{sec:Introduction} dictates that the read must be served on obtaining a response from node $4$. Observe however that server $5$ cannot merely obtain $X_2$ as $Y_5-Y_4$ due to the mismatch in the object versions. Our protocol is designed to store some \emph{history lists}\footnote{The storage of history is common in erasure coding based algorithms, see \cite{spiegelman2016space,Cadambe_Wang_Lynch2016} and references therein.} where nodes store values of potentially several object versions (in uncoded form). This enables node $4$ to re-encode $Y_4$ to a version $Y_{4}'$ using information in its local history list and send it. The version $Y_{4}'$ is designed in a manner that enables node $5$ to further re-encode it to: $Y_{4}'' = X_1(2)+X_2(1)+X_3(1)$ using information in its local history list. On performing these transformations, the node can decode $X_{2}(1)$ as $Y_{5}-Y_{4}''.$ As an example, the version $Y_{4}'$ can be $Y_{4}-X_{1}(3)+(X_3(1)-X_{3}(2)) = X_2(1)+X_3(1),$ assuming that $X_1(3),X_3(1),X_3(2)$ are present in the local history of node $4$ at the time of receipt of the request from node $5$. Then node $5$ can obtain $Y_4'' = Y_4'+X_1(2),$ if $X_1(2)$ are present in the local history at node $5$ at the point of receipt of the message from node $4$. 
    
    It is reasonable to ask what happens if nodes $4,5$ do not have the stated object versions in their local histories?
Indeed, to avoid the long-term transient storage overheads and to satisfy requirement (III) of Sec. \ref{sec:Introduction}, severs purge their local histories. For instance, $X_1(2)$ eventually gets deleted in all the history lists as it has propagated to all the servers and may not be present in node $5$ to enable re-encoding for a later read. The key property of \CausalEC~ is that the local histories have the property that node $4$ can always re-encode and send a value $Y_4'$ that can be used by node $5$ to serve the read - for every read at node $5$ at every point of the execution. Crucially, the algorithm is wait-free - that is, every node (e.g. node $4$), on receiving a query, responds to it  immediately.  The  ``garbage collection'' conditions for deleting items from history lists are a key  technical aspect that enables \CausalEC~ to simultaneously provide  (i) wait-free termination of reads and (ii) eventual removal of old versions in every possible execution. 
The termination of reads also enables us to effectively re-encode the locally stored codeword symbols upon arrival of new data. For example, once $X_3(2)$ arrives at node $5$, it conducts an ``internal'' read to obtain $X_3(1),$ and then recompute the codeword as $Y_5-X_3(1)+X_3(2).$ The version $X_3(1)$ is added to the local history list and deleted based on the garbage collection conditions.

%% file: challenges.tex
\vspace{-5pt}
\section{System Model}
\label{sec:model}
 
\subsection{Deployment Setting}
We assume an asynchronous network where all the servers and connections are known as a priori and the only sources of asynchrony are processing and communication delays. We consider a set of server nodes $\mathcal{N}=\{s_1, s_2, \cdots, s_N\},$ where $N=|\mathcal{N}|$.
Let $\mathcal{X}=\{X_1, X_2,\ldots,X_K\}$ represent a finite set of read-write objects, whose values come from a finite set $\mathcal{V}.$  We assume that $\mathcal{V}$ is a vector space over a finite field $\mathcal{F}$. The system consists of a countable set of possibly infinite client nodes $\mathcal{C}$. We consider a message passing setting where all server nodes are connected to each other through reliable, asynchronous, point-to-point first-in-first-out (FIFO) channels. The set $\mathcal{C}$ is partitioned into non-intesecting sets $\mathcal{C}_{1},\mathcal{C}_{2},\ldots,\mathcal{C}_{N}$, where $\mathcal{C}_{i}$ represents the clients associated with server node $i$ for $i=1,2,\ldots,N$. For $i \in \{1,2,\ldots,N\},$ any client in $\mathcal{C}_{i}$ sends messages only to server $i$. We assume that every client has a unique natural number as an identifier.  In the sequel, when we use the term \emph{node}, it is understood to be a  \emph{server} node; client nodes are identified explicitly.
 
Clients can perform a \textit{write} operation $\texttt{write}(X,v)$ where $X \in \mathcal{X}$ is referred to as the object of the write operation and $v \in \mathcal{V}$ as the value of the write operation. Clients also perform \textit{read} operations $\texttt{read}(X)$ where $X \in \mathcal{X}$. The response to a read operation is a value in $\mathcal{V}$. The system is an I/O automaton which is a composition of the clients, servers, and channels. Clients and servers can \textit{halt} during executions; a halted node does not take any further steps in the execution. The definitions of an execution and fairness follow the standard terminology~\cite{Lynch1996}, with tasks defined to be the set of all message send and deliver actions for non-halting nodes and external actions of non-halting clients. In all the executions, we assume the following well formedness condition: for every client $c \in \mathcal{C}$, there is at most one pending invocation at any point of an execution.

\subsection{Background on Erasure Coding} \label{chap:background}
\begin{definition}[A linear code $\mathfrak{C}(N,K,\mathcal{F})$, where $N,K$ are positive integers and $\mathcal{F}$ is a finite field] 
A linear code $\mathfrak{C}(N,K,\mathcal{F})$ is specified by $N+1$ vector spaces $\mathcal{V},\mathcal{W}_{1},\mathcal{W}_{2},\ldots,\mathcal{W}_{N}$ over field $\mathcal{F}$ and $N$ linear maps $\Phi_i,i=1,2,\ldots,N$:
$\Phi_i: \mathcal{V}^{K} \rightarrow \mathcal{W}_{i}$. 
\end{definition}
In the above definition $\Phi_{i}$ is the encoding function for the $i$th server. The definition implicitly assumes that all $K$ objects are of the same size. However, this assumption is not necessary for the results of the paper and is only made for simplicity of description/notation. The subsequent definitions in this section apply for apply for any $(N,K,\mathcal{F})$ linear code $\mathfrak{C}$ specified by $N+1$ vector spaces $\mathcal{V},\mathcal{W}_{1},\mathcal{W}_{2},\ldots,\mathcal{W}_{N}$ and encoding function $\vec{\Phi}= (\Phi_1,\Phi_2,\ldots,\Phi_N)$.

 A set $S \subseteq \{1,2,\ldots,N\}$ is said to be a recovery set for the $i$th object if the value of the $i$-th object can be recovered from the encoded data stored in the server nodes in $S$. 
\begin{definition} [Recovery sets of the $i$-th object, $1 \leq i \leq K$]
  A set $S \subset \{1,2,\ldots,N\}$ is said to be a recovery set of the $i$-object if there exists a function $\Psi_{S}^{(i)} : \displaystyle\prod_{j \in S} \mathcal{W}_j  \rightarrow \mathcal{V}$ such that: $\Psi_{S}^{(i)}(\Pi _{S}(\vec{\Phi} (\mathbf x))) = x_i$ for every $\mathbf{x}=(x_1, x_2, \ldots, x_K) \in \mathcal{V}^{K},$ where $\Pi_S$ denotes the standard projection mapping from $\displaystyle\prod_{i=1}^{N}\mathcal{W}_{i}$ to $\displaystyle\prod_{i \in S} \mathcal{W}_{i}$  \label{def:ecrecovery}
\end{definition}	

For a \emph{maximum distance separable code} (such as the well-known Reed Solomon code) with parameters $N,K$, every set of $K$ servers contains a recovery set for every object (see \cite{roth2006introduction, macwilliams1977theory}). For a node $s \in \{1,2,\ldots,N\},$ $\mathcal{X}_{s}$ denotes the set of objects that the encoding function $\Phi_s$ depends on. More formally:

\begin{definition}[Objects at server $s$ denoted by $\mathcal{X}_{s}$] $X_k \in \mathcal{X}_{s}$ for $k \in \{1,2,\ldots,K\}$ if and only if there exist $\mathbf{x}=(x_1, x_2,\ldots,x_K), \mathbf{x}'=(x'_1, x'_2,\ldots,x'_K)$ such that $x_{m}=x'_m, m\in\{1,2,\ldots,K\}-\{k\}$ and $\Phi_s(\mathbf{x}) \neq \Phi_{s}(\mathbf{x}').$
\label{def:ecobjectstored}
\end{definition}

\begin{definition}[Re-encoding function $\Gamma_{i,k}$] 
The function $\Gamma_{i,k}: \mathcal{W}_{i} \times \mathcal{V} \times \mathcal{V} \rightarrow \mathcal{W}_{i},i\in\{1,\ldots,N\},k \in \{1,2,\ldots,K\}$ is said to be a re-encoding function if it satisfies the following:

$\Gamma_{i,k}(\mathbf{\Phi}_i(\mathbf{x}), x_k,x_k') = \Gamma_{i,k}(\mathbf{\Phi}_i(\mathbf{x}),\mathbf{0},x_k'-x_k) = \Gamma_{i,k}(\mathbf{\Phi}_i(\mathbf{x}),x_k-x_k',\mathbf{0}) = \ \Phi_{i}(\mathbf{x}')$ for every pair of vectors $\mathbf{x},\mathbf{x}' \in \mathcal{V}^{K}$ such $\mathbf{x}=\mathbf{x}'$ or $\mathbf{x},\mathbf{x}'$ differ only in the $k$th co-ordinate.
\label{def:reencoding}
\end{definition}

It is a simple linear algebraic fact that, for any $N,K$ linear code, there exists re-encoding functions $\Gamma_{i,k}$ for $i=1,2,\ldots,N$ and $k=1,2,\ldots,K$. This fact will be used in our algorithms.

 \begin{example}
 Let $\mathcal{F}$ be any finite field with odd characteristic. Consider three objects $X_1,X_2,X_3$, their values $[x_1, x_2,x_3] \in \mathcal{F}^3$. The example of Sec. \ref{subsec:challenges} encodes these objects using a $(5, 3)$ code as $[x_1, x_2, x_1+x_2, x_1+x_2+x_3,x_1+2x_2+x_3].$ The sets of recovery sets are discussed in Sec. \ref{subsec:motivation}. An example of a re-encoding function is $\Gamma_{5,2}(y_5,x_2, \hat{x}_{2}) = y_{5}-2x_2+2\hat{x}_2$. 
 \end{example}

\subsection{Requirements} 

For an execution $\beta$, for any client $c \in \mathcal{C},$ let $\beta|_{c}$ represent the subsequence of $\beta$ with respect to the invocations and responses of $c$. If operation $\pi_1$ precedes $\pi_2$ in $\beta|_{c}$, we write $\pi_1 \rightarrow \pi_2$. 
\begin{definition}[Causal and eventual consistency]\label{def:original_causal}
An execution $\beta$ is called causally consistent, if there exists an irreflexive partial order $\leadsto$ among all the operations in $\beta$ and a total order\footnote{The partial order $\leadsto$ is sometimes referred to as the \emph{visibility} order, and the total order $\prec$ as the \emph{arbitration order} \cite{burckhardt2014replicated, burckhardt2014principles, attiya2017limitations}. Propeties (a) and (b) in Definition \ref{def:original_causal} are referred to respectively as \emph{Causal Visibility} and \emph{Causal Arbitration} in \cite{burckhardt2014principles,burckhardt2014replicated}. Property (c) in Definition \ref{def:original_causal} corresponds to a key-value store which is a collection of last-write-wins (LWW) registers as per the terminology of \cite{burckhardt2014principles,burckhardt2014replicated}.} $\prec$ for all write operations in $\beta$ such that
{\bf (a)}$\pi_1 \rightarrow \pi_2 \Rightarrow \pi_1 \leadsto \pi_2$, {\bf (b)} if $\pi_1,\pi_2$ are write operations such that $\pi_1 \leadsto \pi_2$, then $\pi_1 \prec \pi_2$   and  {\bf (c)} any read $\phi$ to object $X$ returns the value of the write $\pi^{*}$ - which is the largest write (as per $\prec$) to object $X$ in $S_{\phi}$ where $S_{\phi}$ is the set of write operations $\pi$ such that $\pi \leadsto \phi$. 

A causally consistent execution $\beta$ is called eventually consistent if there exists a partial order $\leadsto$ and total order $\prec$ that, in addition to satisfying the above properties, also satisfies the following: for any write operation $\pi$ in $\beta$, the number of operations in $\beta$ that do not belong in the set $S_{\pi}=\{\phi: \pi \leadsto \phi\}$ is finite.
\end{definition}

An algorithm is causally consistent if its every execution is causally consistent. Given an arbitrary linear code $\mathfrak{C}(N,K,\mathcal{F}$) with encoding functions over input alphabet $\mathcal{V}$, our goal is to design a causally consistent read-write memory emulation algorithm to store $K$ objects over a distributed asynchronous message passing system with $N$ servers. We require our algorithm to satisfy the following liveness properties:\\
\underline{(1) Eventual Consistency:} Every execution  $\beta$ is eventually consistent.\\
\underline{(2) Storage Cost:} In ``stable state'', the storage cost for every server is equal to that guaranteed by the underlying erasure code. More formally, let $\mathcal{W}_{1},\mathcal{W}_{2},\ldots,\mathcal{W}_{N}$ denote the outputs of the encoding functions associated with code $\mathfrak{C}$. For any server node $i \in \{1,2,\ldots,N\},$ let $\mathcal{S}_{i}$ denotes the set of possible states of server $i$ in the algorithm. For any positive integer $w > 0$, there is a fixed subset $\mathcal{T}_{i,w} \subset \mathcal{S}_{i}$ of cardinality $c_w |\mathcal{W}_{i}|$ for some positive constant $c_w$ that does not depend on the code $\mathfrak{C}$ such that, in any fair execution with a finite number $w$ of write operations, the state of the algorithm eventually lies in $\mathcal{T}_{i,w}$\footnote{Since it takes $\log_{2}|\mathcal{T}_{i}|$ bits to store an element in $\mathcal{T}_{i}$, this constraint readily implies that the number of bits stored is $\log_{2}|\mathcal{W}_{i}|+\log_{2}c_W$ - at most a constant number of bits more than  that prescribed by the code. The dependence of $c_w$ on the number of writes to allows the use of vector clocks.}.\\ \underline{(3) Local Writes:} A write issued at by a client in $\mathcal{C}_{s}$ terminate in any fair execution where server $s$ and the client are non-halting.\\ \underline{(4) Liveness of Reads:} Every read operation issued by a client in $\mathcal{C}_{s}$ to an object $X_{i}$ returns in any execution where the following components take infinitely many steps in a fair manner: (a) the client, (b) at least one set of servers that form a recovery set for $X_i$ as per code $\mathfrak{C}$,  (c) server $s$, (d) the channels connecting the components in (a)-(c).


%% file: algorithm.tex
We present the \CausalEC~ algorithm that is parametrized by an error correcting code $\mathcal{C}(N,K,\mathcal{F})$ to store $K$ objects $X_1, X_2,\ldots, X_{K}$ on $N$ servers. For server $s,$ we represent $\mathcal{X}_{s}$ as the set of objects to be stored by server $s$ as per code $\mathfrak{{C}}.$  

\subsubsection*{State variables}\label{state_variables}
 Each server maintains a vector clock $vc$, i.e., one natural number corresponding to every server in the system. We denote by $\mathcal{VC}=\mathbb{N}^{\mathcal{N}}$ the partially ordered set that represents the alphabet of the vector clock. Two vector clocks with (vector) values $vc_1,vc_2$ can be compared by their components: vector $vc_1$ \emph{is less than or equal to} $vc_2$ ($vc_1 \leq vc_2$) if each of $vc_1$'s components is less than or equal to $vc_2$'s corresponding component; $vc_1$ \emph{is less than} $vc_2$ ($vc_1<vc_2$) if $vc_1 \leq vc_2$ and $vc_1 \neq vc_2.$ A \emph{tag} is a tuple of the form $(ts,id)$, where timestamp $ts$ is the value of a vector clock, and $id$ is a natural number that is the unique identifier of a client. We denote by $\mathcal{T}=\mathbb{N}^{\mathcal{N}}\times\mathbb{N}$ the set which represents the alphabet of the tag. A total order among the tags can be formed as follows: tag $t_1$ \emph{is less than}  $t_2$ ($t_1<t_2$) if $vc_1<vc_2$ or $vc_1\not > vc_2, id_1<id_2$. For tag $t$, we denote by $t.ts$ its vector clock component.
 
For a set $\mathcal{A}$, the power set of $\mathcal{A}$ is denoted by $2^{\mathcal{A}}$. We use $\bot \notin \mathcal{V}\cup \mathcal{W}_{1}\cup \mathcal{W}_{2}\ldots \cup \mathcal{W}_{N}$ to represent a null value. We denote $\mathcal{V}^{\bot} = \mathcal{V}\cup \{\bot\}$ and $\mathcal{W}_{i}^{\bot} = \mathcal{W}_i\cup \{\bot\}$ for $i \in \mathcal{N}$. We assume that $\texttt{localhost} \notin \mathcal{C}$ and we denote by $\overline{\mathcal{C}} = \mathcal{C} \cup \{\texttt{localhost}\}$  For convenience, we will denote $\overline{\mathcal{W}} = \prod_{i \in \mathcal{N}}\mathcal{W}_i^{\bot}$. We represent $\Pi_i, i \in \mathcal{N}$ to be the projection operator acting on $\overline{\mathcal{W}};$ that is, $\overline{w} = (w_1, w_2, \ldots, w_N) \in \overline{\mathcal{W}},$ we have $\Pi_i(\overline{w}) = w_i.$ We assume that every operation has a unique \emph{operation identifier} that comes from a set $\mathcal{I}.$ 

The state variables at server $s$ are the following (see, Fig. \ref{alg:state}. for formal presentation and initial states):  
\begin{itemize}
  \setlength\itemsep{1pt}
\item vector clock $vc_s\in \mathcal{VC}$;

\item a priority queue $InQueue_s \subset {\mathcal{N}\times \mathcal{X} \times\mathcal{V} \times \mathcal{T}}$, i.e., an ordered list of server-object-value-tag tuples;

\item list $L_s: \mathcal{X} \rightarrow 2^{\mathcal{T} \times \mathcal{V}},$  i.e., a set of value-tag tuples, one set per object;
\item a deletion list $DelL_s: \mathcal{X}\rightarrow 2^{\mathcal{T} \times \mathcal{N}}$, i.e., a set of tag-server index tuples, one set per object;

\item a codeword symbol value and corresponding tags $M_s \in \mathcal{W}_{s} \times \mathcal{T}^{\mathcal{X}}.$  
\item a pending read list $ReadL_s \subset {\overline{\mathcal{C}} \times \mathcal{I} \times  \mathcal{X}}  \times \mathcal{T}^{\mathcal{X}} \times  \overline{\mathcal{W}}$;
\item Error indicator arrays $Error1, Error2:\mathcal{X}\rightarrow \{0,1\};$ Initially,  $Error1[X]=Error2[X]=0$ for all $X \in \mathcal{X}.$
\item An array of tags $tmax:\mathcal{X}\rightarrow \mathcal{T};$ initially $tmax[X]=\vec{0}$ for all $X \in \mathcal{X}.$ 
\end{itemize}
The server states of \CausalEC~ is shown in Fig. \ref{alg:state}.  In Fig. \ref{alg:state}, we denote channel from server $i$ to server $j$ represented as $c_{i \rightarrow j}, i \neq j.$ Each channel may be thought of as a queue of messages, where each message belongs to a finite set $\mathcal{M}$. Server $i$ can invoke a $\texttt{send}_{i\rightarrow j}(m), m \in \mathcal{M}$ input to channel $c_{i \rightarrow j}$  and the channel output corresponds to a $\texttt{receive}_{i \rightarrow j}(m)$ input to server $j$.

The server protocols associated with \CausalEC~ are described in Algorithms \ref{alg:inputactions_clients}, \ref{alg:input_actions} and \ref{alg:internal_actions}.

\begin{figure}[h] 
\begin{framed}
\footnotesize
\textbf{server $s$:}

\underline{state and initial values:}

$vc_s \in \mathbb{N}^{N}$, initial state $vc=\vec{0}$

$InQueue_s \subset {\mathcal{N}\times \mathcal{X} \times\mathcal{V} \times \mathcal{T}}$, initially empty

Lists $L_s: \mathcal{X} \rightarrow 2^{\mathcal{T} \times \mathcal{V}}$,   initially $L[X]=<(\vec{0},\vec{0})>$ for every $X\in \mathcal{X}$

$DelL_s: \mathcal{X} \rightarrow 2^{\mathcal{T} \times \mathcal{N}}$ initially empty

$M_s \in \mathcal{W}_s \times \mathcal{T}^{\mathcal{X}}$, initial state $M=(\vec{0}, \vec{0}^{\mathcal{X}})$

$ReadL_s \subset {\mathcal{C} \times \mathcal{I} \times \mathcal{X} \times \mathcal{T}^{\mathcal{X}} \times \overline{\mathcal{W}}}$,  initially empty

$Error1, Error2: \mathcal{X} \rightarrow \{0,1\}$ Initially set to $Error1[X]=Error2[X]=0$ for all $X \in \mathcal{X}.$

$t{max}:\mathcal{X} \rightarrow \mathcal{T},$ initially $t{max}[X] = \vec{0}$ for all $X \in \mathcal{X}.$

\underline{actions:}

Input: 
 $\receive(m), m \in \mathcal{M}$ along channels $c_{i\rightarrow s}$ for $i=\{1,2,\ldots,N\}-\{s\}$, and from clients in $\mathcal{C}_{s}$

Output: send($m$) for $m \in \mathcal{M}$ along channels $c_{s\rightarrow i}$ for $i=\{1,2,\ldots,N\}-\{s\}$, and to clients in $\mathcal{C}_{s}$

Internal: 
$\texttt{Apply}\_\texttt{Inqueue},$ $\texttt{Garbage}\_\texttt{Collection},$ $\texttt{Encode}$

\end{framed}
\caption{Server states for server $s$ of \CausalEC}
\label{alg:state}
\end{figure}

We suppress the subscript when the identity of the server can be obtained from context, e.g., we denote the vector clock as $vc.$ The data structure $M$ stores the codeword symbol value and a vector of tags. At server $s$, $M.val$ denotes the value from $\mathcal{W}_s^{\bot}$ in tuple $M$.   For object $X \in \mathcal{X}_{i}$,  $M.tagvec[X]$ denotes the tag corresponding to object $X$. $InQueue$ is ordered by timestamp component of tags, with items with smaller timestamps appearing closer to the head. A new tuple added to the list is placed after all existing items with a smaller or incomparable tag. The head of the queue is denoted as $InQueue.Head$ with ties amongst incomparable timestamps broken arbitrarily. 

The list $L$ at nodes serve as a place to store \emph{history}, i.e, multiple versions of one object. For an object $X$, we denote the tuple with the highest tag in $L[X]$ as $L[X].HighestTagged.$ If tuple $(t,v)$ is $L[X].HighestTagged$, then we denote by  $L[X].HighestTagged.tag=t$, $L[X].HighestTagged.val=v,$ and  $L[X].HighestTagged.tag.ts = t.ts$.  If the list is empty, then we simply use the convention $L[X].Highesttagged.tag = \vec{0}$ and  $L[X].Highesttagged.val = \vec{0}$.
Unnecessary values in $L$ will be deleted by checking the tuples in $DelL$. The data structure $ReadL$ is used to store the pending read operations. For an entry $(clientid,opid,X,tagvector,\overline{w})$ in $ReadL,$ if the non-null entries in $\overline{w}$ form a recovery set for object $X,$ then the value is recovered and sent to the client $clientid,$ and the entry is removed from $ReadL.$


\subsubsection*{Algorithm description: Client protocol}
\label{server_protocol}

The protocol of the clients are described here: since the protocol is simple, we omit its formal description. For each operation, the client generates a unique operation identifier $opid$ from set $\mathcal{I}$. On receiving an invocation for write operation on object $X$ with value $v$, the client from $\mathcal{C}_{s}$ simply sends a $\langle\texttt{write}, opid, X, v\rangle$ message to server $s.$ The client waits for an $\langle \texttt{write-return}, opid, \texttt{ack} \rangle$ = message from server $s$ and terminates the operation on receiving the message. On receiving an invocation to a read operation to object $X$, the client simply sends a $\langle\texttt{read}, clientid, opid, X \rangle$ message to server $s$. On receiving a $\langle \texttt{read-return}, opid, val\rangle$ message for some $val \in \mathcal{V}$, it returns $val$.

\subsubsection*{Algorithm description: Server protocol (High level Description)}

\begin{figure}[!tbh]
\centering
\includegraphics[width=4.5in]{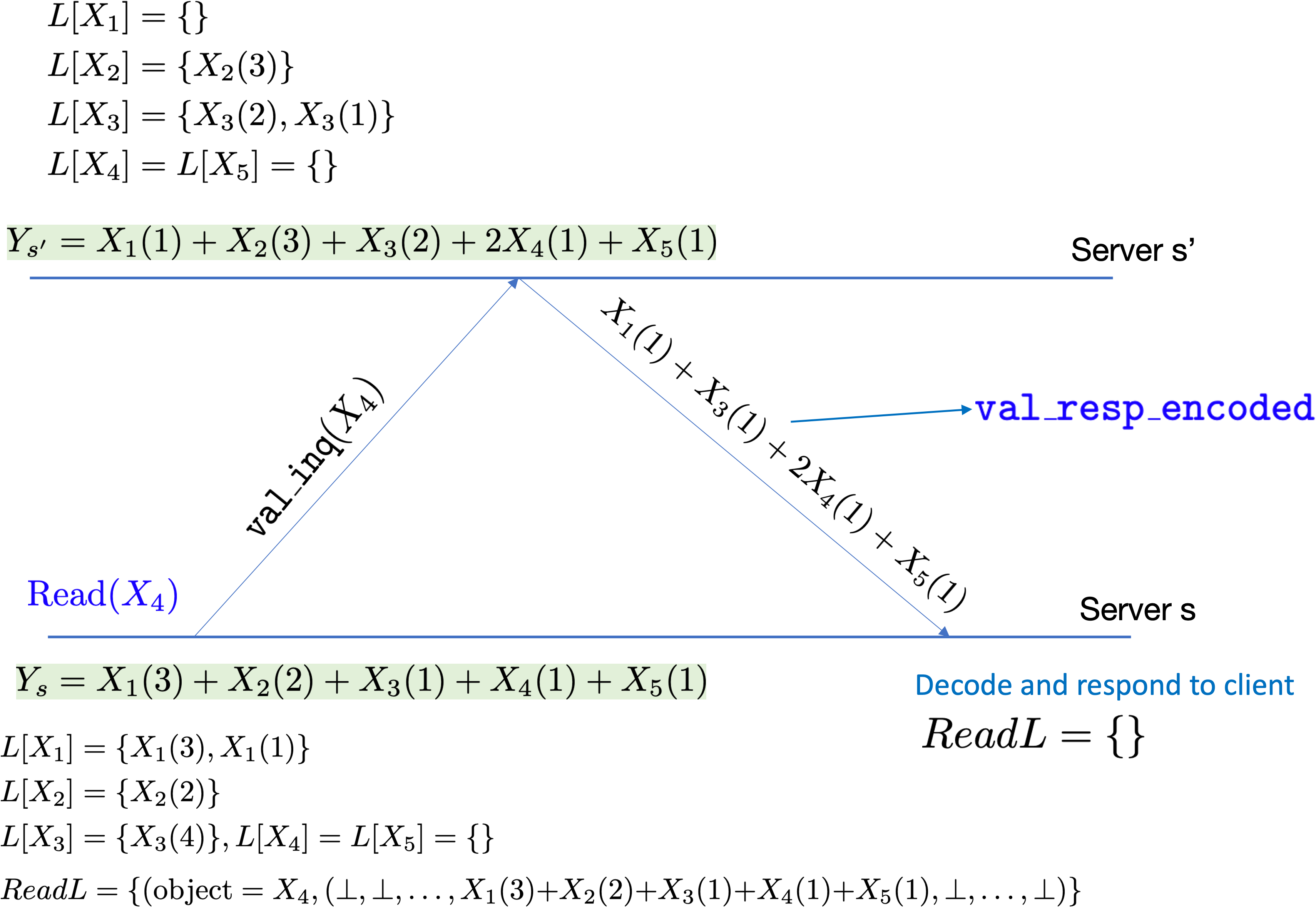}
\caption{\small{An example of a read operation on object $X_4$ at node $s$ for the depicted erasure code; clearly $\{s,s'\}$ forms a recovery set for object $X_4$. In the figure, we use integer indices rather than tags for simplicity of description. On receiving the read request, node $s$ registers the read in $\texttt{ReadL}$ and sends a $\texttt{val\_inq}$ message. Node $s'$ is able to use its local history to re-encode its stored variable to $X_1(1)+X_3(1)+2X_4(1)+X_5(1)$; the effect of the mismatched version $X_2$ is cancelled, and $X_3$ is re-encoded to the desired version. Node $s$ on receiving $\texttt{val\_resp\_encoded}$ message uses $X_1(3),X_1(1),X_2(2)$ stored locally to obtain $X_1(3)+X_2(2)+X_3(1)+2X_4(1)+X_5(1).$ Then, recognizing that $\{s,s'\}$ can serve the read on object $X_4,$ it decodes $X_4(1)$, deletes the corresponding entry in $\texttt{ReadL}$ and responds to the read. Caveat: the notation in this figure has been simplified in comparison with the formal algorithm specification of Algorithms \ref{alg:input_actions} \ref{alg:inputactions_clients} \ref{alg:internal_actions} to enable easier understanding of the underlying ideas.}}
\label{fig:respond}
\vspace{-5pt}
\end{figure}

We describe next the  protocol of server $s$. We first provide a brief, high-level description of the ideas. In the protocol, all the nodes maintain vector timestamps to identify versions and maintain causal consistency in a relatively standard manner. Our description here focuses on aspects related to erasure coding, ensuring liveness, and garbage collection. On receiving a write operation to object $X$, the server stores the object value along with the timestamp in the history list, sends it to every other server and responds to the client. 
On receiving a read operation to an object $X$, server $s$ responds to the client immediately with a value if its local history list $L_s[X]$ has a causally consistent version of the object. Otherwise, it logs the read along with the local codeword value $M_s$ in a pending read list ($ReadL_s$) and sends $\texttt{val}\_\texttt{inq}$ inquiry messages to all the other servers (See Fig. \ref{fig:respond}). The $\texttt{val}\_\texttt{inq}$ message consists the desired timestamps for every object that would enable decoding at server $s$. Every server that receives a $\texttt{val}\_\texttt{inq}$  immediately responds to the message with a $\texttt{val}\_\texttt{resp}$ or a $\texttt{val}\_\texttt{resp}\_\texttt{encoded}$ message. A $\texttt{val}\_\texttt{resp}$ message is sent by the server if it has a causally consistent version of object $X$ in its list. Otherwise, a codeword symbol is sent via a $\texttt{val}\_\texttt{resp}\_\texttt{encoded}$  message.
Server $s$ returns a value to the read on receiving a $\texttt{val}\_\texttt{resp}$ message, or on receiving  $\texttt{val}\_\texttt{resp}\_\texttt{encoded}$   responses from at least one recovery set for object $X$. In the latter case, the server performs a decoding operation on the codeword symbols corresponding to that recovery set.

We explain server actions on receiving a $\texttt{val}\_\texttt{resp}$ or a $\texttt{val}\_\texttt{resp}\_\texttt{encoded}$  through the example scenario depicted in Fig. \ref{fig:respond}. In our explanations, $X_j(i)$ denotes the value of the $i$th write to object $X_j$. As in Sec. \ref{subsec:challenges}, we denote every version of an object by an integer index, whereas the actual protocol implements an indexing of writes via vector timestamps. On receiving the $\texttt{val}\_\texttt{inq}$ message from server $s$, server $s'$ compares the requested tags for each object with the tags of the locally stored codeword symbol. For every object $X$ that the tags do not match, the attempts to perform re-encoding using the information in the corresponding list $L[X].$ In the scenario of Fig. \ref{fig:respond}, server $s'$ attempts to perform re-encoding for objects $X_1,X_2,X_3$.  Re-encoding can be understood as a sequence of two steps. The first step is to cancel the effect of the locally stored version (if possible). If the first step is successful, the second step is to apply the effect of the requested object version.
Because server $s'$  has $X_2(3),X_3(2)$ in the local history list, it first cancels the effect of $X_2(3),X_3(2)$ to obtain $X_1(1)+2X_4(1)+X_5(1) $. The second step is only performed for $X_3$ since the desired object version $X_3(2)$ is present in the local history list of server $s'$, but $X_2(2)$ is not.

On receiving the $\texttt{val}\_\texttt{resp}\_\texttt{encoded}$  response, server $s$ further uses the re-encoding function using the received codeword symbol and the local history lists to ensure that the final stored codeword is an encoding of the desired object versions (See Fig. \ref{fig:respond}). If the entry in $ReadL$ contains sufficient information to decode the desired object (that is, the server receives responses from a recovery set), the object is decoded and returned to the client, and the entry is removed from the pending read list.  The conditions for garbage collection of \CausalEC~ are designed to ensure that the history lists at servers $s,s'$ have sufficient information to enable the desired re-encoding. That is, for every read operation at server $s$, it can, through re-encoding of the codeword symbol contained in the $\texttt{val}\_\texttt{resp}\_\texttt{encoded}$ message server $s'$, obtain a codeword symbol corresponding to an encoding of the versions that match the local codeword version.

The server protocol has three internal actions that we describe next.
(I) \texttt{Apply}\_\texttt{Inqueue}: This action is similar to standard causally consistent protocols, which expose object values to readers based on timestamps in a causally consistent manner. The main difference for us is that, when a version is ready to be exposed, it is added to the history list. 

(II) \texttt{Encoding}: For any object, if a list consists of a later version than the one applied to the stored codeword symbol, then, the encoding action triggers a "local read". For instance, in Fig. \ref{fig:respond}, server $s$ triggers a local read for $X_3(1)$. On receiving this object version, the server uses re-encoding function to update the locally stored codeword symbol: $Y_s-X_3(1)+X_3(4).$

On performing the re-encoding, a delete message from server $s$ is sent to every server node for that object version. On receiving the delete message, server $s'$ server logs the message in the deletion list along with the server index $s$ of the sender.

(III) \texttt{Garbage Collection}: The garbage collection action removes elements from history lists in a manner that ensures that liveness of reads is not affected. Consider a point of the execution where server $s$ stores a codeword symbol that is an encoding of object version $X_j(k).$ Let $k^*$ denote the largest integer (version index) such that the deletion list at server $s$ has logged a delete message from every server for object $X_j$ with a version at least as large as $k^*.$ From the encoding action, we can infer that every server stores an encoding of a version of $X_j$ that is equal to, or later than $X_j[k*].$
An object version $X_j(i)$ in the list $L[X_j]$ at server $s$ is considered for deletion only if $k \geq i$ and $k^* \geq i$. If $k > i$ and $k^* > i,$ then the object can be deleted from the history list $L[X_j]$ so long as there is no pending read that requires $X_j(i)$ to be decoded. Otherwise, version $X_j(i)$ is deleted only if $k=k^*=i$ and server $s$ has logged in its deletion list, a deletion message for object $X_j$ with a version index for exactly equal to $i$. These careful conditions avoid unnecessary deletion of data objects, for instance, if $k=i > k^{*}$ or $k > i = k^*,$  object version $X_j(i)$ may be present in the history list to enable a local re-encoding operation to a higher object version, or to help other read operations that may not have yet even received $X_j(i).$  The formal description next contains some improvements over the above, informal description. For instance, it  avoids certain re-encoding actions for objects that are not stored on servers.

\subsubsection*{Algorithm description: Server protocol (Detailed Description)}
 The formal specification of actions and corresponding state transitions are described in Figures \ref{alg:input_actions}\ref{alg:inputactions_clients} and \ref{alg:internal_actions}; here we provide a textual description. The server protocol description is categorized based on the type of action - reception of messages from clients, reception of messages from other servers, and internal actions. 
In our description, we use function $ObjectIndex:\mathcal{X}\rightarrow \{1,2,\ldots,K\}$ with $ObjectIndex(X_i) = i.$ We also denote $\Pi:\mathcal{N}\rightarrow 2^{\mathcal{N}}$ to be the standard projection function.

\begin{figure}[htb] 
\begin{algorithm*}[H]
\caption{Server protocol in \CausalEC: Transitions for Clients Messages at node $s$} 
\label{alg:inputactions_clients}

\footnotesize
 \underline{On  {$receive\langle\texttt{write}, opid, X,v\rangle$}  from client $id$}:

$vc[s]\leftarrow vc[s]+1$
\label{line:vcincrement}

$t\leftarrow (vc,id)$

$L[X] \leftarrow L[X] \cup \{(t,v)\}$
\label{line:writeaddtolist}

$send\langle\texttt{write-return-ack}, opid\rangle$ to client $id$
\label{line:writeresponse}

$send\langle\texttt{app},X,v,t\rangle$ to all other nodes, $j\neq i$ 
\label{line:writesendapply}

for every $opid$ such that there exists $clientid \neq \texttt{localhost}, \bar{v},tags$ such that $(clientid, opid, X, tags, \bar{v})\in ReadL$:


\hspace{1cm}$send\langle \texttt{read-return}, opid,v \rangle$ to client $clientid$
\label{line:sendtoread1}

\hspace{1cm}$ReadL\leftarrow ReadL - \{clientid, opid, X, tvec, \bar{v})\}$
\label{line:read-remove1}

\underline{On $receive\langle\texttt{read},opid, X\rangle$ from client $id$}:

if $L[X] \neq \{\}$ and $L[X].Highesttagged.tag \geq M.tagvec[X]$: \label{line:readimmediatereturncheck1}

\hspace{1cm}$send\langle \texttt{read-return}, opid, L[X].HighestTagged.val \rangle$ to client $id$    \hspace{10pt} \# return locally
\label{line:readimmediatereturn}

else if  $\{s\} \in \mathcal{R}_X$ \label{line:readreturnimmediatecheck2}

\hspace{1.5cm}$v\leftarrow \Psi_{\{s\}}^{ObjectIndex(X)}( M.val)$

\hspace{1.5cm} send $\langle \texttt{read-return}, opid, v \rangle$ to client $id$ 

\label{line:send-to-read1}
 
else:  \hspace{10pt} \# contact other nodes

\hspace{1cm}$ReadL \leftarrow ReadL\cup (id, opid, X, M.tagvec ,w_1, w_2, \ldots, w_N)$, where $w_{i} = \begin{cases} M.val & \textrm{if }i = s\\ \bot & \textrm{otherwise} \end{cases}$ 
\label{line:addtoreadl}

\hspace{1cm}send$\langle\texttt{val\_inq},id, opid, X, M.tagvec\rangle$ to every node $j$ where $j\neq i$
\label{line:valinq_send}
\end{algorithm*}

\end{figure}

\begin{figure}[htp!] 
\begin{algorithm}[H]
\footnotesize

\underline{On $receive\langle\texttt{del},X, t\rangle$ from node $j$}:
 
 $DelL[X] \leftarrow DelL[X] \cup \{(t,j)\}$ 
\label{line:delmessage}

\underline{On $receive\langle\texttt{val\_inq},clientid, opid,\overline{X}, wantedtagvec\rangle$ from node $j$}:
 
if there exists $v$ such that $(wantedtagvec[\overline{X}],v) \in L[\bar{X}]$:

\hspace{0.5cm}$send\langle\texttt{val\_resp},clientid, opid \overline{X},v,wantedtagvec\rangle$ to node $j$

else

\hspace{0.5cm} $ResponsetoValInq \leftarrow M$
\label{line:Responsetovalinqinit}

\hspace{0.5 cm} For all $X \in \mathcal{X}_{s}$
\label{line:forloop1}

\hspace{1 cm} If $M.tagvec[X] \neq wantedtagvec[X]$: 

\hspace{1.5 cm} If there exists unique $v$ such that $(M.tagvec[X],v) \in L[X]$
\label{line:responsetovalinqcondition1}

 \hspace{2 cm}$ResponsetoValInq.val \leftarrow \Gamma_{s,ObjectIndex(X)}
 (ResponsetoValInq.val,v,\mathbf{0})$
\label{line:ResponsetoValInqvalupdate1} 

  \hspace{2 cm}$ResponsetoValInq.tagvec[X] \leftarrow \mathbf{0}$.
\label{line:ResponsetoValInqtagupdate1} 

\hspace{2 cm}If there exists unique $w$ such that $(wantedtagvec[X],w) \in L[X]$
\label{line:responsetovalinqcondition2}

\hspace{2.5 cm}$ResponsetoValInq.val \leftarrow \Gamma_{s,ObjectIndex(X)} (ResponsetoValInq.val,\mathbf{0},w)$
\label{line:ResponsetoValInqvalupdate2} 
  \hspace{2.5 cm}$ResponsetoValInq.tagvec[X] \leftarrow wantedtagvec[X]$.
\label{line:ResponsetoValInqtagupdate2} 
\hspace{0.5cm} $send \langle \texttt{val\_resp\_encoded}, ResponsetoValInq,clientid,opid, \overline{X}, wantedtagvec \rangle$ to node $j$.

\underline{On $receive\langle \texttt{val\_resp\_encoded},\overline{M},clientid,opid, \overline{X}, requestedtags \rangle$ from node $j$}:

For $X \in \mathcal{X}$

\hspace{0.5 cm}
$Error1[X],Error2[X] \leftarrow 0,$

$Modified\_codeword \leftarrow \overline{M}.val$

\label{line:modifiedcodeword_init}

If there exists a tuple $(clientid,opid, \overline{X},requestedtags,\overline{w}) \in ReadL$ for some $\bar{w} \in \overline{\mathcal{W}}$ \label{line:valrespencoded_readcheck}

For all $X \in \mathcal{X}_{j}$,   \label{line:valrespencoded_objectforloop}

\hspace{0.5 cm} If $requestedtags[X] \neq \overline{M}.tagvec[X]$
\label{line:metaerrorcondition}

\hspace{1.0 cm} If $\overline{M}.tagvec[X] \neq \mathbf{0}$ and there exists unique $(\overline{M}.tagvec[X],w)\in L[X]$
\label{line:metaerrorcondition1}

\hspace{1.5 cm}   $Modified\_codeword \leftarrow \Gamma_{s,ObjectIndex(X)}
 (Modified\_codeword,w, \mathbf{0})$
 \label{line:Modifiedcodeword1}

\hspace{1.0 cm} else if $\overline{M}.tagvec[X] \neq 0$:  $Error1[X] \leftarrow 1.$
\label{line:seterror1}

 \hspace{1.0 cm} If $Error1[X] \neq 1$ and there exists unique $v$ such that $(requestedtags[X],v)\in L[X]$
 \label{line:metaerrorcondition2}

  \hspace{1.5 cm}  $Modified\_codeword \leftarrow \Gamma_{s,ObjectIndex(X)}
 (Modified\_codeword,\mathbf{0},v)$
 \label{line:Modifiedcodeword2}

 \hspace{1.0 cm} else: $Error2[X] \leftarrow 1$\label{line:seterror2}

    If $\bigwedge_{X \in \mathcal{X}} \bigg(Error1[X] \wedge Error2[X]=0\bigg)$
  \label{line:errorcondition}  
  
\hspace{0.5cm}  $ReadL\leftarrow ReadL-\{(clientid,opid,X,requestedtags,\overline{w})\} \cup \{(clientid,opid,X,requestedtags,\overline{v})\}$, where $\overline{v} \in \overline{W}$ is generated so that $\Pi_{i}(\overline{v}) \leftarrow  \begin{cases} \Pi_i(\overline{w}) & \textrm{if  }i \neq j \\ Modified\_codeword & \textrm{if } i = j \end{cases}$ 
\label{line:valrespencoded_addtuple}

\hspace{1cm}$S\leftarrow \{i  ~|~\Pi_i(\overline{v}) \neq \bot\}$

\hspace{1cm}$\ell \leftarrow ObjectIndex[X]$

\hspace{1cm}if there exists $T \in \mathcal{R}_\ell$ such that $T \subseteq S$:
\label{line:readreturncheck}

\hspace{1.5cm}$v\leftarrow \Psi_{T}^{(\ell)} \left(\Pi_{T}(\overline{v})\right)$
\label{line:valresp_handle_decoding}

\hspace{1.5cm}if $clientid =\texttt{localhost}$:

\hspace{2cm}$L[X]\leftarrow L[X] \cup \{( M.tagvec[X], v)\}$
\label{line:readaddtolist}

\hspace{1.5cm}else: $send\langle \texttt{read-return}, opid, v \rangle $ to client $clientid$
\label{line:read-send-val_resp_encoded}

\hspace{1.5cm}$ReadL\leftarrow ReadL- \{(clientid,opid, {X}, requestedtags, \bar{v}): \exists \bar{v}  \textrm{ s. t. }(clientid,opid,{X}, requestedtags, \bar{v}) \in ReadL \}$
\label{line:read-remove2}

\underline{On $receive\langle \texttt{val\_resp},X,v,clientid, opid, requestedtags\rangle$ from node $j$}:

if there exists $\overline{v}$ such that $(clientid, opid, X, requestedtags,\overline{v})\in ReadL$: 
\label{line:valresp_handle_condition}

\hspace{0.5 cm}if $clientid = \texttt{localhost}$ 

\hspace{1cm}$L[X] \leftarrow L[X] \cup  (requestedtags[X],v)$
\label{line:valrespaddtolist}

\hspace{0.5cm}else: $send \langle \texttt{read-return}, opid, v \rangle$ to client $client-id$ 
\label{line:valresp_handle}

\hspace{0.5 cm}$ReadL\leftarrow ReadL - \{ (clientid, opid, X, requestedtags, \overline{v})\}$
\label{line:read-remove3}

\underline{On receipt of \emph{app$(X,v, t)$} from node $j$}:

$InQueue \leftarrow InQueue \cup \{(j,X,v,t)\}$
\caption{Server protocol in \CausalEC: Transitions for input actions at node $s$}
\label{alg:input_actions}
\end{algorithm}

\end{figure}

\begin{figure}[htp!] 

\begin{algorithm}[H]
\footnotesize

$\texttt{Apply}\_\texttt{InQueue}$: precondition: $InQueue \neq \{\},$ 

effect:

$(j,X,v,t) \leftarrow InQueue.Head$

If $t.ts[p]\leq vc[p]$ for all $p\neq j$, and $t.ts[j]=vc[j]+1$:
\label{line:applycondition}

\hspace{1cm}$InQueue \leftarrow InQueue - \{ (j,X,v,t)\}$ 
\label{line:inqueue_remove}

\hspace{1cm}$vc[j]\leftarrow t.ts[j]$
\label{line:vcincrement_apply}

\hspace{1cm} $L[X] \leftarrow L[X] \cup (t,v)$ 
\label{line:list_applyaction}

\hspace{1cm}for every $opid$ such that $(clientid,opid, X, tvec, \bar{v})\in ReadL, clientid \neq \texttt{localhost}, tvec[X] \leq t$: 
\label{line:readreturn_apply_condition}

\hspace{1.5cm}$send \langle \texttt{read-return}, opid, v \rangle$ to client with id $clientid$
\label{line:readreturn_apply}

\hspace{1.5cm}$ReadL\leftarrow ReadL - \{ (clientid, opid, X, tvec, \bar{v})\}$
\label{line:read-remove4}

\hspace{1cm}for every $opid$ such that $(clientid,opid, X, tvec, \bar{v})\in ReadL, clientid = \texttt{localhost}, tvec[X] = t$: 
\label{line:readreturn_apply_condition2}

\hspace{1.5cm}$ReadL\leftarrow ReadL - \{ (clientid, opid, X, tvec, \bar{v})\}$
\label{line:read-remove5}

$\texttt{Encoding}$: precondition: none

effect:

For every $X \in \mathcal{X}_{s}$ such that $L[X] \neq \{\}$ and $L[X].Highesttagged.tag > M.tagvec[X]$
\label{line:Forloopupdatetag1}

\hspace{0.5 cm} If there exists a tuple $(M.tagvec[X],val) \in L[X]$
\label{line:conditionforupdate}

\hspace{1cm} Let $v$ be a value such that $(L[X].Highesttagged, v) \in L[X].$

\hspace{1cm} $M.val \leftarrow \Gamma_{s,k}(M.val, val,v)$,~ $M.tagvec[X] \leftarrow  L[X].Highesttagged.tag$
\label{line:updatetag1}

\hspace{1 cm} $R \leftarrow \{i \in \mathcal{N}: X \in \mathcal{X}_{i}\}$

\hspace{1 cm}send $\langle \texttt{del},X,L[X].Highesttagged.tag \rangle$ to all nodes in $R$
\label{line:Delete_send1}

\hspace{1 cm} $DelL[X] \leftarrow DelL[X] \cup \{L[X].Highesttagged.tag, s)\}$
\label{line:Delete_local1}

\hspace{0.5 cm} else if there exists no tuple $(\texttt{localhost}, \overline{opid}, X, tvec, \overline{w})$ in $ReadL$ with $tvec[X]=M.tagvec[X]$

\hspace{1 cm} Generate a unique operation identifier $opid.$

\hspace{1 cm} $ReadL \leftarrow ReadL \cup \{ (\texttt{localhost}, opid, X, M.tagvec, w_1, w_2, \ldots, w_N, ) \}$, where $w_i = \begin{cases} M.val & \textrm{if }i=s \\ \bot & \textrm{otherwise}\end{cases}$
\label{line:readlentryencoding}

\hspace{1 cm}$send \langle \texttt{val\_inq},\texttt{localhost}, opid, X, M.tagvec\rangle$ to every node $j\neq s$
\label{line:valinq_send2}

For every $X \notin \mathcal{X}_{s}$ such that $L[X] \neq \{\}$ and $L[X].Highesttagged.tag > M.tagvec[X]$
\label{line:Forloopupdatetag2}

\hspace{0.5 cm} $R \leftarrow \{i \in \mathcal{N}: X \in \mathcal{X}_{i}\}$
\label{line:Delete2_Rset}

\hspace{0.5cm}  Let $U$ be the set of all tags $t$ such that $ R \subseteq \{i \in \mathcal{N}: \exists \hat{t}, (\hat{t},i) \in DelL[X], \hat{t} \geq t\}$
\label{line:Delete2_Uset}

\hspace{0.5cm}  Let $\overline{U} = \{t \in \mathcal{T}: \exists  val, (t,val) \in L[X], t > M.tagvec[X]\}$
\label{line:Delete2_Ubarset}

\hspace{0.5 cm} {If $U \cap \overline{U} \neq \{\}$}
\label{line:Delete2_condition}

\hspace{1 cm} $M.tagvec[X] \leftarrow \max(U \cap \overline{U})$, $DelL[X] \leftarrow DelL[X] \cup \{(\max(U\cap \overline{U}), s)\}$
\label{line:updatetag2}
\label{line:Delete_local2}

\hspace{1 cm} send $\langle \texttt{del},X,\max(U\cap \overline{U}) \rangle$ to all nodes
\label{line:Delete_send2}

$\texttt{Garbage}\_\texttt{Collection}$: precondition: None

effect: 

For $X \in \mathcal{X}$

\hspace{0.5cm} {Let $S$ be the set of all tags $t$ such that $\{i: \exists \hat{t}, (\hat{t},i) \in DelL[X], \hat{t} \geq t\} = \mathcal{N}$}
\label{line:sset}

\hspace{.5 cm} $t{max}[X] \leftarrow \max (S)$
\label{line:tmax}

\hspace{0.5 cm} Let $\overline{S}$ be the set of all tags $t$ such that $\{i: \exists (t,i) \in DelL[X]\} = \mathcal{N}$
\label{line:s1set}

\hspace{0.5cm} $T \leftarrow \{tvec[X]: \exists clientid, opid, i, \overline{w},\bar{X} \textrm{ s.t. } (clientid, opid, \bar{X}, i, tvec, \overline{w}) \in ReadL { \textrm{ and } tvec[X] < M.tagvec[X]}\}$ 
\label{line:pendingreads}

\hspace{0.5cm} If $ \big(t{max}[X]=M.tagvec[X]\big) \textrm{AND~} \big(M.tagvec[X] \in \overline{S}) \big) {\textrm{ AND~ } \big(L[X] = \{\} \textrm{ OR } L[X].Highesttagged.tag \leq M.tagvec[X]\big)}$
\label{line:GCcondition1}

\hspace{1 cm}$L[X] \leftarrow L[X] - \{(tag,v): \exists (tag,v) \in L[X] \textrm{ s.t. } tag \leq t{max}[X], tag \notin T\}$
\label{line:GC1}

\hspace{0.5cm} else if $\big(t{max}[X] < M.tagvec[X]\big) \textrm{AND~} \big(X \notin \mathcal{X}_{s}\big)$
\label{line:GCcondition3}

\hspace{1 cm}$L[X] \leftarrow L[X] - \{(tag,v): \exists (tag,v) \in L[X] \textrm{ s.t. } tag \leq t{max}[X], tag \notin T\}$
\label{line:GC3}

\hspace{0.5cm} else: $L[X] \leftarrow L[X] - \{(tag,v): \exists (tag,v) \in L[X] \textrm{ s.t. } tag  < t{max}[X], tag \notin T\}$
\label{line:GC2}

\hspace{0.5 cm} $R \leftarrow \{i \in \mathcal{N}: X \in \mathcal{X}_{i}\}$
\label{line:Delete3_Rset}

\hspace{0.5 cm}  Let $U$ be the set of all tags $t$ such that $ R \subseteq \{i \in \mathcal{N}: \exists \hat{t}, (\hat{t},i) \in DelL[X], \hat{t} \geq t\}$
\label{line:Delete3_Uset}

\hspace{0.5 cm} If $U \neq \{\}$ and $X \in \mathcal{X}_{s}$
\label{line:Delete_cond3}

\hspace{1 cm}send $\langle \texttt{del},X,max(U) \rangle$ to all nodes
\label{line:Delete_send3}

\caption{Server protocol in \CausalEC: Transitions for internal actions at node $s$}
\label{alg:internal_actions}

\end{algorithm}

\end{figure}
\large{\textbf{Message receipt from clients at server $s$.}} See Algorithm \ref{alg:inputactions_clients}.

\underline{On receipt of a write command} from client with identifier $clientid$ to object $X$ with value $v$, it increments the component of vector clock $vc$ corresponding to node $s$ and adds $((vc,clientid),v)$ to the list $L[X]$. The node then sends an acknowledgement to the client, and sends $\langle\texttt{app},X,v,t\rangle$ message to all the other nodes.
Furthermore, the node clears any pending non-local reads to object $X$  (line \ref{line:sendtoread1} in Algorithm \ref{alg:inputactions_clients} and clears the corresponding entry from ReadL (line \ref{line:read-remove1} in Algorithm \ref{alg:inputactions_clients}).

\underline{On receipt of a read for object $X$} with operation identifier $opid$  from client with identifier $clientid$: If the list $L[X]$ is non-empty and $L[X].Highesttagged.tag \geq M.tagvec[X]$, the highest tagged value in the list is returned to the client. If $\{s\}$ is a recovery set for object $X$, then the value is decoded from $M$ and returned to the client (checked in line \ref{line:readimmediatereturn} in Algorithm \ref{alg:inputactions_clients}). Otherwise, an entry \newline $\vec{e} = (clientid, opid, X, M.tagvec, w_1, w_2, \ldots, w_N)$ is made into the read list $ReadL$ for the read, with $w_{s} = M.val$ and $w_j = \bot, j \neq s$. A $\langle \texttt{val\_inq}, clientid, opid, X, M.tagvec\rangle $ message is sent to all\footnote{An optimization is to first send the message to the nearest recovery set for $X$, followed by a broadcast in case of a time-out. We assume such an optimization in Sec. \ref{sec:communication}.} other nodes. The goal of the $\texttt{val\_inq}$ message to server $j \neq s$ is to obtain response to update the entry $\vec{e}$ in $ReadL$ to set $w_{j}$ to a non-null codeword symbol value. After codeword symbols received from a set of nodes that forms a recovery set for $X$, the value corresponding to object $X$ can be decoded.


{\textbf{Message receipt from other servers at server $s$}}
See Algorithm \ref{alg:input_actions}.

\underline{On receipt of a $\texttt{val\_inq}$ message of the form $\langle \texttt{val\_inq},opid, clientid,X,wantedtagvec\rangle$}: See Fig. \ref{fig:respond}. 
$wantedtagvec$ denotes the desired object versions required for decoding the read that sent the message. First, the node checks if $L[X]$ has a value corresponding to the requested tag $wantedtagvec[X]$ for that object. If so, the node sends a $\texttt{val\_resp}$ message with the value of that object. The $\texttt{val\_resp}$ message has the desired object $(X)$ in uncoded form and the node that sent the $\texttt{val\_inq}$ message can use it to respond to the read operation without any special decoding. Otherwise, the node aims to send $M$, possibly with some re-encoding in a $\texttt{val\_resp\_encoded}$ message. In particular, it aims to re-encode $M$ to $wantedtagvec$ if possible. The updated value-tag tuple that is sent in this message is denoted as $ResponsetoValInq$. The $ResponsetoValInq$ variable has the same data type as $M$, that is, it belongs to $\mathcal{W}_{s}\times \mathcal{T}^{\mathcal{X}}.$ The $ResponsetoValInq$ variable is initially populated with $M$ and then the re-encoding (if any) is performed by a for loop cycling through every object $\overline{X} \in \mathcal{X}_{s}$ (line \ref{line:forloop1} in Algorithm \ref{alg:input_actions}.). In  a given iteration of this for loop, with object $\overline{X} \in \mathcal{X}_{s}$ if  $M.tagvec[\overline{X}] = wantedtagvec[\overline{X}],$ then no re-encoding is performed on behalf of $\overline{X}$ - an example is objects $X_4,X_5$ in Fig. \ref{fig:respond} at node $s'$. On the other hand, if $M.tagvec[\overline{X}] \neq wantedtagvec[\overline{X}],$ we consider three cases:
(i) First, the node checks to see  the node checks if $L[\overline{X}]$ contains values corresponding to both $M.tagvec[\overline{X}]$ and $wantedtagvec[\overline{X}].$ If these values exist in $L[\overline{X}],$ then the re-encode function is used to ensure $ResponsetoValInq$ stores an encoded value corresponding to $wantedtagvec[\overline{X}].$ - an example is objects $X_3$ in Fig. \ref{fig:respond} at node $s'$
(ii) Second, if an entry corresponding to $wantedtagvec[\overline{X}]$ does not exist in $L[\overline{X}],$ but an entry corresponding to $M.tagvec[\overline{X}]$ exists, then the re-encode function is used to ensure that $ResponsetoValInq$ stores a value corresponding to $\mathbf{0} \in \mathcal{V}$ corresponding to object $\overline{X}$. That is, the effect of object $\overline{X}$ is removed from $ResponsetoValInq.$ - an example is object $X_2$ in Fig. \ref{fig:respond} at node $s'$. 
(iii) Finally, if no object corresponding to $M.tagvec[\overline{X}]$ is present in $L[\overline{X}],$ then, no re-encoding is performed, and $ResponsetoValInq$ has an encoded value that corresponds to $M.tagvec[\overline{X}]$.  The node  responds with a$\langle \texttt{val\_resp\_encoded}, ResponsetoValInq, clientid, opid, \overline{X},wantedtagvec\rangle$ 
message. In case (i), note that $ResponsetoValInq$ stores a codeword symbol that corresponds to the encoding of the object version denoted in $wantedtagvec[X]$ for object ${X}$. In cases (ii) and (iii), our proof shows that the node that sent the $\texttt{val\_inq}$ message will be able to modify $ResponsetoValInq$ on receiving it to ensure that the modified value corresponds to the encoding of object $X$ for version $wantedtagvec[X].$

\underline{On receiving a $\langle \texttt{val\_resp\_encoded}, \overline{M}, clientid, opid, \overline{X},requestedtags\rangle$ message:} from node $j,$ the node first checks whether $ReadL$ consists of an entry corresponding to $$(clientid,opid, \overline{X},requestedtags, w_1, w_2, \ldots, w_N).$$  If such an entry exists,
the node attempts to re-encode $\overline{M}$ to $Modified\_codeword$ to ensure that it corresponds to an encoding of object versions indicated in $requestedtags.$ The node initializes $Modified\_codeword$ to $\overline{M}.val$. The $Error1, Error2$ flags denote whether such a modification is successful. Specifically, for every object $X$ in $\mathcal{X}_{j},$ node $s$ considers three cases and performs the appropriate steps:
(i) if $requestedtags[X] = \overline{M}.tagvec[X],$ no modification corresponding to object $X$ is needed - an example is object $X_3,X_4,X_5$ in Fig. \ref{fig:respond} at node $5$.
(ii) Otherwise if $\overline{M}.tagvec[X] \neq 0$ and there is an entry $(M.tagvec[X],v1)$ in list $L[X],$ then the node re-encodes to ensure that $Modified\_codeword$ corresponds to an encoding of the value corresponding to the value of $X$ being $\mathbf{0} \in \mathcal{V}.$ An example is $X_1$ at server $s$ in Fig. \ref{fig:respond}. If $\overline{M}.tagvec[X] \neq 0$ and there is no entry $(M.tagvec[X],val)$, then $Error1[X]$ is set to 1. 

(iii) If $\overline{M}.tagvec[X] = 0$, or the re-encoding in step (ii) is successful (that is, $Error1[X] = 0,$) then the node checks if an entry  $(requestedtags[X],v2)$ exists in $L[X].$ If the entry exists, then the node re-encodes $Modified\_codeword$ to store an encoded value from $\mathcal{W}_{j}$ corresponding to value $v_2$ for object $X$ - an example is objects $X_1,X_2$ in Fig. \ref{fig:respond} at node $s$. If the entry does not exist in the list, then $Error2[X]$ is set to $1$.

In summary, if $Error1[X] = Error2[X]=0$ for all objects $X \in \mathcal{X}_{j},$ then \newline $Modified\_codeword$ stores an encoding corresponding to node $j$ with objects taking values corresponding to $requestedtags.$ In fact, we will show in our proofs that in every execution of \CausalEC,  we have  $Error1[X] = Error2[X]=0.$ In this case, the entry \newline $(client,opid, \overline{X},requestedtags, w_1, w_2, \ldots, w_N)$ in $ReadL$ is changed to set \newline $w_j=Modified\_codeword$.  The updated tuple is examined to check if $\bar{X}$ can be decoded, that is, if the non-null entries of vector $(w_1,w_2,\ldots,w_N)$ form a recovery set for $\overline{X}.$ If they form a recovery set, then object $\overline{X}$ is decoded, and the corresponding entry is removed from $ReadL$. If $clientid$ is an external client, the decoded value is sent to the client. If $clientid$ is $\texttt{localhost}$, then the value is added to list $L[\overline{X}]$.

If the node \underline{receives a $\texttt{val\_resp}$ message} with object $X,$ value $v$, client identifier $clientid$, operation $opid$ and tags $requestedtags$, then the node checks if there is a corresponding pending read in $ReadL$. If such a read exists, then the node responds to the client with the value $v$, or if the client is $\texttt{localhost},$ it adds $(requestedtags[X],v)$ to $L[X].$ Then it removes the corresponding entry from $ReadL.$

\underline{On receiving a $\langle\texttt{del}, X, t\rangle$ message} from node $j$, node $s$ adds $(t,j)$ to $Del[X].$

\underline{On receiving a $\langle\texttt{app}, X, v, t\rangle$ message} from node $j$, tuple  $(j,X,v,t)$ is added to $\texttt{Inqueue}.$

{\large {\textbf{Internal Actions of server $s$}}:}
See Algorithm \ref{alg:internal_actions}.

\underline{The \texttt{Apply\_Inqueue} action} - similar to the replication based protocol of \cite{ahamad1995causal} - applies the enqueued tuples in $Inqueue$  on checking a particular predicate to ensure causality. Specifically, if $(j,X,v,t)$ is at the head of $\texttt{Inqueue}$, the predicate checks if the local vector clock $vc_s$ is not behind the $t.ts$ in all components, except the $j$th component; for the $j$th component, the predicate requires that $t.ts = vc_{s}[j] + 1.$ If the predicate is true, then $(t,v)$ is added to the list $L[X]$ and the packet is removed from $Inqueue$. Unlike the replication-based protocol \cite{ahamad1995causal}, the apply action  also clears pending reads to object $X$ in $ReadL$ by responding to the client (line \ref{line:readreturn_apply} in Algorithm \ref{alg:internal_actions}).

\underline{The $\texttt{Encoding}$ internal action} aims to update $M.val$ by encoding later versions of objects (i.e., object versions with higher tags) that are present in the lists $L[X].$ First consider objects contained in server $s$, that is $X \in \mathcal{X}_s$, which are to be encoded into $M.val$. If there exists an object $X \in \mathcal{X}_{s}$ such that $L[X].Highesttagged.tag > M.tagvec[X],$ then the $\texttt{Encoding}$ action checks if there exists a tuple $(M.tagvec[X],val)$ in $L[X]$ as well (line \ref{line:conditionforupdate} in Algorithm \ref{alg:internal_actions}). Depending on the truth-value of this check, the further state changes are:
\begin{itemize}
    \item If such a tuple exists , then the node re-encodes $M$ to use the value corresponding to $L[X].Highesttagged.tag$ and sets $M.tagvec[X] :=  L[X].Highesttagged.tag$ (Line \ref{line:updatetag1} in Algorithm \ref{alg:internal_actions}).  Further, a tuple $(L[X].Highesttagged.tag,s)$ is added to the delete list $DelL$ and a $\langle \texttt{del},X,L[X].Highesttagged.tag\rangle $ message is sent to all nodes containing object $X$.
\item If such a tuple does not exist, the node checks if a tuple $(\texttt{localhost}, \overline{opid} , X, M.tags, \overline{w})$ exists in $ReadL.$ If no such tuple exists in $ReadL,$ it is created by generating unique operation identifier $opid$, and vector $\overline{w}: \Pi_s(\overline{w}) = M.val$,  $\Pi_i(\overline{w}) = \bot, i \neq s$. The generated tuple is added to $ReadL$ (line \ref{line:readlentryencoding} in Algorithm \ref{alg:internal_actions})
, and a $\texttt{val\_inq}$ message is sent to all nodes.
\end{itemize}
The $\texttt{Encoding}$ action also involves ``bookkeeping''' actions pertaining to objects $X \notin \mathcal{X}_{s},$ with respect to keeping the tags $M.tagvec[X]$ updated. For such an object $X \notin \mathcal{X}_{s}$, if there is an entry in $L[X]$ with $L[X].Highesttagged.tag > M.tagvec[X],$ the action sets $U$ to be the set of all tags $t$ such that the delete list $Del[X]$ has at least one entry $(\hat{t},i)$ with $\hat{t} \geq t$ for every server $i$ that contains object $X$. Put differently, for every element $t$ in $U$, every server that contains object $X$ as per the error correcting code $\mathfrak{C}$ has sent at least one $\texttt{del}$ message with a tag at least as large as $t$. The set $\overline{U}$ is the set of all tags larger than $M.tagvec[X]$ where there is at least one element with that tag in the list $L[X].$
If the set $U \cap \overline{U}$ is non-empty then $M.tagvec[X]$ is updated to $max(U\cap \overline{U}).$ A entry $(max(\overline{U}\cap U),s)$ is made in the local delete list and a $\texttt{del}$ message is sent to every other node with tag $max(U\cap \overline{U}).$ These steps enable other nodes to delete the object value corresponding to the updated $M.tagvec[X]$ from their history lists.

\underline{The $\texttt{Garbage\_collection}$ internal action} aims to delete objects from lists. Consider an object $X \in \mathcal{X}$. Let $S$ denote the set of all tags $\overline{t}$ such that server $s$ has element $(\hat{t},j)$ in $DelL[X]$ with $\hat{t} \geq \overline{t}$ for every $j \in \mathcal{N}.$ Note that an element $\overline{t}$ is in $S$ only if it has received a  $\texttt{del}$ message for object $X$ with a tag at least as large as $t$ from every (other) node $j \in \mathcal{N}$. 
Let $tmax[X]$ denote the largest tag in $S$. An element $(t,v)$ in $L[X]$ is considered for deletion by the~$\texttt{Garbage\_collection}$ action if $t \leq tmax[X].$ In the proofs, we show the following invariant for every server $s$: $tmax[X]\leq M.tagvec[X].$  The conditions for deleting the element differ based on whether $X \in \mathcal{X}_{s}$ or whether $X \notin \mathcal{X}_{s}.$  For an object $X \notin \mathcal{X}_{s}$, the element $(t,v)$ is deleted from a list $L[X]$ if:
\begin{enumerate}[(a)]
    \item  $t \leq tmax[X]$ and $t < M.tagvec[X],$ and there is no tuple  $(clientid, opid, \overline{X}, tvec, \overline{w})$ in $ReadL$ with $tvec[X]=t$ \item or,  $t = tmax[X]= M.tagvec[X],$ and node $s$ has an entry $(M.tagvec[X],i)$ in $\texttt{DelL}[X]$ for every $i \in \mathcal{N},$ and $L[X]$ does not contain an element with tag strictly larger than $M.tagvec[X].$ 
\end{enumerate} 
Note that the condition in $(b)$ helps ensure that the lists are eventually empty in executions with finite number of write operations. In  Algorithm \ref{alg:internal_actions}, the above conditions are implied by the statements in line \ref{line:GCcondition1} \ref{line:pendingreads}, and \ref{line:s1set}. If $X \in \mathcal{X}_{s}$, then the conditions of deletion of an element $(t,v)$ are similar to $(a),(b)$ above, with one subtle difference: in $(a),$ we require strict inequality: $t < tmax[X].$  For objects $X \in \mathcal{X}_{s}$, the $\texttt{Garbage\_Collection}$ action performs some additional steps. Let $U$ to be the set of all tags $t$ such that there is an entry $(\hat{t},i)$ in $DelL[X]$ with tag $\hat{t} \geq t$ for every server $i$ that contains object $X$. Node $s$ sends $\texttt{del}$ messages to all other nodes with tag $max(U).$

%% file: proofs.tex
We state the correctness properties of $\CausalEC$ in Section \ref{sec:safety}. Sec. \ref{sec:communication} has a discussion on performance. Proofs are in  Appendices \ref{app:prelim}, \ref{app:safety}, \ref{app:liveness}, \ref{app:eventual_proof}, \ref{app:storage}.

\subsection{Correctness Properties of \CausalEC~}\label{sec:safety}
We start with safety in Theorem \ref{thm:causallyconsistent}. \CausalEC~ satisfies several liveness properties. Properties related to operation termination are described in Theorems \ref{lem:write_terminates} and \ref{lem:read_terminates_f=0}. Eventual Consistency and storage cost properties are stated in Theorems \ref{eventual consistent} and \ref{thm:storagecost}.

\begin{theorem}\label{thm:causallyconsistent}
$CausalEC$ satisfies causal consistency.
\end{theorem}

\begin{theorem}\label{lem:write_terminates}
Let $\beta$ be a fair execution of \CausalEC. Suppose $s$ is a non-halting node. Every write operation $\pi$ issued by a non-halting client  $c \in \mathcal{C}_s$ eventually terminates.
\end{theorem}

\begin{theorem}[Termination of Reads]\label{lem:read_terminates_f=0}
Let $\beta$ be a fair execution of \CausalEC and let $s$ be a server node. Then every read operation $\pi$ on object $X$ issued by a non-halting client $c \in \mathcal{C}_{s}$ terminates so long as there is at least one recovery set $S \in \mathcal{R}_{X}$ such the nodes in $S \cup \{s\}$ are non-halting nodes.
\end{theorem}

\begin{theorem}[Eventual Consistency] \label{eventual consistent}

Every execution of \CausalEC~is eventually consistent
\end{theorem}

\begin{theorem}[Storage Cost]
Consider a fair execution $\beta$ of \CausalEC~ where every server is non-halting. Suppose that after a point $P$ in $\beta$, there are no write operations invoked in the execution. Then, 
\begin{enumerate}[(a)]
\item eventually, for every object $X$, for every server $s$, list $L_s[X]$ is the empty set $\{\}$.
\item eventually, $\texttt{Inqueue}_s = \{\}$ for every server $s$, 
\item if there is a point $Q$ such that there are no read operations invoked after $Q$ in $\beta$, then eventually, $ReadL_{s} = \{\}$ for every server $s$.
\end{enumerate}
\label{thm:storagecost}
\end{theorem}

In plain words, Theorem \ref{thm:storagecost} indicates that all the transient costs associated with lists and the inqueue vanish. In stable state, the only state variable that depends on the object value is the codeword symbol $M$, which implies that the storage cost depends only on the erasure code. 
All proofs use some preliminary definitions and lemmas in Appendix \ref{app:prelim}.. 
The proof of Theorem \ref{thm:causallyconsistent} follows standard approaches, and is placed in Appendix \ref{app:safety}. Theorem \ref{lem:write_terminates} is  straightforward as every server responds immediately on receiving a request (see formal proof in Appendix \ref{app:liveness}). 
 The proof of Theorem \ref{lem:read_terminates_f=0} - placed in Appendix \ref{app:liveness} - relies on lemmas the show that $Error1_s[X]=Error2_s[X]=0$ at all points of an execution. These lemmas imply that if a node receives a $\texttt{val\_resp\_encoded}$ message in response to its $\texttt{val\_inq}$ message, then that the codeword symbol embedded in that $\texttt{val\_resp\_encoded}$ message can be re-encoded to enable a match to the local timestamp. Theorems \ref{eventual consistent} and \ref{thm:storagecost} are proved in Appendices \ref{app:eventual_proof} and \ref{app:storage}.

\subsection{Communication and Transient Storage Costs}
\label{sec:communication}

We analyze the costs of a low-cost variant of \CausalEC~ as follows: (i) inquiry and response messages corresponding to reads, and $\texttt{Del}$ messages for garbage collection contain a (scalar) Lamport timestamp \cite{Lamport_replicatedstatemachine}, rather than a vector timestamp, (ii) $\texttt{Del}$ messages are sent to a leader, who then forwards the messages to all the nodes, and (iii) a read operation sends $\texttt{val\_inquiry}$ messages to only one recovery set of $k$ nodes, and sends these messages to remaining nodes only after a timeout; in a failure-free execution with an appropriately set timeout, these messages will be sent to (and received from) $k$ nodes. 
More details about (i),(ii) and a formal algorithm pseudocode of the variant are presented in Appendix \ref{app:variant}.

Consider a data store over a connected network of $N$ server nodes to store $K$ equal-size objects of $B$ bits each, with the storage capacity of $\alpha B$ bits per node. Assume that $\alpha \geq K/N$ to ensure that the total storage capacity suffices to store all objects. Suppose that objects are grouped into $K/k$ groups of $k$ objects each and an $(N \alpha ,k)$ code over field size of $B$ bits is used for each group. Note from the example in Sec. \ref{subsec:motivation} that the benefits of \CausalEC~ are maximized when the code and placement are carefully fine tuned to the inter-node latencies. However, for simplicity and generality of analysis, we assume that a systematic Reed-Solomon code is used, where any $k$ codeword symbols suffices to decode all $k$ object values of a group. Our analysis assumes a maximum of $L$ updates per server. 

\underline{Write and Read Communication Costs:} A read operation sends $\texttt{val\_inq}$ to $k$ servers. A round trip to each server incurs a cost of $O(B)$ bits of data per message, and $k \log L$ bits  of metadata due to one Lamport timestamp per object in that group. 
The total read communication cost is $O(k) B +  O(k^2\log L)$. A write operation triggers propagation of values via the $app$ message, an internal read triggered by an \texttt{Encoding} action, and sending of delete messages as part of \texttt{Garbage\_Collection} action.The $app$ message sent to all $N$ nodes, incurs a communication cost of $O(N) B$, and the $\texttt{Encoding}$ action incurs $\left(O(k) B +  O(k^2 \log L)\right) = \left(O(N) B +  O(k^2 \log L)\right)$. With our optimized protocol, $O(N)$ delete messages with Lamport timestamps cost $O(N \log L).$ Thus total write communication cost is:
$ \centering O(N) B + O(k^2 \log L) + O(N\log L)\label{eq:cost}$



\underline{Transient Storage Overheads:}
While the stable state storage cost (neglecting metadata) is equal to $B/k,$ \CausalEC~incurs a transient overhead due to history lists; we analyze them here. Let $\rho_{w}$ denote the total write arrival rate for an object $X$, and let each server perform a $\texttt{Garbage\_ Collection}$ once every $T_{gc}$ seconds. Under some mild assumptions (See Appendix  \ref{app:transient}), the expected storage overhead of the history list is at most $\frac{3B}{\rho_w T_{gc}};$ the factor of $3$ appears because $2$ \texttt{Garbage\_Collection} actions may be required to remove a version, and a version may have to wait up to time up to $T_{gc}$ before a \texttt{Garbage\_Collection} is triggered. For workloads where very frequent writes to the same object (i.e., high $\rho_{w}$) occur only to a small set of ``hot'' objects in the system, our analysis suggests that \CausalEC~ would store exactly what is prescribed by the code for most objects, for most of the time. As an example, we conduct a coarse analysis for the default parameters of YCSB workload \cite{cooper2010benchmarking}. The workload has 120 million objects, with Zipfian parameter $0.99$. Assuming throughput $200,000$ requests per second and read/write ratio 50\%, we observe that $\rho_w < 1/1000$ per second for more than $95\%$ of the objects. If erasure coding with dimension $k$ is used for the $95\%$ of the objects with the lowest arrival rates (and replication\footnote{In practice, data stores detect arrival rates and adapt various parameters based on its characteristics; see, for example \cite{abebe2020morphosys,zare2021legostore}. So detecting a subset of objects with very high arrival rates and using replication for them is reasonable.} for the remaining $5\%$ of the objects),  and a lazy garbage collection (GC) performed every $T_{gc}=2 \textrm{minutes}$ (which is much greater than round trip times, see Fig. \ref{fig:AWS}), the average storage cost per erasure coded object is $(\frac{1}{k}+0.05)B$, a mere $5\%$ overhead due to history lists (See Appendix \ref{app:transient} for details).



%% file: related.tex

\textbf{Replication:} There is a long line of work exploring causally consistent data stores starting with the seminal work of \cite{ahamad1995causal}. Much of this work focuses on techniques for reducing the overheads of tracking causal dependencies (e.g,  vector clocks) for  both full replication \cite{lloyd2011don, du2013orbe, du2014gentlerain}, and for partial replication \cite{shen2015causal, xiang2019partially, helary2006efficiency,bravo2017saturn, mehdi2017can, akkoorath2016cure}. 
These works are complementary as they do not use erasure coding.

\noindent\textbf{Intra-object Erasure Coding:}
Erasure coding based algorithms for read/write memory emulation with atomic (linearizable) consistency are developed for both crash faults and byzantine faults in \cite{cadambe2017coded, Konwar_ipdps2016, CT, Hendricks, dobre_powerstore, Dutta}. 
There are systems \cite{chen2017giza, mu2014paxos, wang2020craft, du2016pando} that adapt consensus algorithms to utilize erasure coding and provide more involved data access primitives to the clients. Reference \cite{Anderson_etal} develops an eventually consistent key-value store that uses erasure coding. However, all these works including \cite{Anderson_etal} partition a data object and encodes the partitions of the data object - they do not use cross-object erasure coding. As a consequence, read operations do not return locally; \emph{every read operation} necessarily contact remote nodes. Like \CausalEC~, several previous algorithms store a history of values of older versions. Impossibility results in \cite{Cadambe_Wang_Lynch2016, spiegelman2016space} show that linearizable non-blocking erasure coding based algorithms necessarily store history. Understanding transient storage cost overheads for eventual/causally consistent storage is an area of future work.

\noindent \textbf{Cross-object Erasure Coding:} References \cite{eikel2014robust,AJX,lyu2018erasure} utilize cross-object erasure coding in distributed algorithms for data storage, but in settings that are significantly simpler than ours. 
Reference \cite{eikel2014robust} studies a synchronous system and thus does not include the complexities of disparate versions being combined at different servers at a given execution point. 
Reference \cite{AJX} uses cross-object erasure coding in a much simpler multi-writer single-reader asynchronous system. The reference emulates a read/write object with regular semantics, unlike our system which allows for concurrent distributed operations. In a master's thesis, reference \cite{lyu2018erasure}, like our paper, develops a causally consistent data storage algorithm based on cross-object erasure codes. However, the algorithm of \cite{lyu2018erasure} only applies for the special "toy" scenario where (i) only an $(N,K=2)$ systematic erasure code is supported, and (ii) clients accessing a specific server is  restricted to access certain specific objects. The technical approach of \cite{lyu2018erasure} - which does not use re-encoding functions like us - does not appear to be easily generalizable to arbitrary values of $K$. The most critical differences between \CausalEC~ and prior works \cite{AJX,lyu2018erasure} come from comparing their liveness properties. In essence, both  \cite{AJX,lyu2018erasure} require a systematic code (where there are $K$ servers that store an uncoded copy of the data), and these servers are required to be always available to ensure termination of read operations. That is, their liveness properties do not inherit the fault-tolerance properties of the underlying erasure code unlike \CausalEC. 

%% file: conclusion.tex
\section{Conclusion}\label{sec:conclusion}
    We show that erasure coding is compatible with causal consistency through the development of \CausalEC. The development of a systematic approach to  erasure code design that optimizes storage-latency trade-offs  for general network topologies is a future research direction that complements our paper.  While we provide a coarse performance analysis in Section \ref{sec:communication}, a detailed systems understanding of the performance of \CausalEC~ and algorithmic improvements to reduce overheads is also an important future research direction. 

 \section*{Acknowledgement}
{This work is supported partially by NSF Grants CCF:1553248, CNS: 2211045 and by a Google Faculty Award. We thank Dr. Ramy E. Ali and Prof. Bhuvan Urgaonkar for  discussions that led to the development of a preliminary algorithm called \CausalEC\_\texttt{exp}, which is documented in Shihang Lyu's masters' thesis \cite{lyu2018erasure}; see Sec. \ref{sec:related} for a discussion. We thank Raj Pandey for help with the latency of partial replication in Fig. \ref{table:cost_comparison}. We also thank Prof. Bhuvan Urgaonkar for continued discussions throughout the development of the manuscript.}

%% file: AppendixA.tex
\newpage 
\section{Comparison with Partial Replication Protocols ~and Table in Fig. \ref{fig:AWS}}
\label{app:write_cost}
Recall the example of Sec. \ref{subsec:motivation} where four objects $\mathcal{X}_1, \mathcal{X}_2,\mathcal{X}_3,\mathcal{X}_4,$ are placed along $6$ data centers as follows: $\mathcal{X}_{1}$ is placed in Seoul and Ireland, $\mathcal{X}_{2}$ is placed Mumbai and London, and  $\mathcal{X}_{3},\mathcal{X}_{4}$  are respectively placed in North California and Oregon. For the table Fig. \ref{fig:AWS}, it is shown that the storage scheme has a worst-case latency of $228ms$ and an average latency of $88ms$. It is worth examining a protocol that achieves these numbers. Partial replication has been studied in \cite{helary2006efficiency, xiang2019partially} among others, and in these papers, the common assumption is that read operations are issued only to objects present locally at the data store. Such an assumption restricts a client's access to a subset of the objects in the data store to ensure causality, which is different from our model, where we a client is allowed to access all objects. The algorithm of \cite{xiang2019partially} (see \cite{xiang_vaidya_partial_arxiv}, Appendix F) allows clients to access all objects by enabling clients to access contact different servers for different objects. This flexibility comes due to a careful vector clock construction they develop to ensure causality. However, their operations are not non-blocking. To ensure causality, a read requests may be buffered at nodes that may have to wait for respones from \emph{specific servers} before completing the read. Thus, the latencies reported in Table \ref{fig:AWS} cannot be achieved by \cite{xiang2019partially}, and crash failures of specific servers can lead to reads being blocked for ever (even if a copy of that object exists at a different server). In other words, if we view partial replication a special case of erasure coding, that protocol does satisfy requirement (II) in the introduction.

In fact, to overcome this limitation, writes will have to propagate object values to all the servers - including servers that do not necessarily store that object, similar to \CausalEC. Our write communication cost in Table \ref{table:cost_comparison} is based on such a protocol. The key difference is that in \CausalEC~, writes incur a second overhead of up to $kB$ due to internal reads triggered as a part of re-encoding actions that update the stored codeword symbols with the new versions. Section \ref{sec:communication} has a discussion on communication costs of \CausalEC.

%% file: AppendixB.tex
\section{Preliminary Definitions and Lemmas for Correctness Proofs}
\label{app:prelim} 
This appendix contains notations, definitions and lemmas that are  utilized in proofs of correctness (statements are in Sec. \ref{sec:safety}).  In the sequel, our statements and proofs use the following notation. For state variable $S,$ its value at node $s$ at point $P$ in an execution $\beta$ is denoted by $S_s^{P}.$ For example, for an object $X$, the value of $M_s.tagvec[X]$ at point $P$ at node $s$ is denoted as $M_s.tagvec[X]^{P},$ or equivalently $M_s^{P}.tagvec[X]$.

The safety proofs use a partial order $\leadsto$ and total order $\prec$ defined below.

\begin{definition}[Timestamp of an operation, tag of a write operation] \label{def:ts_tag}
Let $\pi$ be an operation that terminates in an execution $\beta$ of \CausalEC~ with operation identifier $opid$ issued by a client $c \in \mathcal{C}_{s}$ with identifier $clientid$. The \emph{timestamp} of $\pi$, denoted as $ts(\pi)$, is defined as the $vc_{s}^{P}$, that is, node $s$'s vector clock at the point $P$, where:
\begin{itemize}
\item If $\pi$ is a write operation, $P$  is the first point where server $s$ sends a \newline $\langle opid,  \texttt{write-return-ack} \rangle$ message to client $c,$
\item If $\pi$ is a read operation, $P$  is the first point where server $s$ sends a $\langle \texttt{read-return}, opid, v \rangle$ message to client $c,$ for some $v \in \mathcal{V}$. 
\end{itemize}
The tag of a write operation $w$ that is invoked at client with identifier $id$ is denoted as  $tag(w)$, and is defined to be $(ts(w),id)$.
\end{definition}


\begin{definition}For distinct operations $\pi_1, \pi_2$ in an execution $\beta$ of \CausalEC, we define 
$ \pi_1 \leadsto \pi_2$ if:
    {\bf (a)} $\pi_1, \pi_2$ both have timestamps in $\beta$ and $ts(\pi_1) < ts(\pi_2)$, or {\bf (b)}$\pi_1, \pi_2$ both have timestamps with $ts(\pi_1)=ts(\pi_2),$ and $\pi_1$ is a write operation, or
  {\bf (c)} $\pi_1, \pi_2$ are read operations that have timestamps issued by the same client  with $ts(\pi_1)=ts(\pi_2),$ and $\pi_1 \rightarrow \pi_2,$ or {\bf (d)} $\pi_1,\pi_2$ are write operations and
  $tag(\pi_1)< tag(\pi_2)$, or
    {\bf (e)} $\pi_1$ has a timestamp, but $\pi_2$ does not have a timestamp in $\beta.$ 
    
    For write operations $\pi_1,\pi_2$ we say $\pi_1 \prec \pi_2$ if $tag(\pi_1) < tag(\pi_2)$. 
\label{def:ourcausalordering}
\end{definition}

\begin{lemma}\label{lem:vconlyincreases}
Let $\beta$ be an execution of $CausalEC$ and let $P$ and $Q$ be two points of $\beta$ such that $P$ comes after $Q.$ Then, for any server $s$, $vc_{s}^{P} \geq vc_{s}^{Q}.$
\end{lemma}
The proof is standard and omitted. The following Lemma indicates that in any execution, a write can be uniquely identified by its tag, timestamp, or Lamport-timestamp. 

\begin{lemma}\label{lem:uniquetag1}
Let $\pi_1$ and $\pi_2$ be two distinct write operations in a well-formed execution $\beta$ of CausalEC. Then
\begin{enumerate}[(a)]
\item If $ts(\pi_1) \nprec ts(\pi_2)$, and $\pi_1$ is issued by a client in $\mathcal{C}_{s}$ for some server $s,$ then $ts(\pi_1)[s] > ts(\pi_2)[s].$
    \item 
    $tag(\pi_1) < tag(\pi_2)$ or $tag(\pi_2) < tag(\pi_1).$ 
    \item $lt(\pi_1) \neq lt(\pi_2)$ 
\end{enumerate}
\end{lemma}

The following lemma is an immediate consequence of Lemma \ref{lem:uniquetag1} and is stated without proof.

 \begin{lemma}\label{lem:uniquetag}
In any execution of \CausalEC, wvery write operation has a unique timestamp, a unique Lamport-timestamp, and a unique tag.
\end{lemma}

The following lemma will be used in proofs of causal consistency and eventual consistency.

\begin{lemma}\label{lem:returntagofread}
Consider an execution $\beta$ of \CausalEC. Suppose $\pi$ is any read operation to object $X$ invoked by a client in $\mathcal{C}_{s}$, where $s$ is any server. Suppose that $\pi$ completes in $\beta$ and that the $\texttt{read-return}$ message is sent from server $s$ to the client of $\pi$ at point $Q$ in $\beta$.  Let $P$ be any point before $Q$ such that $L_{s}^{P}[X]$ is non-empty. Then $\pi$ returns the value $v$ which is the value of a write operation $\phi$ to object $X$ that satisfies: $tag(\phi) \geq L_s^{P}[X].Highesttagged.tag$
\end{lemma}

We prove Lemmas \ref{lem:uniquetag1} and \ref{lem:returntagofread} next.

\subsection{Proof of Lemma \ref{lem:uniquetag1}}
There are only two cases: (i) $\pi_1,\pi_2$ are issued by clients that correspond to the same server $\mathcal{C}_{s}$ and (ii) $\pi_1, \pi_2$ are issued clients that correspond respectively to different servers $\mathcal{C}_{s},\mathcal{C}_{s'}, s \neq s'$.

\underline{Case (i):} Without loss of generality, we assume that $\langle write, opid,X_1, v_1 \rangle$ of operation $\pi_1$ arrives at server $s$ after the corresponding $\langle write, opid,X_2, v_2 \rangle$ message from
$\pi_2$ arrives at server $s$. Let $P$ be the point if receipt of the first message and $Q$ be the point of receipt of the second message.  Note that $P$ comes after $Q$ in $\beta.$ Lemma \ref{lem:vconlyincreases} implies that $vc_{s}^{P} \geq vc_{s}^{Q}.$ Furthermore, on receipt of the  $\langle \texttt{write},X,v\rangle$ message, in line \ref{line:vcincrement}, the vector clock is incremented. So we can conclude that $vc_{s}^{P}[s] > vc_{s}^{Q}[s],$ and consequently  $vc_{s}^{P} > vc_{s}^{Q}$. Since, by definition, we have $tag(\pi_1).ts = vc_s^{P}$ and $tag(\pi_2).ts =vc_s^{Q}$, we conclude that $ts(\pi_1)[s]>ts(\pi_2)[s]$ and therefore, $ts(\pi_1) > ts(\pi_2), lt(\pi_1) \neq lt(\pi_2)$. Thus lemma statement (a) holds. Further,  this automatically implies that $tag(\pi_1) > tag(\pi_2)$, which implies lemma (b) holds. The proof relies on the following claims.

\underline{Case (ii):}
Suppose  $ts(\pi_1)[s] \leq ts(\pi_2)[s]$. We aim to show that $ts(\pi_2) > ts(\pi_1).$ Let $\pi_1,\pi_2$ be respectively issued by a clients in $\mathcal{C}_{s}, \mathcal{C}_{{s}'}.$

\begin{claim}
Consider any any tuple  $(\overline{s},\overline{X},\overline{v},\overline{t})$ in $Inqueue.Head.$ If we index all the writes by clients in $\mathcal{C}_{\overline{s}}$ as $\pi_{1},\pi_{2}, \ldots, $ in the order of the arrival of the corresponding $\texttt{write}$ messages at server $\overline{s}$, then:
 $\pi_{\overline{t}.ts[\overline{s}]}$ is a write to object $\overline{X}$ with value $v$. Further $\overline{t}.ts=ts(\pi_{\overline{t}.ts[\overline{s}]}).$
 \label{claim:writeordering11}
\end{claim}
\begin{proof}
The claim follows from the protocol at server $\overline{s}.$ Specifically, the server  increments $vc_{\overline{s}}[\overline{s}]$ by $1$ on receiving a  $\texttt{write}$ message, and never decrements it. Therefore, $\pi_{\overline{t}.ts[\overline{s}}$ is a write to object $\overline{X}$ with value $v$. Since any $\langle \texttt{app},\overline{s},\overline{X},\overline{v},\overline{t})\rangle$ message has the $\overline{t}$ being the timestamp, $\overline{X}$ being the object and $\overline{v}$ being the value of the write operation, so the claim follows.
\end{proof}

\begin{claim}
For any servers $s,\overline{s}$ for any point $Q$ of $\beta$, let $\pi$ denote the $vc_s^{Q}[\overline{s}]$-th write operation at server $s$ as per the ordering of Claim \ref{claim:writeordering}. Then $vc_s^{Q} \geq ts(\pi).$
\label{claim:writeordering21}
\end{claim}
\begin{proof}
From the protocol, server $s$ only increments $vc_{s}[\overline{s}]$ on an $\texttt{Apply\_Inqueue}$ action, where a tuple from $\overline{s}$ at $Inqueue.Head$ is processed and line \ref{line:applycondition} returns true. Consider the tuple $(\overline{s},\overline{X},\overline{v},\overline{t})$ which corresponds to $\pi$. This tuple is processed by point $Q$ and on processing this tuple, \ref{line:applycondition} is true, and line \ref{line:vcincrement_apply} is executed. Therefore, after the $\texttt{Apply\_Inqueue}$ that processes this tuple, $vc_{s} \geq \overline{t}.ts = ts(\pi).$ The claim follows for point $Q$ applying Lemma \ref{lem:vconlyincreases}.
\end{proof}

Because of claim \ref{claim:writeordering11}. we know that $\pi_1$ is the $ts(\pi_1)[s]$th write arriving at server $s$. Let $\overline{\pi}_{1}$ denote the $ts(\pi_2)[s]$th write arriving at server $s$. Because of the same reason as $(i)$  $ts({\pi}_{1})< ts(\overline{\pi}_{1}).$ Let $Q$ be the point at which server $s'$ receives the $\texttt{write},opid,X,v\rangle$ message from $\pi_2.$ Claim  \ref{claim:writeordering21} implies that $ts(\pi_2) \geq ts(\overline{\pi}_1) > ts({\pi}_1).$ 
This completes the proof of (a).
The proofs of (b) follows from noting that the initial (server) indices of $lt(\pi_1),lt(\pi_2)$ are distinct. For proof of $(c)$, note that  $tag(\pi_1).id, tag(\pi_2).id$ are distinct. If $tag(\pi_1).ts,tag(\pi_1).ts$ are incomparable, then $tag(\pi_1),tag(\pi_2)$ can be compared using their ids.
 \subsection{Proof of Lemma \ref{lem:returntagofread}}

 We first show the following weaker version.
 
 \begin{lemma}\label{lem:returntagofread2}
Consider an execution $\beta$ of \CausalEC. Suppose $\pi$ is any read operation to object $X$ invoked by a client in $\mathcal{C}_{s}$, where $s$ is any server. Suppose that $\pi$ completes in $\beta$ and that the $\langle\texttt{read},opid,X\rangle$ message from the client of $\pi$ is received by server $s$ at point $Q$ in $\beta$.  Let $P$ be any point before $Q$ such that $L_{s}^{P}[X]$ is non-empty. Then $\pi$ returns the value $v$ which is the value of a write operation $\phi$ to object $X$ that satisfies: $tag(\phi) \geq L_s^{P}[X].Highesttagged.tag$
\end{lemma}
We mimic the proof of Lemma \ref{lem:readLtagsmallerthanvc}, with appropriate modifications.
 
\begin{proof}
Let $\pi$ be an operation that terminates in $\beta$. Without loss of generality, assume that $\pi$ is issued by client $c$ in $\mathcal{C}_{s}$ for some server $s.$
Based on the protocol $\pi$ returns because at a point $P'$ in $\beta$, server $s$ sends  a $\langle \texttt{read-return}, opid, v\rangle$ to the client.

This message is sent on executing one of the following lines:
\begin{enumerate}
\item Line \ref{line:sendtoread1} in Algorithm \ref{alg:inputactions_clients}.
    \item Line \ref{line:readimmediatereturn} in Algorithm \ref{alg:inputactions_clients}, or
    \item Line \ref{line:read-send-val_resp_encoded} in Algorithm \ref{alg:input_actions}, or 
    \item Line \ref{line:valresp_handle} in Algorithm \ref{alg:input_actions}, or 
    \item Line \ref{line:readreturn_apply} in Algorithm \ref{alg:internal_actions}.
\end{enumerate}
If the read $\pi$ is not returned via Line \ref{line:readimmediatereturn} in Algorithm \ref{alg:inputactions_clients}, then, at point $Q$, $L_{s}^{Q}[X]$ is empty, or $L_{s}^{Q}[X].Highesttagged.tag < M_s^{Q}[X].tag$. In either case, at point $Q$ which is after point $P$, because of Lemma \ref{lem:listtag}, $M_{s}^{Q}[X].Highesttagged.tag \geq L_{s}^{P}[X].Highesttagged.tag.$ So, for reads that returns in all the above cases except Line \ref{line:readimmediatereturn} of Algorithm \ref{alg:inputactions_clients}, it suffices to show that $tag(\pi) \geq M_{s}^{Q}.tagvec[X].$

\underline{1. Line  \ref{line:sendtoread1} in Algorithm \ref{alg:inputactions_clients}.} 
 Note that $ts(\pi) = vc_s^{P'}.$
    In this case point $P'$ is the point of receipt of a $\langle \texttt{write}, opid, X, v\rangle$ message from a write operation $\phi$, and the read $\pi$ with operation id $opid$ returns the value of $\phi$. By Definition \ref{def:ts_tag}, we have $ts(\phi) = vc_s^{P'}$.  Because $P'$ comes after $P$, from Lemma \ref{lem:vc_after_apply}, we have $ts(\phi) = vc_s^{P'} > vc_{s}^{Q}.$ From Lemma \ref{lem:vc}, we have $vc_{s}^{P'} \geq M_s^{Q}[X].tag.$ Thus, we have   $tag(\phi) > M_s^{Q}[X].tag$ as desired.
\underline{2. Line  \ref{line:readimmediatereturn} in Algorithm \ref{alg:inputactions_clients}.}

The value $v$ returned by $\pi$ is one that forms a tuple $(t,v)$ in list $L[X],$ where $t = L_s^{Q}[X].Highesttagged.tag$. By Lemma \ref{lem:valuesarelegitimate}, $v$ is the value of the unique write $\phi$ with $tag(\phi)=t$. Because Line \ref{line:readimmediatereturn} is executed only if  \ref{line:readimmediatereturncheck1} is satisfied, we infer that $t \geq M_s^Q.tagvec[X]$. Since Lemma \ref{lemma:tagsalwaysincrease} implies that $M_s^Q.tagvec[X] \geq M_s^P.tagvec[X]$, we have $t \geq M_s^P.tagvec[X]$. Because $L_s^{P}[X]$ is non-empty, Lemma \ref{lem:listtag} combined with the fact that $t = L_s^Q[X].Highesttagged.tag \geq M_s^Q.tagvec[X]$ implies that $t \geq L_s^P[X].Highesttagged.tag$ as desired.

\underline{3. Line  \ref{line:read-send-val_resp_encoded} in Algorithm \ref{alg:input_actions}:} 
Note that at the point $R$ of receipt of the $\texttt{val\_resp\_encoded}$ with operation id $opid$ and tag vector $requestedtags,$ the fact that Line \ref{line:read-send-val_resp_encoded} was executed implies there exists an entry in $ReadL_s^R$ with operation id $opid$ and tag vector $requestedtags.$
From Lemma \ref{lem:valuesarelegitimate}, the input to the function $\chi_{T_O}$ on Line \ref{line:valresp_handle_decoding} is a set whose elements are of the form $(\overline{i},{v}_{\overline{i}})$ where ${v}_{i} = \Phi_i(w_1, w_2, \ldots, w_K)$ where, for $\ell \in \{1,2,\ldots,K\},$ we have $w_\ell$ is the value of the unique write with tag $requestedtags[X_\ell].$

Therefore, the output of the decoding function, which is returned to the read, is the value $w$ of the value of the unique write $\phi$ with tag $requestedtags[X].$  
Suppose the entry operation id  $opid$, and tag vector $requestedtags$  is added to $ReadL_s$ at point $Q$. 
From the algorithm, at point $Q$, the node executed  line \ref{line:addtoreadl} in Algorithm \ref{alg:inputactions_clients} or Line \ref{line:readlentryencoding} in Algorithm \ref{alg:internal_actions}. In either case, we have $requestedtags=M_s^{Q}.tagvec[X].$ Thus, the read returns the value of the unique write operation $\phi$ with tag $M_s^{Q}.tagvec[X].$

\underline{4. Line  \ref{line:valresp_handle} in Algorithm \ref{alg:input_actions}} 
Let $R$ be the point of a $$\langle \texttt{val\_resp}, X, v, clientid, opid, tvec\rangle.$$ From the condition on Line \ref{line:valresp_handle_condition}, we know that there exists an entry in $ReadL_s^{R}$ with operation id $opid$ and tag vector $tvec$. Let $Q$ be the point before $R$ where this entry was added to $ReadL_s.$ From Lemma \ref{lem:valuesarelegitimate}, we infer that $v$ is the value of the unique write $\phi$ with tag $tvec[X].$ The rest of the proof is similar to Case 3.

Suppose the entry operation id  $opid$, and tag vector $tvec$  is added to $ReadL_s$ at point $Q$. 
From the algorithm, at point $Q$, the node executed  line \ref{line:addtoreadl} in Algorithm \ref{alg:inputactions_clients} or Line \ref{line:readlentryencoding} in Algorithm \ref{alg:internal_actions}. In either case, we have $tvec=M_s^{Q}.tagvec.$ Thus, the read returns the value of the unique write operation $\phi$ with tag $M_s^{Q}.tagvec[X]$ as desired.

\underline{5. Line  \ref{line:readreturn_apply} in Algorithm \ref{alg:internal_actions}} 

The proof is similar to Cases  $3$ and $4$. Line \ref{line:readreturn_apply} is executed in an $\texttt{Apply\_Inqueue}$ internal action. The value $v$  returned is the value of the unique write operation $\phi$ with tag $t.$  The fact that Line \ref{line:readreturn_apply} was executed implies that the condition in \ref{line:readreturn_apply_condition} is satisfied. This condition implies that an entry with operation $opid$ and tag vector $tvec$ exists in $ReadL$ at the point of execution of the $\texttt{Apply\_Inqueue}$ action. The condition also implies that $t \geq tvec[X]$. 

Suppose the entry operation id  $opid$, and tag vector $tvec$  is added to $ReadL_s$ at point $Q$. 
From the algorithm, at point $Q$, the node executed  line \ref{line:addtoreadl} in Algorithm \ref{alg:inputactions_clients} or Line \ref{line:readlentryencoding} in Algorithm \ref{alg:internal_actions}. In either case, we have $tvec=M_s^{Q}.tagvec.$ Therefore, we have $tag(\phi) \geq tvec[X] = M_s^Q.tagvec[X]$ as desired.

\end{proof}
\begin{proof}[Proof of Lemma \ref{lem:returntagofread}]
Let $Q'$ denote the point at which server $s$ receives a $\langle read, opid, X\rangle$ message from the client on behalf of $\pi$. If, in $\beta,$ there exists any point before $Q'$ such that $L_{s}^{P'}[X].Highesttagged.tag \geq L_{s}^{P}[X].Highesttagged.tag,$ then Lemma \ref{lem:returntagofread2} implies that the read $\pi$ returns the tag of an operation $\phi$ that satisfies $tag(\phi) \geq L_{s}^{P'}[X].Highesttagged.tag \geq L_{s}^{P}[X].Highesttagged.tag.$ So to show the lemma, it suffices to consider executions $\beta$ where for all point $P'$ before $Q$, 
\begin{equation}L_{s}^{P'}[X].Highesttagged.tag < L_{s}^{P}[X].Highesttagged.tag \label{eq:12}
\end{equation}
Note that this automatically implies that $P$ occurs after $Q'.$

Let $P''$ be the first point after $Q'$ such that $$L_{s}^{P''}[X].Highesttagged.tag \geq L_{s}^{P}[X].Highesttagged.tag.$$

\begin{claim}
At point $P''$, there is a tuple of the form $(clientid,opid,X,tvec,\overline{w})$ with $tvec[X] = M_s^{Q'}.tagvec[X]$ in $ReadL_{s}^{P''}.$
\end{claim}
\begin{proof}
Because of (\ref{eq:12}), the point $P''$ comes after $Q'.$ Because of the Lemma hypothesis, $P''$ comes at or before $Q$. Because the read did not return at $Q',$ server  added a tuple $(clientid,opid,X,tvec,\overline{w}$ to $ReadL_{s}^{Q'}.$ Because the read has not yet returned before $P''$, the tuple exists in $ReadL_{s}^{P''}.$ 
\end{proof}

From the server protocol, we observe that $P''$ must be a point of arrival of a write in line \ref{line:writeaddtolist} of Algorithm \ref{alg:input_actions}, or an $\texttt{Apply\_Inqueue}$ action at line \ref{line:list_applyaction} of Algorithm \ref{alg:internal_actions}. We show that the lemma holds for both cases.

In first case, server $s$ sends a \texttt{read-return} message to $\pi$ in line \ref{line:sendtoread1} in Algorithm \ref{alg:input_actions}. Consequently, $P''=Q$ and $tag(\pi)=tag(\phi),$ where $\phi$ is the write operation being performed in line line \ref{line:writeaddtolist} of Algorithm \ref{alg:input_actions}. 

In the second case, it is instructive to note that $M_s^{Q'}.tagvec[X] < L_s^{P}[X].Highesttagged.tag \leq L_s^{P''}[X].Highesttagged.tag.$ Consequently, the condition in line \ref{line:readreturn_apply_condition} is satisfied, and line \ref{line:readreturn_apply} is executed with tag $t= L_s^{P''}[X].Highesttagged.tag \geq L_s^{P}[X].Highesttagged.tag$. Since the tag $t$ is the tag of the write operation $\phi$ that on behalf of which the $\texttt{apply}$ message is sent, the lemma statement readily holds.
\end{proof}

 \section{Proof of Theorem \ref{thm:causallyconsistent}}
\label{app:safety}

For any execution $\beta,$ the relation $\leadsto$ is a partial order among all the operations of $\beta.$
Let $\pi$ be any read operation to object $X$ that completes in an execution $\beta.$ Then $\pi$ returns the value $v$ of a write operation $\phi$ to object $X$ in $\beta$ with $ts(\phi) \leq ts(\pi)$. The proof of Theorem \ref{thm:causallyconsistent} relies on Lemmas  \ref{lem:causal1}, \ref{lem:causal2} stated next.

 \begin{lemma}
For any execution $\beta,$ the relation $\leadsto$ is an irreflexive partial order among all the operations of $\beta.$
\label{lem:causal1}
\end{lemma}

\begin{lemma}
Let $\pi$ be any read operation to object $X$ that completes in an execution $\beta.$ Then $\pi$ returns the value $v$ of a write operation $\phi$ to object $X$ in $\beta$ with $ts(\phi) \leq ts(\pi)$.
\label{lem:causal2}
\end{lemma}
 
 We state and prove some preliminary lemmas before proving Lemmas \ref{lem:causal1},\ref{lem:causal2}. Then we prove Theorem \ref{thm:causallyconsistent}.

\subsection{Preliminary Lemmas}

\begin{lemma}
Let $\beta$ be an execution of \CausalEC~ and $s$ be a server. Consider any point $P$ such that there exists a tuple \newline $(clientid,opid, X, tvec,\bar{v})$ for $id \in \mathcal{C}$ in $ReadL_s,$ where $opid$ is the identifier of read operation  $\pi$. Let $Q$ be a point after $P$ in $\beta$ such that a tuple $(clientid, opid, X, tagvec,\bar{w})$ does not exist in $ReadL_{s}$. Then server $s$ responds to read operation $\pi$ in $\beta$ with a $\texttt{read-return}$ message by point $Q$. 
\label{lem:read-liveness}
\end{lemma}

\begin{proof}
From the protocol, an element is removed from $ReadL_{s}$ in line \ref{line:read-remove1} in in Algorithm \ref{alg:inputactions_clients},  or lines \ref{line:read-remove2} or \ref{line:read-remove3} in Algorithm \ref{alg:input_actions}, 
or in line \ref{line:read-remove4}
 in Algorithm \ref{alg:internal_actions}. Based on the protocol, in each case, just before the removal of the $ReadL,$ a $\texttt{read-return}$ message to a read. Therefore, server $s$ responds to read operation $\pi$ in $\beta$ with a $\texttt{read-return}$ message by point $Q$.
\end{proof}

\begin{lemma}\label{lem:vc_after_apply}
Consider a node $s$ and a point $P$ at which the \newline $\texttt{Apply\_Inqueue}$ action is performed such that the if statement on line \ref{line:applycondition}~returns true. Suppose that just before point $P$, $InQueue.Head = (j,X,v,t).$ Then, at any point $Q$ which is identical to or after point $P$, we have  $vc_{s}^{Q} \geq t.ts$.
\end{lemma}

\begin{proof}
When a node $s$ performs an $\texttt{Apply}\_\texttt{InQueue}$ internal action, at the beginning of this action, $(j,X,v,t)$ is the head of $InQueue$, and $t.ts[p]\leq vc[p]$ for all $p\neq j$, and $t.ts[j]=vc[j]+1$.

On performing this $\texttt{Apply}\_\texttt{InQueue}$ action, the $j$th component of local vector clock is incremented to $vc[j]=t.ts[j]$.

So after this $\texttt{Apply}\_\texttt{InQueue}$ action is complete, $ vc[p]\geq t.ts[p]$ for all $p\neq j$ and $vc[j]=t.ts[j]$, so $vc\geq t.ts$.
\end{proof}

\begin{lemma}\label{lem:listtag}
Consider a point $P$ of an execution $\beta$ of $CausalEC$ such that at server $s$, the list $L_s^{P}[X]$ is non-empty. At any point $Q$ of $\beta$ that comes after $P$, at least one of the following statements is true:
\begin{itemize}
    \item  $L_s^{Q}[X]$ contains an element $(L_{s}^{P}[X].Highettagged.tag,v)$ for some $v$, and consequently $$ L_s^{Q}[X].Highesttagged.tag \geq L_s^{P}[X].Highesttagged.tag,$$ or 
    \item $M_s^{Q}.tagvec[X] \geq L_s^{P}[X].Highesttagged.tag$
\end{itemize}
\end{lemma}
\begin{proof}
Let $t=L_s^{P}[X].Highesttagged.tag.$ Since $L_s^{P}[X]$ is non-empty, there exists a tuple $(t,v)$ in $L_s^{P}[X].$ If $(t,v)$ exists in $L_s^{Q}[X]$ then the lemma statement holds trivially. We consider these case where $(t,v)$ does not exist in $L_s^{Q}[X]$. Let $Q'$ be the first point after $P$ at which tuple $(t,v)$ does not exist in $L_s^{Q}[X];$  note that $Q'$ is no later than $Q$.

From the protocol, a $\texttt{Garbage\_Collection}$ action took place at $Q'.$ In the execution of the action, it executed either line \ref{line:GC1} or line \ref{line:GC2} in Algorithm \ref{alg:internal_actions}. From the code in these lines, we have $t \leq tmax[X].$ From the definition of $tmax[X]$ in line \ref{line:tmax}, we conclude that there is a tag $t' \geq t$ such that server $s$ has an element of $(t',s)$ in $DelL_s^{Q'}[X].$ From Lemma \ref{lem:tag-delete-ordering}, we infer that $t' \leq M_s^{Q'}.tagvec[X].$ We thus infer that $L_s^{P}[X].Highesttagged.tag = t \leq t' \leq M_s^{Q'}.tagvec[X] \stackrel{(a)}{\leq} M_s^{Q}.tagvec[X],$ where $(a)$ follows from Lemma \ref{lemma:tagsalwaysincrease}. This completes the proof.

\end{proof}

\begin{lemma}\label{lem:vc}
Consider an execution $\beta$ of $CausalEC$. 
At any point $P$ be $\beta,$ for any server $s$, we have:
    
    \begin{enumerate}[(a)]
      \item If list $L[X]$ is non-empty, then, for any entry $(t,v)$ in $L_s^{P}[X],$ we have  $vc_s^P \geq t.ts$
    \item $vc_s^P \geq M_s^{P}.tagvec[X].ts$ for any object $X$, and 
  
    \end{enumerate}
    
\end{lemma}

\begin{proof}

We begin the proof with the following claim.

\begin{claim}
Let $t=M_s^{P}[X].tag.$
There, either $t = \vec{0}$ or there is a point $Q$ before point $P$ such that there exists a tuple of the form $(t,v)$ in $L_s^{Q}[X].$
\label{claim:mtagslist}
\end{claim}

\begin{proof}[Proof of Claim]
If $t \neq \vec{0},$ then the $M_s[X].tag$ was updated in an $\texttt{Encoding}$ action either in lines \ref{line:updatetag1} or \ref{line:updatetag2} at some point $Q'$ that is not later than $P$. Based on the conditions to execute these lines in lines \ref{line:conditionforupdate} and \ref{line:Forloopupdatetag2}, there exists a tuple of the form $(t,v)$ in $L_s^{Q}[X],$ at point $Q'.$ Furthermore, since an $\texttt{Encoding}$ action does not add elements to the list $L_s[X]$,  this element must exist in the list at some point $Q$ before $Q'.$
\end{proof}

The claim above implies that if $(a) \Rightarrow (b).$ To see this, suppose as a contradiction, $(b)$ is violated, that is,  then $t=M_s^{P}{[X]}.tag.ts > vc_s^{P}$ for some point $P$. Then the claim shows that there is a point $Q$ before $P$ such that an element $(t,v)$ is in $L_s[X]^{Q}$ at some point $Q$ before $P$. Therefore $L_s^{Q}[X].Highesttagged.ts \geq t \geq vc_s^{P} \stackrel{(1)}{\geq} vc_s^{Q},$ where in $(1)$ we have used Lemma \ref{lem:vconlyincreases}. This is a contradiction of $(a)$. 

So to complete the proof, it suffices to show $(a).$

\underline{Proof of (a)}

Let $P$ be the first point of $\beta$ where $(a)$ is violated.

At point $P$ where $L_s^{P}[X]$ is non-empty, and there is an element $(t,v) \in L_s^{P}[X]$ such that  $t.ts \leq vc_{s}^{P}.$
Since $vc_s$ cannot decrease (Lemma \ref{lem:vconlyincreases}, we infer that an element $(t,v)$ is added  at point $P$ to $L_s^{P}[X]$ by executing:
\begin{enumerate}[(i)]
\item lines
\ref{line:writeaddtolist} in Algorithm \ref{alg:inputactions_clients}, or 
\item line \ref{line:list_applyaction} in Algorithm \ref{alg:internal_actions}, or
\item line \ref{line:readaddtolist} in Algorithm \ref{alg:input_actions}, or 
\item line \ref{line:valrespaddtolist} in Algorithm \ref{alg:input_actions}.
 \end{enumerate}
We show a contradiction in each case.

\underline{(i):}From the protocol, $L_s^{P}[X].Highesttagged.ts = vs_s^{P},$ which does not violate $(a),$ a contradiction.
\underline{(ii):} From the protocol, we know that before point $P$, the element at $\texttt{Inqueue}.Head$ had tag $t$ and value $v$. Furthermore, the fact that Line \ref{line:list_applyaction} implies that the condition in Line \ref{line:list_applyaction} returned true. From Theorem \ref{lem:vc_after_apply}, we have $vc_s^{P} \geq t.ts$. Therefore $(a)$ is not violated at point $P$.

\underline{(iii) and (iv)}
In these cases, because of conditions in lines \ref{line:valrespencoded_readcheck},  \ref{line:valresp_handle_condition}, a tuple with  tag-vector $tvec$ exists in $ReadL_s^{P}$ with $tvec[X] = t$ at the point of receipt of the $\texttt{val\_resp}$ or $\texttt{val\_resp\_encoded}$ message. 
Such an element is added at a point $Q$ before point $P$ either at line \ref{line:addtoreadl} in Algorithm \ref{alg:inputactions_clients}, or line \ref{line:readlentryencoding} in Algorithm \ref{alg:internal_actions}. In either case, we have $tvec= M_s^{Q}.tagvec$. Because we assume that $(a)$ is contradicted at point $P$, we have $t.ts = tagvec[X].ts > vc_s^{P}.$ From Lemma \ref{lem:vconlyincreases}, we have $tagvec[X].ts > vc_s^{Q},$ which readily implies that $M_s^{Q}.tagvec[X].ts > vc_s^{Q}.$ Therefore, point $Q$ violates $(b)$. Claim \ref{claim:mtagslist} implies that there is a point $R$ before $Q$ which violates $(a)$. Because $Q$ is no later than $P$, this violates the hypothesis that $P$ is the first point at which $(a)$ is violated. This completes the proof.

\end{proof}

\begin{lemma}\label{lem:readlist}
Consider an execution $\beta$ of $CausalEC$. 
At any point $P$ be $\beta,$ for any server $s$, suppose there exists an entry $$(clientid, opid, X, tagvec, w_1, w_2, \ldots, w_N)$$ in $ReadL_{s}^{P}.$ Then, for all objects $\overline{X} \in \mathcal{X}_{s},$ we have:
\begin{itemize}
    \item $ M_s^P.tagvec[\overline{X}]  \geq tagvec[\overline{X}],$ and 
    \item $vc_s^P \geq tagvec[\overline{X}].ts$
\end{itemize} 
\label{lem:readLtagsmallerthanvc}
\end{lemma}
\begin{proof}
Suppose the entry $(clientid, opid, X, tvec, w_1, w_2, \ldots, w_N)$  is added to $ReadL_s$ at point $Q$, which is before $P$. 
From the algorithm, at point $Q$, the node executed  line \ref{line:addtoreadl} in Algorithm \ref{alg:inputactions_clients} or Line \ref{line:readlentryencoding} in Algorithm \ref{alg:internal_actions}. In either case, we have $tvec=M_s^{Q}.tagvec.$ Consider any object $\overline{X} \in \mathcal{X}_{s}$. From Lemma \ref{lemma:tagsalwaysincrease}, we have $tvec[\overline{X}] = M_s^{Q}.tagvec[\overline{X}] \leq M_s^{P}.tagvec[\overline{X}]$ as desired. Further, this  readily implies that $tvec[\overline{X}].ts \leq M_s^{P}.tagvec[\overline{X}].ts$. From Lemma \ref{lem:vc}, we have $M_s^{P}.tagvec[\overline{X}].ts \leq vc_s^{P}.$ Therefore, we have    $tvec[\overline{X}].ts \leq vc_s^{P}.$ This completes the proof.
\end{proof}

\begin{lemma}
Consider any execution $\beta$ of $CausalEC.$
\begin{enumerate}[(i)]
\item 
If an item $(t,v)$ is in list $L[X],$ then $v$ is the value of the unique write with tag $t$
\item at any point $P$ in the execution, 
$M_s^{P}.val = \Phi_{s}(w_1, w_2, \ldots, w_N)$ where $w_i$ is the value of the unique write with tag $M_s^{P}.tag[O_i]$ if $M_s^{P}.tag[O_i] \neq 0$, and $w_{i} = \vec{0}$ if $M_s^{P}.tag[O_i] = \vec{0}.$
\item When node $s$ sends a $$\langle \texttt{val\_resp\_encoded}, \bar{M}, clientid, opid, X, tagvec \rangle, $$ then 
$$M.val = \Phi_{s}(w_1, w_2, \ldots, w_N)$$ where $w_i$ is the value of the unique write with tag $\bar{M}.tagvec[O_i]$ if $\bar{M}.tag[O_i] \neq 0$, and $w_{i} = \vec{0}$ if $\bar{M}.tagvec[O_i] = \vec{0}.$
 \item 
 Suppose there exists an entry \newline $(clientid, opid, X, tagvec, w_1, w_2, \ldots, w_N)$ in $ReadL_{s}^{P}$. Then $w_{i}$ is either equal to $\bot$, or $\Phi_{i}(v_1, v_2, \ldots, v_K)$ where $v_i$ is the value of the unique write with tag $tagvec[O_i]$.
    \end{enumerate}
    \label{lem:valuesarelegitimate}
 \end{lemma}

\begin{proof}[Proof Sketch]

The lemma holds because of the following claims.

\begin{claim}
Consider an execution $\beta$ where at point $P$, a \newline $\texttt{val\_resp\_encoded}$ message is received. Suppose in execution $\beta$, at every point before point $P$, all the invariants stated by the lemma holds. Then, the lemma holds at point $P$ as well. 
\label{lem:val_resp_encoded_integrity1}
\end{claim}

\begin{proof}
At point $P$, a $$\langle\texttt{val\_resp\_encoded},\bar{M},clientid, opid,requestedtags, \rangle$$ message is received by server $s$ from server $s'$. If the condition in \ref{line:valrespencoded_readcheck} in Algorithm \ref{alg:input_actions} is not satisfied, then the server state does not change, the lemma continues to hold. Henceforth, we assume that the condition in \ref{line:valrespencoded_readcheck}, that is, there is a tuple $(clientid, opid, \overline{X}, requestedtags, \overline{w})$ in $ReadL_s^{P}.$

We show that if $Error1[X] =Error2[X]=0$ for all $X \in \mathcal{X}_{s'},$ then by the time line \ref{line:errorcondition} is executed, $Modified\_Codeword$ contains a codeword that is equal to $\Phi_{s'}(v_1, v_2, \ldots, v_K)$ where $v_\ell, \ell=1,2,\ldots,K$ represents the value of the unique write whose tag is equal to $requestedtags[X_\ell].$ 
 Before entering the for loop, because theitem (iii) of the lemma statement holds before point $P$, and $Modified\_Codeword$ is initialized to $\overline{M}.val$ in line \ref{line:modifiedcodeword_init}, we know that  
$Modified\_Codeword = \Phi_{s'}(w_1, w_2, \ldots, w_K)$
where $w_{i}$ is the value of the unique write with tag $\overline{M}.tag[X_i].$ Based on Definition \ref{def:ecobjectstored}, it follows that we can write $$Modified\_Codeword = \Phi_{s'}(\overline{w}_1, \overline{w}_2, \ldots, \overline{w}_K)$$ where
$$ \overline{w}_{k} =\left\{\begin{array}{cc} \textrm{the value of the write with tag}requestedtags[X_k] & \textrm{if }X_k \notin \mathcal{X}_{s'} \\ \textrm{the value of the write with tag}\overline{M}.tagvec[X_k]& \textrm{if }X_k \in \mathcal{X}_{s'} \end{array}\right.$$ 

We examine the for loop in Line \ref{line:valrespencoded_objectforloop} iteratively over the object in $\mathcal{X}_{s'}$, and show that the loop ensures that for each $X \in \mathcal{X}_{k},$ the codeword in $Modified\_Codeword$ is modified to reflect the encoding of the value of the write with tag $requestedtags[X].$ 

For an object $X \in \mathcal{X}_{s'}$, if $requestedtags[X] = \overline{M}.tagvec[X],$ then no modification is required and the For loop \ref{line:valrespencoded_objectforloop} proceeds to the next object. If $requestedtags[X] = \overline{M}.tagvec[X],$ since the hypothesis of the claim states that $Error1[X] \neq 1$, we infer from condition in line \ref{line:metaerrorcondition1}  that either $\overline{M}.tagvec[X] = \mathbf{0}$ or line \ref{line:Modifiedcodeword1} is executed. In either case, by the time line \ref{line:metaerrorcondition2} is checked, $Modified\_Codeword$ stores a codeword that corresponds to value $\mathbf{0}$ for object $X$. Because $Error2[X] \neq 1,$ line \ref{line:Modifiedcodeword2} is executed. At the end of this line,   $Modified\_Codeword$ stores a codeword that corresponds to value of the write $requestedtags[X]$ for object $X$.

Therefore,at the point where line \ref{line:errorcondition} is executed, 

$$Modified\_Codeword = \Phi_{s'}(\overline{v}_1, \overline{v}_2, \ldots, \overline{v}_K)$$ where
$$ \overline{v}_{k} = \textrm{the value of the write with tag}requestedtags[O_k]$$ 

This implies that if $ReadL_s^{P}$ or if $L_s^{P}[X]$ are modified in line \ref{line:valrespencoded_addtuple} or line \ref{line:readaddtolist} are consistent with the (i),(iv) stated in the lemma. 
\end{proof}

\begin{claim}
Consider an execution $\beta$ where at point $P$, a \newline $\texttt{val\_resp\_encoded}$ is sent by server $s$. Suppose in execution $\beta$, at every point before $P$, all the invariants stated by the lemma holds. Then, the lemma holds at point $P$ as well. 
\label{lem:val_resp_encoded_integrity2}
\end{claim}

\begin{proof}
Note that point $P$ is the point of receipt of a $\texttt{val\_inq}$ message. From Line \ref{line:Responsetovalinqinit} in Algorithm \ref{alg:input_actions}, we know that \newline $ResponsetoValInq$ is initialized to $M.val$. Because invariant (ii) of the lemma is true by point $P$, we know that after executing line \ref{line:Responsetovalinqinit}, we have 
 $ResponsetoValInq.val = \Phi_{s'}({w}_1, w_2, \ldots, w_K)$ where
$w_{k}, k=1,2,\ldots,K$ is the  value of the write with tag $M_s^{P}.tagvec[X_k]$. Also note that $ResponsetoValInq.tag$ is equal to $M_s[X].tag$. Since invariant (i) stated by the lemma is true by point $P$, and the action at point $P$ does not modify $L_s^{P}[X]$, the invariant is true at point $P$ as well. Specifically, any value $(t,v)$ in $L_s^P[X]$ has the property that $v$ is the value of the unique write with tag $t$. Therefore, after executing   \ref{line:ResponsetoValInqvalupdate1}, \ref{line:ResponsetoValInqtagupdate1}, or after executing \ref{line:ResponsetoValInqvalupdate2}, \ref{line:ResponsetoValInqtagupdate2}, the following invariant is true:

$ResponsetoValInq.val = \Phi_{s'}(\overline{w}_1, \overline{w}_2, \ldots, \overline{w}_K)$ 
where $\overline{w}_{k}, k=1,2,\ldots,K$ is the  value of the write with tag $ResponsetoValInq.tagvec[X_k].$  

Therefore $(iii)$ is true at point $P.$ Furthermore, the action at point $P$ does not modify $M_s^{P}$ or $ReadL_{s}^{P}.$ Since these variables satisfy Lemma invariants (ii),(iv)  before point $P$  they satisfy (ii) and (iv) at point $P$ as well.  
\end{proof}

The lemma can then be proved by induction on the sequence of states in $\beta.$ The lemma is clearly true at the initial point of $\beta.$ Assume that the lemma is true point $P$ Of $\beta.$ We can show that the lemma is true at point $P'$ which immediately succeeds $P$. We omit the mechanical details for the sake of brevity, and only provide a sketch here. Informally, at point $P'$,  $L[X]$ satisfies (i) because it is either updated by writes, or $\texttt{Apply\_Inqueue}$ actions, or receipt of $\texttt{val\_resp}$ or $\texttt{val\_resp\_encoded}$ messages. Because $(i),(ii),(iii),(iv)$ are satisfied at point $P'$, the only non-trivial case is to show (i) after the receipt of $\texttt{val\_resp\_encoded}$ message, which is shown in Claim \ref{lem:val_resp_encoded_integrity1}. 

Because $(i), (ii)$ is satisfied at $P'$, and because $M.tagvec$ and $M.val$ are only updated as a part of $\texttt{Encoding}$ action, the result of the actions satisfy $(ii)$ at point $P$ as well. Claim \ref{lem:val_resp_encoded_integrity2} shows that (iii) is satisfied. $ReadL$ is updated only on receipt of a read, or receipt of $\texttt{val\_resp\_encoded}$ message, or on behalf of an encoding operation. Because of Because $(ii), (iii)$  are satisfied at $P'$, and because of  Claim \ref{lem:val_resp_encoded_integrity1}, $ReadL$ satisfies $(iv).$ as well.

\subsection{Proof of Lemma \ref{lem:causal1}}
\begin{proof}
To show that $\leadsto$ is an irreflexive partial order, it suffices to consider any execution $\beta$ and show that:
\begin{enumerate}[(I)]
    \item  For any two distinct operations $\pi_1,\pi_2$ in $\beta$,  $\pi_1\leadsto \pi_2 \Rightarrow \pi_2 \not \leadsto \pi_1$.
    \item $\pi_1 \leadsto \pi_2, \pi_2 \leadsto \pi_3 \Rightarrow \pi_1 \leadsto \pi_3$. 
\end{enumerate}

\underline{Proof of {(I)}}
Suppose for a contradiction that $\pi_1 \leadsto \pi_2$, and $\pi_2  \leadsto \pi_1$.
By Definition \ref{def:ourcausalordering}, $\pi_1,\pi_2$ both acquire timestamps, and  $ts(\pi_1) \leq ts(\pi_2)$, and $ts(\pi_2) \leq ts(\pi_1)$. This implies that $ts(\pi_1)=ts(\pi_2)$. 

\begin{claim}
$\pi_1,\pi_2$ are both read operations.
\end{claim}
\begin{proof}[Proof of Claim]
Because of Lemma \ref{lem:uniquetag1} which states that distinct write operations have different timestamps, $\pi_1$ and $\pi_2$ cannot both be write operations. We show that even one of them cannot be a write operation. Suppose as a contradiction, $\pi_1$ is a write and $\pi_2$ is a read. From the definition of $\leadsto$ in Definition \ref{def:ourcausalordering}, we cannot have $\pi_2 \leadsto \pi_1.$ Therefore $\pi_1$ is not a write. By a symmetric argument, $\pi_2$ is also not a write.
\end{proof}

Because $ts(\pi_1)=ts(\pi_2)$ and $\pi_1 \leadsto \pi_2,$ Definition \ref{def:ourcausalordering} implies that $\pi_1$ and $\pi_2$ happen at a same client and $\pi_1 \leadsto \pi_2$. Because $ts(\pi_1)=ts(\pi_2)$ and $\pi_2 \leadsto \pi_1,$ we have $\pi_2 \leadsto \pi_1.$ However, this is a contradiction to our earlier conclusion that $\pi_1 \leadsto \pi_2;$ this completes the proof.

\underline{Proof of {(II)}}
Consider three operations $\pi_1, \pi_2, \pi_3$ such that $\pi_1 \leadsto \pi_2, \pi_2 \leadsto \pi_3.$ This implies that $\pi_1,\pi_2$ acquire timestamps in $\beta$. If $\pi_3$ does not acquire a timestamp in $\beta$, then Definition \ref{def:ourcausalordering} implies that $\pi_1 \leadsto \pi_3$ to complete the proof.
We now assume that $\pi_3$ acquires a timestamp in $\beta.$

From Definition \ref{def:ourcausalordering}, we have $ts(\pi_1) \leq ts(\pi_2) \leq ts(\pi_3).$ If at least one of these inequalities is strict, we readily have $ts(\pi_1)< ts(\pi_3) \Rightarrow \pi_1 \leadsto \pi_3,$ which completes the proof. 

We now consider the case where both these inequalities are not strict, that is, $ts(\pi_1) = ts(\pi_2)=ts(\pi_3).$ The fact that $\pi_1 \leadsto \pi_3 \leadsto \pi_3$  implies that $\pi_2, \pi_3$ are both read operations issued at the same client with $\pi_2 \rightarrow \pi_3$. The fact that $\pi_1 \leadsto \pi_2$ with $ts(\pi_1) = ts(\pi_2)$ implies that, $\pi_1$ is a write operation or $\pi_1 \rightarrow \pi_2;$ in either case, we show that $\pi_1 \leadsto \pi_3$ to complete the proof.
Specifically, if $\pi_1$ is a write operation, then combined with the fact that $ts(\pi_1) = ts(\pi_3),$ Definition \ref{def:ourcausalordering} implies that $\pi_1 \leadsto \pi_3$. On the other hand, if $\pi_1 \rightarrow \pi_2,$ because we have already shown that $\pi_2 \rightarrow \pi_3,$ we have $\pi_1 \rightarrow \pi_3.$ Definition \ref{def:ourcausalordering} implies that $\pi_1 \leadsto \pi_3$ to complete the proof.

\end{proof}

\subsection{Proof of Lemma \ref{lem:causal2}}

\begin{proof}
Let $\pi$ be a read operation that terminates in $\beta$. Without loss of generality, assume that $\pi$ is issued by client $c$ in $\mathcal{C}_{s}$ for some server $s.$
Based on the protocol $\pi$ returns can return because at a point $P$ in $\beta$, server $s$ sends  a $\langle \texttt{read-return}, opid, v\rangle$ to the client. This message is sent on executing one of the following lines:
\begin{enumerate}
\item Line \ref{line:sendtoread1} in Algorithm \ref{alg:inputactions_clients}.
    \item Line \ref{line:readimmediatereturn} in Algorithm \ref{alg:inputactions_clients}, or
    \item Line \ref{line:read-send-val_resp_encoded} in Algorithm \ref{alg:input_actions}, or 
    \item Line \ref{line:valresp_handle} in Algorithm \ref{alg:input_actions}, or 
    \item Line \ref{line:readreturn_apply} in Algorithm \ref{alg:internal_actions}.
\end{enumerate}

We show that the lemma holds for each of these cases. Note that $ts(\pi) = vc_s^{P}.$

\underline{1. Line  \ref{line:sendtoread1} in Algorithm \ref{alg:inputactions_clients}.} 
In this case point $P$ is the point of receipt of a $\langle \texttt{write}, opid, X, v\rangle$ message from a write operation $\phi$. By Definition \ref{def:ts_tag}, we have $ts(\phi) = vs_{s}^{P}.$
        Therefore $\pi$ returns the value $v$ of a write operation $\phi$ with $ts(\phi) = ts(\pi).$

\underline{2. Line  \ref{line:readimmediatereturn} in Algorithm \ref{alg:inputactions_clients}.} 

In this case, the value $v$ is one that forms a tuple $(t,v)$ in list $L[X].$ The value $v$ is the value of the unique write operation $\phi$ with timestamp $ts(\phi) = t.ts.$ Furthermore, from Lemma \ref{lem:listtag}, we have $t.ts \leq vs_s^{P}.$ Therefore $\pi$ returns the value of a write operation $\phi$ with $ts(\phi) \leq ts(\pi).$

\underline{3. Line  \ref{line:read-send-val_resp_encoded} in Algorithm \ref{alg:input_actions}} 
Note that at the point of receipt of the $\texttt{val\_resp\_encoded}$ with operation id $opid$ and tag vector \newline $requestedtags,$ the fact that Line \ref{line:read-send-val_resp_encoded} was executed implies there exists an entry in $ReadL$ with operation id $opid$ and tag vector $requestedtags.$
From Lemma \ref{lem:valuesarelegitimate}, the input to the function $\chi_{T_O}$ on Line \ref{line:valresp_handle_decoding} is a set whose elements are of the form $(\overline{i},{v}_{\overline{i}})$ where ${v}_{i} = \Phi_i(w_1, w_2, \ldots, w_K)$ where, for $\ell \in \{1,2,\ldots,K\},$ we have $w_\ell$ is the value of the unique write with tag $requestedtags[X_\ell].$ Therefore, the output of the decoding function, which is returned to the read, is the value $w$ of the value of the unique write $\phi$ with tag $requestedtags[X].$ Since, at this point, there exists a tuple in $ReadL$ with tag vector $requestedtags,$ Lemma \ref{lem:readLtagsmallerthanvc}  implies that $vc_{s}^{P} \geq requestedtags[X].ts.$  Therefore, the read returns the value of a write $\phi$ with $ts(\phi) \leq ts(\pi)$. 

\underline{4. Line  \ref{line:valresp_handle} in Algorithm \ref{alg:input_actions}} 
This is the point of a $$\langle \texttt{val\_resp}, X, v, clientid, opid, tvec\rangle.$$ From the condition on Line \ref{line:valresp_handle_condition}, we know that there exists ane entry in $ReadL$ with operation id $opid$ and tag vector $tvec$ at point $P$. Furthermore, from Lemma \ref{lem:valuesarelegitimate}, we infer that $v$ is the value of the unique write $\phi$ with tag $tvec[X].$  From Lemma \ref{lem:readLtagsmallerthanvc}, we conclude that $tvec[X].ts \leq vc_s^{P} = ts(\pi).$ Therefore, the read returns the value of a write $\phi$ with $ts(\phi) \leq ts(\pi)$. 

\underline{5. Line  \ref{line:readreturn_apply} in Algorithm \ref{alg:internal_actions}} 

The proof is similar to Case $2.$ The value $v$ is one that forms a tuple $(t,v)$ in list $L[X].$ This is because the tuple is added to list $L[X]$ in the same action in Line \ref{line:list_applyaction}.  The value $v$ is the value of the unique write operation $\phi$ with timestamp $ts(\phi) = t.ts.$ Furthermore, from Lemma \ref{lem:listtag}, we have $t.ts \leq vs_s^{P}.$ Therefore $\pi$ returns the value of a write operation $\phi$ with $ts(\phi) \leq ts(\pi).$

\end{proof}

\subsection{Proof of Theorem \ref{thm:causallyconsistent}.}
Lemma \ref{lem:causal1} implies that $\leadsto$ is a partial order on all operations in an execution $\beta$ of \CausalEC. We show that the  orderings $\leadsto, \prec$ of Definition \ref{def:ourcausalordering} satisfy properties {(a)}, {(b)} and {(c)} in Definition \ref{def:original_causal}. 

We first show property (a), that is: $\pi_1 \rightarrow \pi_2 \Rightarrow \pi_1 \leadsto \pi_2.$

\underline{Proof of (a) in Definition \ref{def:original_causal}}
If $\pi_1 \rightarrow \pi_2,$ then they are issued from the same client $c$. Note that $\pi_1$ terminates in $\beta$ because $\pi_2$ cannot be issued by $c$ before $\pi_1$ terminates. Therefore $\pi_1$ acquires a timestamp. If $\pi_2$ does not acquire a timestamp, then $(d)$ in Definition \ref{def:ourcausalordering} implies that $\pi_1 \leadsto \pi_2.$ We now handle the case where $\pi_2$ acquires a timestamp. 

Note that these operations send their message to the same server $s$ where $c \in \mathcal{C}_{s}.$ Since $\pi_1$ terminates before $\pi_2$ begins, the server sends a response to $\pi_1$ at a point $P$  before the point $Q$ where it responds to $\pi_2.$ Because of Lemma \ref{lem:vconlyincreases}, $vc_{s}^{Q} \geq vc_s^{P}.$ Further, since $vc_{s}^{P} = ts(\pi_1)$ and $vc_s^{P}=ts(\pi_2),$ we have $ts(\pi_1) \leq ts(\pi_2).$ If the inequality is strict, then we readily have $\pi_1 \leadsto \pi_2$ from Definition \ref{def:ourcausalordering}. 

We examine the case that $ts(\pi_1)=ts(\pi_2).$ We claim that if this is the case, then $\pi_2$ is a read. As a contradiction, if $\pi_2$ is a write, based on the server protocol in Alg. \ref{alg:inputactions_clients}, server $s$ increments the vector clock at point $Q$ before responding to $\pi_2$. Therefore, we have $vc_{s}^{Q} > vc_{s}^{P}$ which implies $ts(\pi_1) < ts(\pi_2)$ contradicting the equality of the timestamps. Therefore $\pi_2$ is a read operation, and $\pi_1$ can be a read or a write operation. In either case, (b) and (c) in Definition \ref{def:ourcausalordering}, imply that $\pi_1 \leadsto \pi_2.$

\underline{Proof of (b) in Definition \ref{def:original_causal}}
The proof follows from the definition of $\leadsto $ and $\prec$. Specifically, if $\pi_1,\pi_2$ are writes and $\pi_1 \prec \pi_2,$ then $tag(\pi_1) < tag(\pi_2).$ Definition \ref{def:ourcausalordering} implies $\pi_1 \leadsto \pi_2.$

\underline{Proof of (c) in Definition \ref{def:original_causal}:}
Consider any read operation $\pi_2$ on object $X$ at server $s$ that completes with value $v$. Let $S$ be the set of all write operations $w$ such that $w \leadsto \pi_2.$ Because of Lemma \ref{lem:causal2}, we know that $S$ is non-empty. Because $\pi_2$ completes in $\beta,$ it acquires a timestamp $ts(\pi_2)$. Because every write operation acquires a unique timestamp, and every timestamp is a vector of non-negative entries, $S$ is finite.  Let $\pi_1$ be the write with the highest tag in $S$ to object $X$. Note that $\pi_1$ is unique and well-defined because of Lemma \ref{lem:uniquetag1} combined with the fact that $S$ is non-empty and finite. We aim to show that $\pi_2$ returns the value of write $\pi_1$. 

Because of Lemma \ref{lem:readLtagsmallerthanvc}, $\pi_2$ returns the value of some write $\pi$ in $S$ with $ts(\pi) \leq ts(\pi_2).$ Therefore $\pi \leadsto \pi_2$ and consequently $\pi \in S.$
 Lemma \ref{lem:uniquetag1} implies that $tag(\pi) < tag(\pi_1)$. Consequently, $ts(\pi).ts[\overline{s}] < ts(\pi_1).ts[\overline{s}] \leq ts(\pi_2).ts[\overline{s}],$ where $\pi_1$ is a write issued to server $\overline{s}$ from some client in $\mathcal{C}_{\overline{s}}.$  From Definition \ref{def:ts_tag}, $vc_{s}^{P}[\overline{s}] \geq ts(\pi_1).ts[\overline{s}]$, where $P$ is the point of read-return.

We consider two cases (i) $\overline{s}=s$ and (ii) $\overline{s} \neq s$. 
In case (i), from server protocol, read $\pi_2$ returns after server $s$ returns an acknowledgement to write $\pi_1$. From Lemma \ref{lem:returntagofread}, we conclude that $\pi_2$ returns the value  of a write whose tag that is at least as large as $tag(\pi_1).$ Consequently, it could not have return the value of $\pi.$

In case (ii), note from server protocol that $vc_{s}[\overline{s}]$ is only incremented on receiving $\texttt{Apply\_Inqueue}$ messages from server $\overline{s}.$ In particular, it reaches $vc_{s}^{P} [\overline{s}] \geq ts(\pi_1)[\overline{s}]$ implies that server $s's$ $\texttt{Apply\_Inqueue}$ action empties the $(\overline{s},X,v,tag(\pi_1)$ message from $\pi_1$ before returning read $\pi_2$. From Algorithm \ref{alg:internal_actions}, we infer that $(tag(\pi),v)$ is in $L_s[X]$ before point $P$.  Lemma \ref{lem:returntagofread} implies that $\pi_2$ returns the value  of a write whose tag that is at least as large as $tag(\pi_1).$ Consequently, it could not have return the value of $\pi.$

\end{proof}

\section{Liveness: Proof of Theorems \ref{lem:write_terminates} and \ref{lem:read_terminates_f=0}}
\label{app:liveness}

We begin with a formal proof of termination of write operations.

\begin{proof}[Proof of Theorem \ref{lem:write_terminates}]
The client $c \in \mathcal{C}_{s}$ sends a $\langle \texttt{write},X,v\rangle $ message to node $s$.  On receipt of this $\texttt{write},X,v$ from client $c$, node $s$ sends $ack$ back to client $c$ during this input action transition (Line \ref{line:writeresponse} in Algorithm \ref{alg:inputactions_clients}). Since node $s$ is a non-halting node, based on its protocol, the write operation eventually gets a corresponding $ack$ from node $s$. So every write operation terminates.
\end{proof}

The main goal of this section is to prove Theorem \ref{lem:read_terminates_f=0}. The proof relies on Lemmas \ref{lem:error1},\ref{lem:error2} and \ref{lem:livenessvalresp} stated next.

 \begin{lemma}
  At any point $L$ of an execution $\beta,$ for any server $s$ and object $X \in \mathcal{X}$:
$Error1_{s}^{L}[X] = 0$
 \label{lem:error1}
 \end{lemma}

  \begin{lemma}
  At any point $L$ of an execution $\beta,$ for any server $s$ and object $X \in \mathcal{X}$,
$Error2_s^{L}[X] = 0$
 \label{lem:error2}
 \end{lemma}

\begin{lemma}
Let $\beta$ be an execution of \CausalEC~ and $s$ be a server. Consider any point $P$ such that there exists a tuple  $(clientid, opid, X, tags,\bar{w})$ in $ReadL_s$. Suppose that node $s$ receives  $\langle val\_resp\_encoded, clientid,opid, X, tags, \overline{v}\rangle$ message from node $j$ at some point $Q$ before point $P$, we have: $ \Pi_j(\bar{w}) \neq \bot$. 
\label{lem:livenessvalresp}
\end{lemma}

 We begin with some preliminary lemmas in Section \ref{sec:live_prelim}. We then prove 
 Lemmas \ref{lem:error1}, \ref{lem:error2},\ref{lem:livenessvalresp} in the subsequent subsections, and then prove Theorem \ref{lem:read_terminates_f=0}.

\subsection{Preliminary Lemmas}
\label{sec:live_prelim}
\begin{lemma}
Let $P,Q$ be two points in an execution $\beta$ of $CausalEC$ such that $P$ comes after $Q.$ Then for any node $s,$ for any object $X$, $M_s^{P}.tagvec[X]\geq M_s^{Q}.tagvec[X].$ 
\label{lemma:tagsalwaysincrease}
\end{lemma}

\begin{proof}
From the protocol code in Algorithms \ref{alg:inputactions_clients} and \ref{alg:internal_actions}, we note that at any server $s$, the only actions that can change $M_s.tagvec[X]$ are $\texttt{Encoding}$ actions. To show the lemma, it therefore suffices to show the following claim:
\begin{claim}
Consider two consecutive points $P,Q$ of the execution, where $P$ is a point preceding an  $\texttt{Encoding}$ action, and $Q$ is the point of the  $\texttt{Encoding}$ action. Then $M_s^{Q}.tagvec[X] \geq M_s^{P}.tagvec[X].$
\end{claim}

If $X \in \mathcal{X}_{s},$ then line \ref{line:updatetag1} of Algorithm \ref{alg:internal_actions} is the only assignment that can sets $M_s^{Q}.tagvec[X] = L_s^{Q}[X].Highesttagged.tag.$  To execute line \ref{line:updatetag1}, we note that the condition in the For loop in line \ref{line:Forloopupdatetag1} needs to be satisfied, and this condition dictates that \newline $L_s^{Q}[X].Highesttagged.tag >M_s^{P}.tagvec[X]$.  Therefore, $M_s^{Q}.tagvec[X] >M_s^{P}.tagvec[X]$. 

If $X \notin \mathcal{X}_{s},$ then the argument is similar to the case of $X \in \mathcal{X}_{s},$ with the  For loop condition in \ref{line:Forloopupdatetag2} coupled with the construction of set $\overline{U}$ in line \ref{line:Delete2_Ubarset} ensuring that 
the assignment of line \ref{line:updatetag2} results in $M_s^{Q}.tagvec[X] >M_s^{P}.tagvec[X]$. 
\end{proof}

 \begin{lemma}
Let $P$ be a point in an execution $\beta$ of $CausalEC$. For any node $s$, either $M_s^{P}.tagvec[X] = \vec{0},$ or there is some point $P'$ which is no later than (and possibly equal to) $P$ in $\beta$ such that a tuple of the form $(M_s^{P}.tagvec[X],val)$ in $L[X]$ at point $P',$ that is, in $L_s^{P'}[X].$
\label{lem:pointsexec}
 \end{lemma}

 \begin{proof}
 Let $P'$ be the first point where $M_s^{P'}.tagvec[X] = M_s^{P}.tagvec[X].$ Note that $P'$ is no later than $P$ in $\beta.$ Since, $M_s.tagvec[X]$ can only be changed by an $\texttt{Encoding}$ internal action, we infer that an $\texttt{Encoding}$ was performed at $P.$ Furthermore, as a part of the action, line \ref{line:updatetag1} or line \ref{line:updatetag2} in Algorithm \ref{alg:internal_actions} was executed at $P'.$ If $X \in \mathcal{X},$ from the protocol in line \ref{line:updatetag1}, we conclude that $L_s^{P'}[X].Highesttagged.tag = M_s^{P'}.tagvec[X].$ If $X \notin \mathcal{X}$, from the construction of set $\overline{U}$ in line \ref{line:Delete2_Ubarset}, we conclude that line \ref{line:updatetag2} updates the tag only if there exists $val'$ that satisfies $(M_s^{P'}.tagvec[X], val') \in L_{s}^{P'}[X]$. 
 Furthermore, from the protocol, every element in $L[X]$ is of the form $(tag,val)$ with $val$ taking the value of the unique write whose tag is $tag$. Therefore, there exists a tuple of the form $(M_s^{P}.tagvec[X],val)$ at point $P'$ in $L[X]$ at server $s.$
 \end{proof}

 \begin{lemma}
Let $P$ be a point in an execution $\beta$ of $CausalEC$. For any node $s$, for an object $X$, let $t$ be the highest tag such that node $s$ has  sent a message of the form $\langle \texttt{del},X,t\rangle$ or added a tuple of the form $(t,s)$ to $DelL_s[X]$ before point $P$. Then $t \leq M_s^{P}.tagvec[X]$. 
\label{lem:tag-delete-ordering}
 \end{lemma}
 
 \begin{proof}
 We consider two cases: $X \in \mathcal{X}_{s}$ and $X \notin  \mathcal{X}_{s}$

\underline{Case (I): $X \in \mathcal{X}_{s}$}

Based on the protocol,  node $s$ has  sends $\langle \texttt{del},X,t\rangle$ message or adds a tuple of the form $(t,s)$ in either Line \ref{line:Delete_send1} or Line \ref{line:Delete_send3} of Algorithm \ref{alg:internal_actions}. Let $Q$ denote a point of the execution where one of these internal actions is performed. To show the lemma, it suffices to show that $M_s^P.tagvec[X] \geq M_s^Q.tagvec[X]$ for at any point $P$ that comes after $Q$.

\underline{Case (IA) - Line \ref{line:Delete_send1}:} This line is performed as a part of the \texttt{Encoding} internal action. As a part of the state changes in this internal action, in line \ref{line:updatetag1}, the node sets $M_s^Q.tagvec[X]$ to be $t$. From Lemma \ref{lemma:tagsalwaysincrease}, we conclude that at any point $P$ that comes after $Q$, we have $M_s^P.tagvec[X] \geq M_s^Q.tagvec[X] = t.$ This completes the proof.

\underline{Case (IB) - Line \ref{line:Delete_send3}:} Notice that $X \in \mathcal{X}_{s},$ which implies that $s \in R_{s}^{Q}.$
Based on the condition imposed in Line \ref{line:Delete_cond3}, the fact that line (\ref{line:Delete_send3}) was executed implies that $(t,v) \in DelL^{s}_Q[X]$. Specifically, the tuple $(t,v)$ was added at some point $Q'$ before $Q,$ and this was added as a part of line \ref{line:Delete_send1} in an $\texttt{Encoding}$ action. From the result of Case (IA), and because $P$ comes after $Q'$, we conclude that a have $M_s^{P}.tagvec[X] \geq M_s^{Q'}.tagvec[X] = t.$ This completes the proof.

\underline{Case (II): $X \notin \mathcal{X}_{s}$}
 
Based on the protocol,  node $s$ has  sends $\langle \texttt{del},X,t\rangle$ message or adds a tuple of the form $(t,s)$ to $DelL[X]$ in Line \ref{line:Delete_send2} of Algorithm \ref{alg:internal_actions} as a part of the $\texttt{Encoding}$ internal action. Let $Q$ denote the point of the execution where this internal action is performed.  As a part of the state changes in this internal action, in line \ref{line:updatetag2}, the node sets $M_s^{Q}.tagvec[X]$ to be $t$. From Lemma \ref{lemma:tagsalwaysincrease}, we conclude that at any point $P$ that comes after $Q$, we have $M_s^{P}.tagvec[X] \geq M_s^{Q}.tagvec[X] = t.$ This completes the proof.
 \end{proof}
 
 \begin{lemma}
Let $P$ be a point in an execution $\beta$ of $CausalEC$. For any node $s$ and any object $X$, suppose no tuple of the form $(M_s^{P}.tagvec[X],val)$ exists in $L_s^{P}[X].$ Then, node $s$ received a tuple of the from $\langle \texttt{del},X,M_s^{P}.tagvec[X],val\rangle$  from every node in the system by some point before $P$. 
\label{lem:currenttag-delete}
 \end{lemma}

 \begin{proof}
 From Lemma \ref{lem:pointsexec}, we know that there is a point $P'$ at or before $P$ such that tuple of the form $(M_s^{P}.tagvec[X],val)$ exists in $L_s^{P'}[X].$ Let $Q$ be the first point after $P'$ such that this tuple does not exist in $L_s^{P'}[X].$ By the hypothesis of the lemma $Q$ can be no later than $P$. From the protocol, since the only action that removes objects from $L_s^{P}[X]$ is a $\texttt{Garbage\_Collection}$ action, $\texttt{Garbage\_Collection}$ action took place at $Q$, and the tuple was removed from $L_s[X]$ as a part of line \ref{line:GC1} or \ref{line:GC3} or  \ref{line:GC2} of Algorithm \ref{alg:internal_actions}. From Lemma \ref{lem:tag-delete-ordering}, we infer that for any tag $t > M_s^{Q}.tagvec[X],$ the element $\{(t,s)\}$ does not belong to $Del_s^{Q}[X].$ Therefore, $tmax_s^{Q}[X] \stackrel{(a)}{\leq} M_s^Q.tagvec[X].$ Furthermore,   \newline $M_s^Q.tagvec[X] \stackrel{(b)}{\leq} M_s^P.tagvec[X]$ because $Q$ is no later than $P$. Since lines \ref{line:GC1}, \ref{line:GC3}, \ref{line:GC2} can only remove elements from $L_s^{Q}[X]$ with tags no bigger than $tmax_s^{Q}[X]$, we can conclude that:
 \begin{itemize}
     \item Inequalities $(a),(b)$ are in fact met with equality, which further implies that $t{max}_s^{Q}[X] = M_s^P.tagvec[X]$, and 
     \item Line \ref{line:GC1} is executed at point Q.
 \end{itemize}

 Since $t{max}_s^{Q}[X] = M_s^{P}.tagvec[X]$ and line \ref{line:GC1} was executed we conclude that  $M_s^{P}.tagvec[X] \in \overline{S}_{s}^{Q},$ where $\overline{S}_{s}^{Q}$ is formed in line \ref{line:s1set}. The definition of line \ref{line:s1set} implies that at $Q,$ $\{i:(M_s^{P}.tagvec[X],i) 
 \in DelL[X]$ is equal to $\mathcal{N}.$ Since tuples are only added to  $DelL[X]$ on receipt of $\langle \texttt{del},X,t,val\rangle$ messages in Line \ref{line:delmessage} in Algorithm \ref{alg:input_actions}, we conclude that node $s$ received a tuple of the from $\langle \texttt{del},X,M_s^{P}.tagvec[X],val\rangle$  from every node in the system by some point before $Q$, and hence before $P$. 
 \end{proof}

 \begin{lemma}
 Let $P$ be a point of an execution $\beta$ such that for a node $s$ at point $P$, there exists $val$ such that there is tuple $(M_s^{P}.tagvec[X], val)$ in $L_s^{P}[X]$ for some object $X$. Let $Q$ be a point in $\beta$ at or after $P$ such that a tuple $(clientid,opid, X,tvec,i,\overline{w})$ exists in $ReadL_s^{Q}$ where $tvec[X]=M_s^{P}.tagsvec[X].$ If $M_s^{Q}.tagvec[X] >M_s^{P}.tagvec[X],$ then there exists a tuple $(M_s^{P}.tagvec[X], val)$ in $L_s^{Q}[X].$
 \label{lem:listsavespendingread}
 \end{lemma}
 \begin{proof}
Let $P'$ be the first point after $P$ such that $M_s^{P'}.tagvec[X] > M_s^{P}.tagvec[X]$. Notice that $P'$ is no later than $Q.$ The proof involves three claims: Claim \ref{claim:l1}, \ref{claim:l2} and \ref{claim:l3}.

\begin{claim}
There exists a tuple $(M_s^{P}.tagvec[X], val)$ in $L_{s}^{P'}[X].$
\label{claim:l1}
\end{claim}
 \begin{proof}[Proof of Claim \ref{claim:l1}]
At the point $P''$ which is immediately before $P',$ we have $M_s^{P''}.tagvec[X]  = M_s^{P}.tagvec[X]$, because $P'$ be the first point after $P$ such that $M_s^{P'}.tagvec[X] > M_s^{P}.tagvec[X]$.
 From the protocol, it readily follows that an $\texttt{Encoding}$ action took place at $P',$ and line \ref{line:updatetag1} or line \ref{line:updatetag2} of Algorithm \ref{alg:internal_actions}  was executed. Since the condition in line \ref{line:conditionforupdate} or \ref{line:Delete2_condition} returns true for the execution of line \ref{line:updatetag1} or line \ref{line:updatetag2}, we conclude that a tuple of the form $(M_s^{P}.tagvec[X], val)$ exists in $L_s^{P'}[X].$
 \end{proof}
 \begin{claim}
  The tuple $(clientid,opid, X,tvec,i,\overline{w})$ in $ReadL_s^{Q}$ with $tvec[X]=M_s^{P}.tagvec[X]$ is added to $ReadL_s$ at some point before $P'.$
  \label{claim:l2}
 \end{claim}
 \begin{proof}[Proof of Claim \ref{claim:l2}]
 Consider any tuple \newline $(clientid,opid, X,tvec,i,\overline{w})$ in $ReadL_s^{Q}$ with $tvec[X]=M_s^{P}.tagvec[X].$

At $P',$ we have $M_s^{P'}.tagvec[X] > M_s^{P}.tagvec[X].$ Because of Lemma \ref{lemma:tagsalwaysincrease}, at any point $R$ after $P',$ $M_s^{R}.tagvec[X] > M_s^{P}.tagvec[X].$ From the protocol, the tuple is added to $ReadL$ during execution of line \ref{line:addtoreadl} in Algorithm \ref{alg:inputactions_clients} or \ref{line:readlentryencoding} in Algorithm \ref{alg:internal_actions}. If the tuple is added at point $Q'$, by examining both lines, we have $tvec[X]=M_s^{Q'}.tagvec[X].$ However, $M_s^{Q'}.tagvec[X] = M_s^{P}.tagvec[X].$ Because for any point $R$ at or after $P'$, we have $M_s^{R}.tagvec[X] > M_s^{P}.tagvec[X],$ we conclude that $Q'$ must be before $P'.$
 \end{proof}
 \begin{claim}
If a \texttt{Garbage\_Collection} action that takes place between $P'$ and $Q,$ the action does not remove the tuple $(M_s^{P}.tagvec[X], val)$ from $L[X].$
\label{claim:l3}
 \end{claim}
\begin{proof}[Proof of Claim \ref{claim:l3}]
Suppose a \texttt{Garbage\_Collection} is performed at point $Q''$ that is between $P'$ and $Q.$ Because $Q''$ comes after $P'$, we have $M_s^{Q''}.tagvec[X] >  M_s^{P}.tagvec[X]$. Because of Claim \ref{claim:l2}, the tuple of the form $(clientid,opid X,tvec,i,v)$ with $tvec[X] = M_s^{P}.tagvec[X]$ in $ReadL_s^{Q}$ that is noted in the hypothesis of the lemma is in $ReadL_{s}$ at every point between $P'$ and $Q$, specifically, it is in $ReadL_{s}^{Q''}.$ The \texttt{Garbage\_Collection} action removes elements from $L[X]$ for $X \in \mathcal{X}_{s}$ only if it executes line \ref{line:GC1}, line \ref{line:GC3} or line \ref{line:GC2}. However these lines do not remove elements from $T$ which is found in Line \ref{line:pendingreads}. Because $ReadL_s^{Q''}$ has a tuple of the form $(clientid,opid, X,tvec,i,v)$ with $tvec[X] = M_s^{P}.tagvec[X]$ and because $M_s^{P}.tagvec[X] < M_s^{Q''}.tagvec[X],$ we conclude $tagvec[X]$ is in the set $T$. Therefore, the \texttt{Garbage\_collection} action at point $Q''$ does not remove the tuple $(M_{s}^{P}.tagvec[X], val)$ from $L[X].$ 
\end{proof}
 \end{proof}
 
\begin{lemma}
 Consider a node $s$ and an object $X \notin \mathcal{X}_{s}.$  Let $P$ be a point of an execution $\beta$ such that $M_s^{P}.tagvec[X] > \mathbf{0}.$  Then, for any server $s'$ such that $X \in \mathcal{X}_{s'},$ node $s$ has received a message $\langle \texttt{del}, X,t)\rangle$ from server $s'$ before point $P$ for some $t \geq M_{s}^{P}.tagvec[X]$
 \label{lem:handlenonpresentobjects}
 \end{lemma}

  \begin{proof}
  
  For $X \notin \mathcal{X}_{s},$ $M_s.tagvec[X]$ is updated in Line \ref{line:updatetag2} in Algorithm \ref{alg:internal_actions}. From the condition to execute  Line \ref{line:updatetag2} in Line \ref{line:Delete2_condition}, we infer that node $s$ received a $\langle \texttt{del},X,t\rangle$  with $t \geq M_{s}^{P}.tagvec[X]$ no later than point $P$ from all nodes in $\{i: X \in \mathcal{X}_{i}\}$.  Specifically it received the message from node $s'$ no later than $P.$
  \end{proof}
  
\begin{lemma}
 Consider nodes $s,s'$ and an object $X$ such that $X \notin \mathcal{X}_{s}, X \in \mathcal{X}_{s'}.$ At any point $P$ be a point of an execution $\beta,$ we have $M_{s}^{P}.tagvec[X] \leq M_{s'}^{P}.tagvec[X].$
 \label{lem:nonpresentobjectshavesmallertags}
 \end{lemma}

 \begin{proof} If $M_s^P.tagvec[X]=\mathbf{0},$ the lemma is trivially true. Otherwise, 
 the hypothesis of Lemma \ref{lem:handlenonpresentobjects} holds, and from the lemma, we infer that node $s'$ sent a  $\langle \texttt{del},X,t\rangle$ before $P$ with $t \geq M_{s}^{P}.tagvec[X]$. Lemma \ref{lem:tag-delete-ordering}   therefore implies that   $M_{s'}^{P}.tagvec[X] \geq M_s^{P}.tagvec[X].$
 \end{proof}

\begin{lemma}
 Let $m$ be a message sent from $s$ and delivered to node $s'$ in an execution $\beta$. Let $P$ be the point of sending of message $m$  from $s$ and $P'$ be the point of receipt at $s'$. For an object $X \in \mathcal{X}_{s'},$ if $M_s^{P}.tagvec[X] < M_{s'}^{P'}.tagvec[X]$ and if there is a point $Q$ before $P'$ such that node $s'$  sends a $\langle \texttt{del}, X, M_s^{P}.tagvec[X] \rangle$ message to any node in the system at point $Q$, then there there exists a tuple of the form $(M_s^{P}.tagvec[X],val)$ in list $L_{s'}^{P'}[X]$,

 \label{lem:objectpresentfordecoding1}
 \end{lemma}

 \begin{proof}
 Based on the hypothesis of the lemma, node $s'$ sent a message of the form \newline $\langle \texttt{del},X,v,M_s^{P}.tagvec[X]\rangle$ at some point $Q$ that is before $P'$. Based on the protocol, at point $Q$, we have $M_{s'}^{Q}.tagvec[X] = M_s^P.tagvec[X].$ Furthermore, at point $P'$ which is after $Q,$ $M_{s'}^{P'}.tagvec[X] > M_{s'}^{Q}.tagvec[X].$ Let $Q'$ be the first point after $Q$ at which $M_{s'}^{Q'}.tagvec[X] > M_{s}^{P}.tagvec[X].$

Based on the protocol, because $X \in \mathcal{X}_{s'}$, we infer that an $\texttt{Encode}$ action takes place at $Q'$. Based on the $\texttt{Encode}$ action steps, at $Q'$, a tuple of the form $({M}_{s}^{P}.tagvec[X],val)$ exists in $L_{s'}^{Q'}[X].$ We claim that node $s'$ does not delete this tuple from $L_{s'}[X]$ between $Q'$ and $P'.$ To show the claim, observe that elements are deleted from $L[X]$ at $s'$  only because of the $\texttt{Garbage\_Collection}$ internal action at node $s'$. We argue that no $\texttt{Garbage\_Collection}$ performed by node $s'$ between $Q'$ and $P'$ deletes the tuple.

Consider a $\texttt{Garbage\_Collection}$ action performed by node $s'$ at point $Q''$ that is between $Q'$ and $P'$. 
Because $Q''$ is after $Q'$, from Lemma \ref{lemma:tagsalwaysincrease}, we note that $M_{s'}^{Q''}.tagvec[X] \geq M_{s'}^{Q'}.tagvec[X]  > {M}_{s}^{P}.tagvec[X]$.  Because of Lemma \ref{lem:tag-delete-ordering}, node $s$ does not send a $\langle\texttt{Del},X,tag\rangle$ for some $tag$ larger than $M_{s}^{P}.tagvec[X]$ from node $s$ by point $P$. Because of the FIFO nature of the channel, node $s'$ does not receive a $\langle\texttt{Del},X,tag\rangle$ for some $tag$ larger than $M_{s}^{P}.tagvec[X]$ from node $s$ by point $P'$. Therefore, as a part of the $\texttt{Garbage\_Collection}$ internal action performed by node $s$ at t $Q''$, we have $t_{max}^{Q''}[X] \leq M_s^{P}.tagvec[X].$ Because $M_{s'}^{Q''}.tagvec[X] > M_{s}^{P}.tagvec[X],$ we have $t_{max,s'}^{Q''}[X] < M_{s'}^{Q''}.tagvec[X].$ Because $X \in \mathcal{X}_{s'}$ and $t_{max,s'}^{Q''}[X] < M_{s'}^{Q''}.tagvec[X]$, lines \ref{line:GCcondition1} and \ref{line:GCcondition3} are not true for object $X$. the $\texttt{Garbage\_Collection}$ action performed at $Q''$ executes line \ref{line:GC2}, which only deletes tags that are strictly smaller than $t{max}[X]$ in Line \ref{line:GC2} in Algorithm \ref{alg:internal_actions}.  As $tmax_{s}^{Q''}[X] \leq M_s^{P}.tagvec[X],$ the tuple of the form  $(M_s^{P}.tagvec[X], val)$ is not deleted by the \newline $\texttt{Garbage\_Collection}$ action at $Q''.$ 

 \end{proof}

 \begin{lemma}
 Consider two nodes $s,s'$, and consider some message $m$ sent from $s$ to $s'$ in an execution $\beta$. Let $P$ be the point of sending of message $m$  from $s$ and $P'$ be the point of receipt. If $M_s^{P}.tagvec[X] \neq M_{s'}^{P'}.tagvec[X],$ then at least one of the following statements is true:
 \begin{itemize}
     \item There exists a tuple of the form $(M_{s}^{P}.tagvec[X],val)$ in list $L_{s}^{P}[X]$, or
     \item $M_s^{P}.tagvec[X] < M_{s'}^{P'}.tagvec[X]$ and $X \notin \mathcal{X}_{s'},$ or
    
     \item $M_s^{P}.tagvec[X] < M_{s'}^{P'}.tagvec[X]$ and there exists a tuple of the form $(M_{s}^{P}.tagvec[X],val)$ in list $L_{s'}^{P'}[X]$.
     
 \end{itemize}
 \label{lem:objectpresentfordecoding}
 \end{lemma}

 \begin{proof}
Consider the hypothesis of the lemma.
Because \newline $M_{s}^{P}.tagvec[X] \neq M_{s'}^{P'}.tagvec[X],$ and because tags are comparable, we infer that either (I): $M_{s}^{P}.tagvec[X] > M_{s'}^{P'}.tagvec[X],$ or 
(II) $M_{s}^{P}.tagvec[X] < M_{s'}^{P'}.tagvec[X].$ We consider cases (I) and (II) separately and prove that the lemma statement holds for both cases.

\underline{Case (I): :$M_{s}^{P}.tagvec[X] > M_{s'}^{P'}.tagvec[X].$} From Lemma \ref{lem:pointsexec}, node $s$ consists of a tuple of the form $(M_{s}^{P}.tagvec[X],val)$ in its list $L^{s}[X]$ for some $val \in \mathcal{V}$ at some point before point $P$. 
We argue that this tuple exists in the list at point $P$ as well. 

Lemma \ref{lem:tag-delete-ordering} implies that every message sent by node $s'$ before $P'$ of the form  $\langle \texttt{del},X,v,t\rangle$ has $t \leq M_{s'}^{P'}.tagvec[X].$ Because \newline $M_{s'}^{P'}.tagvec[X]^{P'} < M_{s}^{P}.tagvec[X]$ and because $P$ comes before $P'$, we conclude that node $s'$ does not send a message of the form $\langle \texttt{del},X,v,t\rangle$ where $t \geq M_{s}^{P}.tagvec[X]$ before point $P$ to any node. Therefore, node $s$ does not receive a $\langle \texttt{del},X,v,t\rangle$ message $t \geq M_{s}^{P}.tagvec[X]$ from node $s'$ before point $P$. Therefore, for any point $Q$ that is not later than point $P,$ we have $t_{max,s}^{Q}[X] < M_{s}^{P}.tagvec[X]$. 

We conclude that $s$ contains an item of the form \newline $(M.tagvec[X],val)$ in $L[X]$ at point $P$. This is because (i) a tuple of the form $(M_{s}^{P}.tagvec[X],val)$ is exists in $L_{s}[X]$ at node $s$ at some point before $P$, (ii) based on the protocol, this tuple is not deleted by node $s$ before a point $R$ such that $t_{max,s}^{R}[X] \geq M_s^P.tagvec[X]$ and (iii) at every point $Q$ that is no later than $P$, we have $t_{max,s}^Q[X] < M_s^P.tagvec[X]$.

\underline{Case (II):~$M_{s}^{P}.tagvec[X] < M_{s'}^{P'}.tagvec[X].$} If a tuple of the form $(M.tagvec_s^P[X],val)$ exists in $L_s^P[X]$ or if $X \notin \mathcal{X}_{s'}$ then the statement of the lemma is true. So it suffices to prove the lemma statement for the scenario where a tuple of the form $(M_s^P.tagvec[X],val)$ does not exist in $L_s^P[X]$ for $X \in \mathcal{X}_{s'}.$ In this case, the hypothesis of Lemma \ref{lem:currenttag-delete} holds and therefore, we infer that node $s'$ sent a message of the form $\langle \texttt{del},X,v,M_s^P.tagvec[X]\rangle$ at some point $Q$ that is before $P$. Because  $M_s^P.tagvec[X] < M_{s'}^{P'}.tagvec[X]$ and $X \in \mathcal{X}_{s'}$ the hypothesis of Lemma \ref{lem:objectpresentfordecoding1} holds. Therefore, we conclude that there exists a tuple of the form $(M_s^P.tagvec[X], val)$ in $L_{s'}^{P'}[X]$. 
 \end{proof}

 \begin{lemma}
 Consider two nodes $s,s'$, and consider some message $m$ sent from $s$ to $s'$ in an execution $\beta$. Let $P$ be the point of sending of message $m$  from $s$ and $P'$ be the point of receipt. Consider an object $X \in \mathcal{X}_{s'}.$ If $M_{s}^{P}.tagvec[X] < M_{s'}^{P'}.tagvec[X],$ then there exists a tuple of the form $(M_{s'}^{P'}.tagvec[X],val)$ in list $L_{s'}^{P'}[X].$
     
 \label{lem:objectpresentfordecoding2}
 \end{lemma}

 \begin{proof}
 From Lemma \ref{lem:pointsexec}, we infer that there exists a point $Q$ before $P'$ such that a tuple of the form $(M_{s'}^{P'}.tagvec[X],val)$ exists in list $L_{s'}^{P'}[X].$ Since  $\texttt{Garbage\_Collection}$ internal actions are the only actions that remove entries from list $L_s[X],$ it suffices to show that no  $\texttt{Garbage\_Collection}$ action that occurs between $Q$ and $P'$ removes the element $(M_{s'}^{P'}.tagvec[X],val)$ from the list $L_s[X].$ We show this next.
 
 From Lemma \ref{lem:tag-delete-ordering}, we infer that node $s$ does not send a \newline $\langle \texttt{Del},X,t,val \rangle$ message before point $P$ for any $t > M_s^P.tagvec[X].$ Because of the FIFO nature of the channel, node $s'$ does not receive a $\langle \texttt{Del},X,t,val \rangle$ from node $s'$ for  any $t > M_{s}^{P}.tagvec[X]$ before point $P'$. Thus, for any   $\texttt{Garbage\_Collection}$ that takes place before $P',$ the variable $tmax_{s}[X]$ in line \ref{line:tmax} is no bigger than $M_{s}^{P}.tagvec[X].$ Therefore, the state changes that take place in lines \ref{line:GC1} and \ref{line:GC2} of Algorithm \ref{alg:internal_actions} does not remove the element $(M_{s'}^{P'}.tagvec[X],val)$ from the list $L_s[X].$  \end{proof}

 \begin{lemma}
 Consider two nodes $s,s'$, and consider some message $m$ sent from $s$ to $s'$ in an execution $\beta$. Let $P$ be the point of sending of message $m$  from $s$ and $P'$ be the point of receipt. Consider an object $X$ such that $X \notin \mathcal{X}_{s}, X \in \mathcal{X}_{s'}.$ If $M_s^P.tagvec[X] \neq M_{s'}^{P'}.tagvec[X],$ then there exists a tuple of the form $(M_{s'}^{P'}.tagvec[X],val)$ in list $L_{s'}^{P'}[X].$    
 \label{lem:objectpresentfordecoding3}
 \end{lemma}
 \begin{proof}
Because $X \in \mathcal{X}_{s'}, X \notin \mathcal{X}_{s},$ from Lemma \ref{lem:nonpresentobjectshavesmallertags}, we have $M_{s}^{P}.tagvec[X] \stackrel{(a)}{\leq} M_{s'}^{P}.tagvec[X].$ Because $P'$ comes after $P$, from Lemma \ref{lemma:tagsalwaysincrease}, we have $M_{s'}^{P}.tagvec[X] \stackrel{(b)}{\leq} M_{s'}^{P'}.tagvec[X].$ The lemma hypothesis that $M_{s}^{P}.tagvec[X] \neq M_{s'}^{P'}.tagvec[X]$ implies that at least one of the inequalities $(a),(b)$ is strict and $M_{s}^{P}.tagvec[X] < M_{s'}^{P'}.tagvec[X].$ Thus, the hypothesis of Lemma \ref{lem:objectpresentfordecoding2} is satisfied and its conclusion implies that there exists a tuple of the form $(M_{s'}^{P'}.tagvec[X],val)$ in list $L_{s'}^{P'}[X].$
 \end{proof}

\subsection{Proof of Lemma \ref{lem:error1}}

 \begin{proof}
 Since the only action that modifies $Error1$ variable is the receipt of a $\texttt{val\_resp\_encoded}$ message, it suffices to consider such a point. Every  $\texttt{val\_resp\_encoded}$ message is sent in response to a $\texttt{val\_inquiry}$ message. Let point $P$ denote the point of sending of the $\langle \texttt{val\_inq},clientid, opid, \overline{X},tvec\rangle$ from node $s$, let $Q$ be the point of receipt of this message at node $s'$ and let $R$ denote the point of receipt of $\langle \texttt{val\_resp\_encoded},\overline{M},clientid,opid,\overline{X},reqtvec\rangle$ message at node $s.$ From the protocol, we conclude that the point of sending of the $\texttt{val\_resp\_encoded}$ from node $s'$ is is also $Q$. Also note that $tvec = M_s^{P}.tagvec$.

 It suffices to consider every object $X \in \mathcal{X}_{s'}$ and argue that on receipt of $\texttt{val\_resp\_encoded}$ message, the action performed by node $s$ does not set $Error1[X] \leftarrow 1.$  Note that $Error1[X]$ is not set to $1$ if $\overline{M}.tagvec[X] = {M}_{s}^P.tagvec[X].$ So it suffices to assume that $\overline{M}.tagvec[X] \neq {M}_{s}^{P}.tagvec[X]$ for the remainder of the proof.
 
 \begin{claim}
 If there exists an element $(M_{s'}^{Q}.tagvec[X],val)$ in $L_{s'}^{Q'}[X],$ then we have $\overline{M}.tagvec[X] \in \{\mathbf{0}, M_{s}^{P}.tagvec[X]\}$. Furthermore, if an element $M_{s}^{P}.tagvec[X]$ also exists in $L_{s'}^{Q'}[X],$ then  $$\overline{M}.tagvec[X] = M_{s}^{P}.tagvec[X].$$ 
 \label{claim:error1vals}
 \end{claim}
 \begin{proof}
If there exists an element $(M_{s'}^{Q}.tagvec[X],val)$, then at point $Q$ node $s'$ the condition in \ref{line:responsetovalinqcondition1} in Algorithm \ref{alg:input_actions} returns true. Therefore the node execute line \ref{line:ResponsetoValInqtagupdate1} which sets $\overline{M}.tagvec[X] = \mathbf{0}$. Further, if  $M_{s}^{P}.tagvec[X]$ also exists in $L_{s'}^{Q'}[X],$ then the condition in \ref{line:responsetovalinqcondition2} returns true, and line \ref{line:ResponsetoValInqtagupdate2} is executed. This line sets $\overline{M}.tagvec[X] = M_{s}^{P}.tagvec[X]$
 \end{proof}
 
 \begin{claim}  There exists $val \in \mathcal{V}$ such that\\ (I) $(M_{s'}^{Q'}.tagvec[X],val) \in L_{s'}^{Q}[X]$ or\\ (II) $(M_{s'}^{Q'}.tagvec[X],val) \in L_{s}^{R}[X].$ 
 \end{claim}
 \begin{proof}[Proof of Claim:] If $X \in \mathcal{X}_{s},$ then the $\texttt{val\_resp\_encoded}$ satisfies the hypothesis of Lemma  \ref{lem:objectpresentfordecoding},  whose statement implies either (I) $(M_{s'}^{Q}.tagvec[X],val) \in L_{s'}^{Q'}[X]$ or (II) $(M_{s'}^{Q}.tagvec[X],val) \in L_{s}^{R}[X] .$ If, on the other hand, $X \notin \mathcal{X}_{s},$ then the $\texttt{val\_inq}$ message from $s$ to $s'$ satisfies the hypothesis of Lemma \ref{lem:objectpresentfordecoding3}, which implies that (I) $(M_{s'}^{Q'}.tagvec[X],val)\in L_{s'}^{Q'}[X].$
 \end{proof}

 We consider cases (I) and (II) separately.
 
 \underline{Case (I): $(M_{s'}^{Q}.tagvec[X],val) \in L_{s'}^{Q}[X]$}
From, Lemma \ref{claim:error1vals} we know that $\bar{M}.tagvec[X] \stackrel{(a)}{=} \mathbf{0},$ or $\bar{M}.tagvec[X] \stackrel{(b)}{=} M_s^{P}.tagvec[X].$ 
  In both cases (a) and (b), based on the protocol, we conclude that in line \ref{line:seterror1} of Algorithm \ref{alg:input_actions} does not return true and $Error1[X]$ is not set to $1$ .
  
  \underline{Case (II):} $$\underline{(M_{s'}^{Q}.tagvec[X],val) \notin L_{s'}^{Q}[X],}$$ $$\underline{(M_{s'}^{Q}.tagvec[X],val) \in L_{s}^{R}[X]}$$
  
  From the protocol, in this case,$\bar{M}.tagvec[X] = M_{s'}^{Q}.tagvec[X].$ On receipt of this message at server $s$ at point $R$, because \newline $(M_{s'}^{Q}.tagvec[X],val) \in L[X]^{s,R},$ we note that that the condition in line \ref{line:metaerrorcondition1} holds. Therefore $Error1[X]
  $ is not set to $1$ at point $R$. 
 \end{proof}
 
\subsection{Proof of Lemma \ref{lem:error2}}

Since the only action that modifies $Error2$ variable is the receipt of a $\texttt{val\_resp\_encoded}$ message, it suffices to consider the point of receipt an arbitrary $\texttt{val\_resp\_encoded}$ message and show the the corresponding action does not set $Error2$ to $1$. Every  \newline $\texttt{val\_resp\_encoded}$ message is sent in response to a $\texttt{val\_inquiry}$ message. Let point $P$ denote the point of sending of the $$\langle \texttt{val\_inq},\overline{M},clientid, opid,\overline{X},tvec\rangle$$ message from node $s$, let $Q$ be the point of receipt of this message at node $s'$ and let $R$ denote the point of receipt
of the $$\langle \texttt{val\_resp\_encoded},\overline{M},clientid,opid,\overline{X},reqtvec\rangle$$ message at node $s.$ From the protocol, we conclude that the point of sending of the $\texttt{val\_resp\_encoded}$ from node $s'$ is is also $Q$. Also note that $tvec = M_s^{P}.tagvec$.

It suffices to take every object $X \in \mathcal{X}_{s'}$ such that $M_{s}^{P}.tagvec[X] \neq {M}_{s'}^{Q}.tagvec[X]$ and show that the for loop in line \ref{line:seterror2} in Algorithm \ref{alg:input_actions} does not set $Error2_s^{L}[X] \leftarrow 1.$ Consider object $X \in \mathcal{X}_{s'}.$ Because, $M_{s}^{P}.tagvec[X] \neq  M_{s'}^{Q}.tagvec[X],$ Cases (I), (II) below are exhaustive.
 
 \underline{Case (I): $ M_{s}^{P}.tagvec[X]< M_{s'}^{Q}.tagvec[X].$}  Because a message $m = \texttt{val\_inquiry}$ is sent at point $P$ from node $s$ and arrives at point $Q$ at node $s'$ and because $X \in \mathcal{X}_{s'}$, from Lemma \ref{lem:objectpresentfordecoding2}, we note that there exists an element of the form $(val2,M_{s'}^{Q}.tagvec[X]) \in L_{s'}^{Q}[X].$ Therefore, from Claim \ref{claim:error1vals} in Lemma \ref{lem:error1}'s proof, we have $\overline{M}.tagvec[X] \in \{\mathbf{0},M_{s}^{P}.tagvec[X]\}.$  From Lemma  \ref{lem:objectpresentfordecoding} applied to the $\texttt{val\_inq}$ message, we know that there exists $val1$ such that either (I-A) $(M_{s}^{P}.tagvec[X],val1) \in L_{s'}^{Q}[X]$ or (I-B) $(M_{s}^{P}.tagvec[X],val1) \in L_{s}^{P}[X] .$ We consider cases (I) and (I-B) separately.
 
 \underline{Case (I-A): $(M_{s}^{P}.tagvec[X],val1) \in L_{s'}^{Q}[X]$}
 In this case, because there exists an element of the form $(val2,M_{s}^{Q}.tagvec[X]) \in L_{s}^{Q}[X],$ we conclude that $\overline{M}.tagvec[X] = M_s^{P}.tagvec[X].$ From the protocol, the condition in line \ref{line:metaerrorcondition} in Algorithm \ref{alg:input_actions} is not satisfied; therefore $Error2[X]$ is not set to 1 at point $R$.
 
  \underline{Case (I-B): }$$\underline{(M_{s}^{P}.tagvec[X],val1) \in L_{s}^{P}[X],}$$ $$\underline{(M_{s}^{P}.tagvec[X],val1) \notin L_{s'}^{Q}[X]}$$
 
 At point $R$, there are two possibilities (i) there exists a tuple of the form $$(clientid,opid,X,tvec,M_s^{P}.tagvec,\overline{v})$$ exists in $ReadL_{s}^{R}$ with $tvec[X] = M_{s}^{P}.tagvec[X],$ or (ii) there exists no such tuple. Based on the protocol, in case $(ii)$ , $Error2[X]$ is not set to $1$ at point $R.$ So we only consider case $(i)$ here. The point $R$ satisfies the hypothesis of Lemma \ref{lem:listsavespendingread}. From the lemma statement, there are only two possibilities:   $M_{s}^{P}.tagvec[X] \in L_{s}^{R}[X]$ or  $M_s^{P}.tagvec[X] = M_s^{R}.tagvec[X].$ If 
 $M_{s}^{P}.tagvec[X] \in L_{s}^{R}[X]$, coupled with Lemma \ref{lem:error1}, we conclude that the line condition in line \ref{line:metaerrorcondition2} returns true; therefore condition  $Error2[X]$ is not set to $1$. We consider the case where $M_{s}^{P}.tagvec[X] = M_{s}^{R}.tagvec[X].$
 
 \allowdisplaybreaks{
 Because we are considering case (I)-B, $(M.tagvec_s^{P}[X],val)$ does not belong to $L_{s'}^{Q}[X]$. In this case, we note that node $s'$ does not send a $\langle del, X, M_s^{P}.tagvec[X]\rangle$ at any point before $Q.$ To see this consider the contradictory hypothesis that the node sent such a message. Then combined with the fact that $M_{s}^{P}.tagvec[X] < M_{s}^{Q}.tagvec[X], X \in \mathcal{X}_{s'}$ we infer from Lemma \ref{lem:objectpresentfordecoding1} that  \newline $(M_{s}^{P}.tagvec[X],val)$ belongs to $L_{s'}^{Q}[X]$, which violates the hypothesis of Case (I-B). By the FIFO nature of the channel, we infer that a message $\langle \texttt{del}, X, M_{s}^{P}.tagvec[X]\rangle$ is not received by node $s$ from node $s'$ by point $R$. From the contrapositive of Lemma \ref{lem:currenttag-delete}, and because $M_{s}^{P}.tagvec[X] = M_{s}^{R}.tagvec[X],$ we infer that a tuple of the form $(M_{s}^{P}.tagvec[X],val)$ exists in $L_{s}^{R}[X].$
  
 Thus, we have shown that a tuple of the form $(M_{s}^{P}.tagvec[X],val)$ belongs to  $L_{s}^{R}[X]$ for every object $X \in \mathcal{X}_{s'}.$ Combined with Lemma \ref{lem:error1}, we infer that the condition in line \ref{line:metaerrorcondition2} is true, and therefore $Error2[X]$ is not set to $0.$
 }

  \underline{Case (II): $M_{s}^{P}.tagvec[X] > M_{s'}^{Q}.tagvec[X].$}
  From Lemma \ref{lem:pointsexec}, we observe that there is a point $P'$ before $P$ such that a tuple of the form $(M_s^P.tagvec[X],val)$ belongs to $L_{s}^{P'}[X].$ We show that this tuple is not removed from $L[X]$ before point $R.$   To show this, it suffices to show that for any $\texttt{Garbage\_Collection}$ action performed by node $s$ at point $R'$ which is between $P'$ and $R$, the element  $(M_{s}^{P}.tagvec[X],val)$ is not removed from the list $L_{s}^{R'}[X].$   From Lemma \ref{lem:tag-delete-ordering}, we infer that the highest tag $t$ such that node $s'$ has sent a message of the form $\langle \texttt{del},X, t\rangle$ satisfies $t \leq M_{s'}^{Q}.tagvec[X] < M_{s}^{P}.tagvec[X].$ Therefore, node $s$ does not receive a tuple of the form $\langle \texttt{del},X, t\rangle$ with $t \geq M_{s}^{P}.tagvec[X]$ from node $s$ before point $R.$ Therefore, $R$, $tmax_{s}^{R'}[X] < M_{s}^{P}.tagvec[X]$.  Therefore, from the protocol lines \ref{line:GC1}, \ref{line:GC3}, \ref{line:GC2} in Algorithm \ref{alg:internal_actions}, we infer that the element $(M_{s}^{P}.tagvec[X],val)$ is not removed from the list $L_s^{R'}[X].$ This completes the proof.

  \subsection{Proof of Lemma \ref{lem:livenessvalresp}}
 
  \begin{proof}
We begin with the following claim:
\begin{claim}
If there is a point $R$ after $P$ where a tuple $$(\overline{clientid},opid, \overline{X}, \overline{tags},\bar{v})$$ does not exist in $ReadL_s^R$, then no tuple with the same operation id exists in $ReadL_s^{R}$ for any point after $R$ in $\beta.$
\label{claim:liveness1}
\end{claim}
\begin{proof}[Sketch:]
The claim follows from noting that the protocol adds a tuple to $ReadL_s$ at most once for a given operation. Since every operation has a unique $opid,$ once a tuple with operation identifier $opid$ is removed, no tuple with the same identifier can be added to $ReadL_s.$ 
\end{proof}
Based on the protocol, on receiving a $$\langle \texttt{val\_resp\_encoded}, clientid, opid, X, tags, \overline{v}\rangle$$ message,  server $s$ checks if an entry exists in $ReadL_s$ with operation identifier $opid$ in line \ref{line:valrespencoded_readcheck} in Algorithm \ref{alg:input_actions}. From the protocol, node $j$ sends $\texttt{val\_resp\_encoded}$ message in response to a $val\_inq$ message, which in turn is sent by node $s$ when an entry with operation identifier $opid$ is made in $ReadL_s.$ Therefore, an entry $(clientid, opid, X, tags,\bar{w})$ is added to $ReadL_s$ before point $Q$. Since the lemma hypothesis indicates a tuple with operation identifier $opid$ exists at $P$ which is after point $Q$, Claim \ref{claim:liveness1} implies that the tuple is in  $ReadL_s^Q$. Therefore, line \ref{line:valrespencoded_readcheck} returns true at point $Q$. Because of Lemmas \ref{lem:error1}, \ref{lem:error2}, the condition in line \ref{line:errorcondition} also returns true. Therefore, line \ref{line:valrespencoded_addtuple} is executed at $Q$. Based on this line, if $(clientid, opid, X, tags,\bar{u})$ is the tuple in $ReadL_s^Q,$ we infer that $\Pi_j(\overline{u}) \neq \bot.$

Finally, from the protocol, we note that for any entry \newline $(clientid, opid, X, tags,\bar{z})$ in $ReadL_s^R$ at any point $R$ after $Q$, we have  $\Pi_i(\bar{u}) \neq \bot \Rightarrow \Pi_i(\bar{z}) \neq \bot.$ This is because the protocol does not replace codeword entries by null entries for a given $opid$ entry in $ReadL.$ 
\end{proof}

\subsection{Proof of Theorem \ref{lem:read_terminates_f=0}}

\begin{proof}
Let $\beta$ be a fair execution of \CausalEC~ satisfying the lemma hypotheses. Consider a read operation $\pi$ for object $X$ issued to a non-failing node $s;$ the read operation sends a $\langle \texttt{read}, opid,X\rangle$ message to the server. Based on the code in Algorithm \ref{alg:inputactions_clients}, if one of the conditions in line \ref{line:readimmediatereturncheck1} or line \ref{line:readreturnimmediatecheck2} is true, then the server responds to the client with a value and the read terminates. So, to show termination of $\pi,$ it suffices to focus on the case where the conditions in line \ref{line:readimmediatereturncheck1} and line \ref{line:readreturnimmediatecheck2} in Algorithm \ref{alg:inputactions_clients} are not true. In this case, based on line \ref{line:valinq_send}, a $\texttt{val\_inq}$ message is sent to every node. 
Based on the protocol, every non-failing node  eventually responds to a $\texttt{val\_inq}$ message via a $\texttt{val\_resp}$ or a $\texttt{val\_resp\_encoded}$ message. Specifically, in $\beta$, node $s$ receives a $\texttt{val\_resp}$ or  $\texttt{val\_resp\_encoded}$  message from every node in $S$. Let $P$ be the first point of $\beta$ such that at $P$, node $s$ has received a  $\texttt{val\_resp}$ or $\texttt{val\_resp\_encoded}$   message corresponding to $\pi$ from every node in $S.$ Note that the point $P$ itself is the point of receipt of the last such  message from a node in $S.$ We show that the node $s$ sends a value to the client in response to the read by or at point $P$. 
At point $P$, there are two cases: (i) there exists no entry in $ReadL$ corresponding to $\pi$, or (ii) there is an entry $(clientid, opid, X, tags,\bar{w})$ in $ReadL$ corresponding to $\pi$, where $opid$ is its identifier. In the former case, Lemma \ref{lem:read-liveness} implies that operation $\pi$  completes. We examine the later case here. 

By point $P$, every node in $S$ sends a $\texttt{val\_resp}$ or  $\texttt{val\_resp\_encoded}$ message. We claim that in fact, every node in $S$ sends a $\texttt{val\_resp\_encoded}$. To see this, note that if a node $n \in S$ sends a \texttt{val\_resp}, the code in line \ref{line:valresp_handle} in Algorithm \ref{alg:input_actions} responds to the read operation and removes the $(clientid, opid, X, tags,\bar{w})$  tuple from $ReadL.$ Since, for every operation, an entry is made into $ReadL$ corresponding to the operation at most once, and the removal of the tuple happens before point $P$, we conclude that the tuple does not exist in $ReadL$ at point $P$ - this contradicts our earlier hypothesis that such a tuple does exist. Therefore, we conclude that every node in $S$ sends a $\texttt{val\_resp\_encoded}$ message to node $s$ in response to $\pi.$ 

From Lemma \ref{lem:livenessvalresp}, we therefore conclude that $S \subseteq \{j: \Pi_j(\bar{w}) \neq \bot\}.$ Therefore line \ref{line:readreturncheck} returns true. From the protocol code lines \ref{line:read-send-val_resp_encoded}, we conclude that node $s$ eventually sends a response to the read operation $\pi$ and removes the tuple $(clientid, opid, X, tags,\bar{w})$  from ReadL. Therefore $\pi$ eventually completes.
\end{proof}

 \section{Eventual Consistency}
 
 \label{app:eventual_proof}
 
Theorem \ref{eventual consistent} is based on the following two lemmas.
\begin{lemma}\label{lem:eventually_apply}
Consider a fair execution $\beta$ of \CausalEC~ where every server is non-halting.
Consider an element $(s',X,v,t)$  in $InQueue_{s}$ at server $s$ at point $P$ of $\beta$.
Eventually, there is a point $Q$ in $\beta$ such that the element $(s',X,v,t)$ is not in $Inqueue_{s}$ at any point after $Q$. Additionally, for any point $Q'$ that is equal to, or after $Q$ in $\beta$, $vc_{s}^{Q'} \geq t.ts.$
\end{lemma}

The above lemma applies for most causal consistency protocols (see \cite{ahamad1995causal}).

We first prove lemma \ref{lem:eventually_apply}. We then prove Theorem \ref{eventual consistent}.
 \subsection{Proof of Lemma \ref{lem:eventually_apply}}
 
\begin{proof}
According to the algorithm, tuple $(s',X,v,t)$ is sent to every node other than $s$, eventually this tuple will be in every other node $s$'s $InQueue$. For node $s$, because $InQueue$ is a priority queue, and because there are finite number of timestamps that are smaller than $t.ts$, the number of tuples that are ever placed ahead of ahead of $(s,X,v,t)$ is finite. 

We repeat two claims that were shown in the proof of Lemma \ref{lem:uniquetag1}.
\begin{claim}
Consider any any tuple  $(\overline{s},\overline{X},\overline{v},\overline{t})$ in $Inqueue.Head.$ If we index all the writes by clients in $\mathcal{C}_{\overline{s}}$ as $\pi_{1},\pi_{2}, \ldots, $ in the order of the arrival of the corresponding $\texttt{write}$ messages at server $\overline{s}$, then:
 $\pi_{\overline{t}.ts[\overline{s}]}$ is a write to object $\overline{X}$ with value $v$. Further $\overline{t}.ts=ts(\pi_{\overline{t}.ts[\overline{s}]}).$
 \label{claim:writeordering}
\end{claim}

\begin{claim}
For any servers $s,\overline{s}$ for any point $Q$ of $\beta$, let $\pi$ denote the $vc_s^{Q}[\overline{s}]$-th write operation at server $s$ as per the ordering of Claim \ref{claim:writeordering}. Then $vc_s^{Q} \geq ts(\pi).$
\label{claim:writeordering2}
\end{claim}

We prove the lemma via induction on $t.$ 
\underline{Base case:}
Consider a tag $(s_1,O_1,v_1,t_1)$ such that server $s$ receives adds no tuple to $Inqueue$ with a tag smaller than $t_1$ in $\beta.$ 
\begin{claim}
$t_1.ts$ is a $N$ dimensional vector such that the $s_1$th component of $t_1.ts$ is equal to $1$, and all other components are $0.$  
\end{claim}
\begin{proof}
Let $\pi_1$ be the write operation with the smallest tag in $\beta$, that is $tag(\pi_1) = t_1.$ From the assumption, it follows that $\pi_1$ is issued by a client in $\mathcal{C}_{s_1}.$ Note that $t_1.ts = ts(\pi_1).$ Because there is no tuple with tag smaller than $t_1$ in any servers $Inqueue,$ $\pi_1$ is the first write at $s_1$. From Claim \ref{claim:writeordering}, $t_1.ts[s_1]=1.$ Suppose there is a server $s' \neq s_1$ such that $t_1.ts[s'] \neq 0.$ Let $\pi_2$ be the $t_1.ts[s']$th write at server $s'.$ From Claim \ref{claim:writeordering2}, we have $ts(\pi_2) < ts(\pi_1),$ which violates the assumption that $\pi_1$ is the write with the smallest tag - a contradiction.
\end{proof}

Because $t_1$ has the smallest tag, the corresponding tuple \newline $(s_1,O_1,v_1,t_1)$ is in $Inqueue.Head$ for any server $s$. Note that that before this tuple is processed at server $s$, $vc_{s}[s_1]=0.$ From the above claim, we note that the first $Apply\_Inqueue$ action performed after the addition of the tuple $(s_1,O_1,v_1,t_1)$ returns true of line \ref{line:applycondition}. Therefore, line \ref{line:inqueue_remove} is executed and this tuple is removed from $Inqueue_s.$

\underline{Inductive step:}

We make the inductive hypothesis that for any tuple $(\overline{s}, \overline{X}, \overline{v}, \overline{t})$ with $\overline{t} < t,$ there is a point $\overline{Q}$ such that, at any point after $\overline{Q}$, the tuple is not in $Inqueue_s$. There are two possibilities: (i) the tuple $(s,X,v,t)$ is not in $Inqueue_s$ at any p
oint after $\overline{Q}$, or (ii) there is a point $P'$ after $\overline{Q}$ such that the tuple $(s,X,v,t)$ is in $Inqueue_s^{P'}$. In case $(i)$ the lemma readily holds. So we consider case $(ii).$ By the inductive hypothesis, we recognize that at point $P'$, and at all points following $P'$ where $(s,X,v,t) \in \textrm{Inqueue},$ the tuple is actually $Inqueue.Head.$ 

\begin{claim}
For any server $\overline{s} \neq s',$ $vc_{s}^{P'}[\overline{s}] \geq t.ts[\overline{s}]$
\label{claim:vcts1}
\end{claim}
\begin{proof}
Let $\pi$ be the $t.ts[\overline{s}]$th write performed at node $\overline{s}$ as per the ordering in Claim \ref{claim:writeordering}. From Claim \ref{claim:writeordering2}, $vc_{s}^{P'} \stackrel{(a)}{\geq} ts(\pi)$. From Lemma \ref{claim:writeordering}, we have $ts(\pi)[\overline{s}] \stackrel{(b)}{=} t.ts[\overline{s}]$. From $(a)$ and $(b),$ the clam follows.
\end{proof}

\begin{claim}
For any server, $vc_{s}^{P'}[s'] + 1 =  t.ts[{s}']$
\label{claim:vcts2}
\end{claim}
\begin{proof}
Let $\phi,\pi$ respectively be the $t.ts[s']-1$th and $ t.ts[s']$ th writes performed at node $\overline{s}$ as per the ordering in Claim \ref{claim:writeordering}.

From Theorem \ref{thm:causallyconsistent}, we know that $ts(\phi) < ts(\pi).$ By the inductive hypothesis, we know that the $\texttt{app}$ message sent on behalf of $\phi$ has been emptied from the Inqueue by server $s$. From the server protocol, we have $vc_{s}^{P'} \stackrel{(a)}{\geq} t.ts[s'] - 1.$ Further, from the protocol, any $\texttt{app}$.

From the protocol, server $s$ only increments $vc_{s}[s']$ on an \newline $\texttt{Apply\_Inqueue}$ action, where a tuple from $s'$ at $Inqueue.Head$ is processed and line \ref{line:applycondition} returns true. Since the server has not processed tuple $(s,X,v,t)$ by point $P'$,  we know that  $vc_{s}^{P'} \stackrel{(b)}{<} t.ts[s'].$  From $(a)$ and $(b)$, the lemma follows.
\end{proof}

Consider the first $\texttt{Apply\_Inqueue}$ action after $P'$. Claims \ref{claim:vcts1}, \ref{claim:vcts2} imply that line \ref{line:applycondition} is true for this line. Therefore, line \ref{line:inqueue_remove} is executed. This completes the proof.
\end{proof}

\subsection{Proof of Theorem \ref{eventual consistent}}

\begin{proof}

Consider any write operation $\pi.$
From the protocol, $\pi$ sends an $\langle \texttt{app},X,val,tag(\pi_X)\rangle$ message to every server. On receipt, this message is added to a servers $Inqueue.$
Lemma \ref{lem:eventually_apply} implies that there is a point $P$ in $\beta$ after which the tuple does not exist in any server in $\beta$. Based on Lemma \ref{lem:vc_after_apply}, for every server $s$, for any point after $P$, $vc_{s}^{P} \geq ts(\pi).$ So, for any read $\phi$ that begins after $P$ and returns,  $tag(\phi) \geq tag(\pi_X)$, that is $\pi_X \leadsto \phi$. The set of reads invoked before $P$ is finite. Therefore, in any execution, the set $S=\{\phi: \pi_X \leadsto \phi\}$ contains all but a finite number of operations in $\beta$.
\end{proof}

\section{Storage Cost: Proof of Theorem \ref{thm:storagecost}}
\label{app:storage}

Theeorem \ref{thm:storagecost} relies on the following key lemma, which states that no object version remains in the list $L[X]$ indefinitely in any execution.

\begin{lemma}
Consider a fair execution $\beta$ where every server is non-halting. Let $\pi$ be a write to object $X$ that completes in $\beta$ with value $v$. Then, there is a point $P$ in $\beta$ after which, for any server $s,$ there is no entry of the form $(tag(\pi),val)$ in $L_s[X]$ for any value $val.$
\label{lem:storagecostkey}
\end{lemma}

We first prove some preliminary lemmas. Then we present a proof of Lemma \ref{lem:storagecostkey}, which is followed by a proof of Theorem \ref{thm:storagecost}.

\subsection{Preliminary Lemmas}
\begin{lemma}
Consider a fair execution $\beta$ of $CausalEC$ where every server is non-halting. Consider a point $P$ and an entry $(t,val) \in L_s[X]$ for some server $s$ and some object $X$. Then, 
\begin{enumerate}[(a)]
    \item there is a point $Q$ in $\beta$ such that, for every point  $Q'$ after $Q$ in $\beta,$ we have $M_s^{Q'}.tagvec[X] \geq t$.
    \item there is a point $\overline{Q}$ in $\beta$ such that, by $\overline{Q}$ every node $s$ has at least one entry $(t_{s'},s')$ in $DelL_s[X]$ with $t_{s'} \geq t,$ for every node $s' \in \mathcal{N}.$
\end{enumerate}
 \label{lem:Mexceedslisteventually}
\end{lemma}

\begin{proof}
We show $(a)$ first.

\underline{Proof of $(a)$}

Assume the hypothesis of the lemma, that is $(t,val) \in L_s^{P}[X]$  in execution $\beta$. Suppose that at some point $Q$ in $\beta$, $M_s^{Q}.tagvec[X] \geq t$. Then, because of Lemma \ref{lemma:tagsalwaysincrease}, we know that at any point $Q'$ after $Q$, $M_s^{Q'}.tagvec[X] \geq M_s^{Q}.tagvec[X] \geq t.$ Therefore the statement of the lemma holds. If possible, suppose there exists no such point $Q$. Then there is a point $P'$ such that, at any point $Q'$ after $P',$  $M_s^{Q'}.tagvec[X] = M_s^{P'}.tagvec[X] < t$. We will show a contradiction. We consider two cases: $X \in \mathcal{X}_{s}$ and $X \notin \mathcal{X}_{s}.$ 

\underline{Case 1: $X \in \mathcal{X}_{s}$}

\begin{claim}
At any point of the execution after $P$, the tuple $(t,val)$ is in $L_s[X].$ 
\label{claim:rr0}
\end{claim}
\begin{proof}

From Lemma \ref{lem:tag-delete-ordering} and the hypothesis that $t > M_s^{P'}.tagvec[X]$ implies that for any entry $(s,t')$ in $DelL_s[X]$ at any point of $\beta$, we have $t' < t$. Items are removed from $L_s[X]$ only as a part of $\texttt{Garbage\_Collection}$ actions in line \ref{line:GC1} or \ref{line:GC2}. By examining lines \ref{line:GC1}, \ref{line:GC2}, the $\texttt{Garbage\_Collection}$ action does not remove the element $(t,val)$ from $L_s[X]$. Therefore, at any point of the execution after $P$, the tuple $(t,val)$ is in $L_s[X].$ 
\end{proof}
In particular, consider a point $Q'$ after $P'$ where an $\texttt{Encoding}$ action is performed. At this point, because $(t,v)$ is in $L_s^{Q'}[X],$ we have $L_s^{Q'}[X].Highesttagged > t.$ Therefore, the condition for the for loop in line \ref{line:Forloopupdatetag1} in Algorithm \ref{alg:internal_actions} returns true for object $X$. If the condition in line \ref{line:conditionforupdate} is satisfied, then line \ref{line:updatetag1} is executed and $M_s^{Q'}.tagvec[X]$ gets updated to a tag that is at least $t$, which is a contradiction to our earlier assumption that $M_s.tagvec[X]$ is always smaller than $t$ for any point after $P'$.

If \ref{line:conditionforupdate}
returns false, then line \ref{line:readlentryencoding} is executed. By Theorem \ref{lem:read_terminates_f=0}, the read eventually responds and there is a point after $P''$ such that $(M_s^{P'}.tagvec[X], v)$ belongs to $L_s[X].$ Let $Q'$ be the first point after $P''$ where $(M_s^{P'}.tagvec[X], v)$ lies in  $L_s[X].$ Let $Q''$ be the point of the 
the first  $\texttt{Encoding}$ action performed after $Q''$ in $\beta.$ The result follows from the following claims:

We make the following claim:
\begin{claim}
A $\texttt{Garbage\_Collection}$ action performed between $Q'$  and $Q''$ does not remove the element $(M_s^{P'}.tagvec[X], v)$.
\label{claim:rr1}
\end{claim}
\begin{claim}
In the $\texttt{Encoding}$ action performed at $Q'',$ line \ref{line:conditionforupdate} is executed and returns true.
\label{claim:rr2}
\end{claim}
Once we show the above claims, the lemma statement readily follows, line \ref{line:updatetag1} is executed as a part of the $\texttt{Encoding}$ action at $Q'$, and $M_s^{P'}.tagvec[X]$ gets updated to a tag that is at least $t$, which is a contradiction.

\begin{proof}[Proof of Claim \ref{claim:rr1}]
Consider a $\texttt{Garbage\_Collection}$ action performed at point $U$ between $Q'$ and $Q''.$ Because of Lemma \ref{lem:tag-delete-ordering} and because our assumption implies that  $M_s^{P'}.tagvec[X] = M_s^{U}.tagvec[X]$, we infer that $tmax_s^{U}[X] \leq M_s^{P'}.tagvec[X].$ If the inequality is strict, then from lines \ref{line:GC1}, \ref{line:GC2}, we infer that the element  $(M_s^{P'}.tagvec[X], v)$ is not deleted from $L_s^{U}[X].$ If the inequality is not strict, we have $tmax_s^{U}[X] = M_s^{P'}.tagvec[X] = M_s^{U}.tagvec[X].$ Note that point $Q'$, $(t,val)$ is in $L_s^{U}[X]$, so $L_s^{U}[X].Highesttagged.tag > M_s^{U}.tagvec[X]$.  Therefore, the condition in line \ref{line:GCcondition1} returns false. Therefore  line \ref{line:GC2} is executed (and not line \ref{line:GC1}), and from the line, the element $(M_s^{P'}.tagvec[X], v)$ is not deleted from $L_s[X]$.
\end{proof}
\begin{proof}[Proof of Claim \ref{claim:rr2}]
 From Claim \ref{claim:rr0}, $(t,val) \in  L_s^{U}[X] $ and $t > M_s^{P'}$. Therefore, for object $X$, line \ref{line:Forloopupdatetag1} returns true. Therefore, line \ref{line:conditionforupdate} in Algorithm \ref{alg:internal_actions} is executed as a part of the $\texttt{Encoding}$ action.
Because at point $Q',$ we have $(M_s^{P'}.tagvec[X], v), \in L_s^{Q'}[X]$, and because of Claim \ref{claim:rr1}, we have  $(M_s^{P'}.tagvec[X], v), \in L_s^{U}[X].$ Therefore, line \ref{line:conditionforupdate} returns true. 
\end{proof}

\underline{Case 2: $X \notin \mathcal{X}_{s}$}

Let $R=\{i \in \mathcal{N}: X \in \mathcal{X}_{i}\}.$ Note that $s \notin R.$ 
\begin{claim}
In $\beta$, eventually, every node $s' \in R$  sends a $\langle \texttt{del},X,t_{s'}\rangle$ message to all the other nodes for some $t_{s'}\geq t$. In addition, node $s$ eventually adds a $(\overline{t}_{s'},s)$ element to $DelL_{s}[X]$ eventually with $\overline{t}_{s'} \geq t$. 
\label{claim:rr3}
\end{claim}

\begin{proof}
Consider any node $s'$ such that $X \in \mathcal{X}_{s'}$, that is $s' \in R$. Because the lemma has been shown for Case $1,$ we know that there is a point $Q'$ such that $M_{s'}^{Q'}.tagvec[X] \geq t.$ In particular, at the point where $M_{s'}.tagvec[X]$ becomes at least $t$, it executes line \ref{line:Delete_send1}. Therefore, it adds an element $(\overline{t}_{s'}, s')$ element to $DelL_{s'}[X]$ and sends a $\langle \texttt{del},X,\overline{t}_{s'}\rangle$ message to every other server $s'' \in R$ with tag $\overline{t}_{s'} \geq t.$ By a similar argument, by some point $Q'',$ it receives  $\langle \texttt{del},X,\overline{t}_{s''}\rangle$ message from every other node $s'' \in R$ with $\overline{t}_{s''} \geq t.$ 

Consider the first $\texttt{Garbage\_Collection}$ action after $Q^{''}$ performed by server $s'.$
Note that by $Q'',$ $s'$ has received a $\texttt{del}$ message from server in $R$ for object $X$ with a tag at least $t$. Therefore $t$ belongs to the set $U$ executed in Line \ref{line:Delete3_Uset}, which implies that $t_{s'} = \max(U) \geq t.$  Therefore, in line \ref{line:Delete_send3}, it sends a $\texttt{del}$ message to all nodes with a tag $t_{s'}$ which is at least $t$. 
\end{proof}

From the above claim, we conclude that in $\beta$, node $s$ receives at least one $\langle \texttt{del},X,t_{s'}\rangle$ from every node $s'$ in $R$ with $t_{s'} \geq t$. Let $\overline{P}$ be the first point no sooner than $P$ at which the node $s$ has received such messages.
Consider any $\texttt{Encoding}$ action performed by node $s$ after $\overline{P}$. At point $\overline{P}$, from our earlier assumption $L_s^{\overline{P}}[X].Highesttagged.tag \geq t > M_{s}^{\overline{P}}.tagvec[X].$ Because the $\texttt{Encoding}$ is performed after $\overline{P}$, we have $t \in U,$ where $U$ is found on executing line \ref{line:Delete2_Uset}. Further, because $(M_{s}^{P'}.tagvec[X],v) \in L_{s}^{Q''}[X]$, we have $t \in \overline{U},$ where $\overline{U}$ is found in line \ref{line:Delete2_Ubarset}. Thus, $\max(U \cap \overline{U}) \geq t > M_{s}^{\overline{P}}.tagvec[X],$ which implies that line \ref{line:Delete2_condition} returns true and line \ref{line:updatetag2} is executed. Thus, after the $\texttt{Encoding}$ action, we have $M_s.tagvec[X] = \max(U \cap \overline{U}) \geq t,$ which completes the proof.

\underline{Proof of (b)}
It suffices to show that (i) every node $s'$ sends a $\texttt{del}$ message for object $X$ with tag at least $t$ eventually to every other node $s$,  and (ii) that every node $s$ adds an element $(t_s,s)$ to $DelL_{s}[X]$ with $t_{s} \geq t$ eventually.
From Claim \ref{claim:rr3}, we know that every node $s$ with $X \in \mathcal{X}_{s}$ eventually sends  $\texttt{del}$ message with tag $t_{s} \geq t.$ Therefore, the lemma holds if $X \in \mathcal{X}_{s}$. 

For a node $s$ such that $X \notin \mathcal{X}_{s}$, because we have shown property (a) of the lemma, line \ref{line:updatetag2} is executed with updating $M_s.tagvec[X]$ with a tag at least $t$. We note that an update of the $M_{s}.tagvec[X]$ variable is also associated with the sending of a $\texttt{del}$ message in line \ref{line:Delete_send2}, and the adding of a $(t_s,s)$ element to $DelL_s[X].$  Therefore, the lemma holds.

\end{proof}

\begin{lemma}
Consider a fair execution $\beta$ where every server is non-halting. Let $\pi$ be a write to object $X$ that completes in $\beta$ with value $v$. Then,
\begin{enumerate}[(a)]
    \item there is a point $P$ in $\beta$ after which, for any server $s,$ for any point $Q$ which is after $P$, $M_s^{Q}.tagvec[X] \geq tag(\pi).$
    \item there is a point $Q$ in $\beta$ by which every server $s$ has, for every server $s' \in \mathcal{N}$,  at least one tuple $\langle t_{s'}, s'\rangle$ in $Del_{s}^{Q}[X]$ with $t_{s'} \geq tag(\pi)$
\end{enumerate} 
\label{lem:Meventuallyexceedseverywrite}
\end{lemma}

\begin{proof}
We show (a) first.
Suppose $\pi$ is issued by a client in $\mathcal{C}_{s'},$ for some server $s'$. 
Based on the protocol, on receiving the value write $\pi,$ server $s$ adds element $(tag(\pi),v)$ to $L_s[X]$ and sends an $\langle \texttt{app},X,v,tag(\pi)\rangle$ message to all other servers. From Lemma \ref{lem:eventually_apply}, and from line \ref{line:list_applyaction} in Algorithm \ref{alg:internal_actions}, we infer that in $\beta,$ for any server $s$, the element $(tag(\pi),v)$ is eventually added to $L_s[X]$ in $\beta.$ Combined with Lemma \ref{lem:Mexceedslisteventually} (a), we conclude that there is a point $P'$ such that, for every point after $P_s,$ $M_s.tagvec[X] > tag(\pi).$ Therefore, property (a) of the lemma is satisfied by choosing point $P$ to be the latest of $P_1, P_2, \ldots, P_N$.

Statement (b) of Lemma \ref{lem:Meventuallyexceedseverywrite} similarly follows from statement (b) of Lemma \ref{lem:Mexceedslisteventually}.

\end{proof}

\begin{lemma}
Consider any infinite execution $\beta$ of $CausalEC$, and consider  point $P$ of $\beta$,  server $s$ and object $X$. At least one of the following statements is true:
\begin{enumerate}
    \item An element $(M_s^{P}.tagvec[X],v)$ is not added to $L_s[X]$ at any point of the execution after $P$, for any value $v$
    \item There exists a point $Q$ of $\beta$ after $P$ such that $M_s^{Q}.tagvec[X] > M_s^{P}.tagvec[X]$
\end{enumerate} 
\label{lem:listisnotaddedbeyond}
\end{lemma}
\begin{proof}

From Lemma \ref{lem:pointsexec}, we know that there is a point before $P$ at which an element is added to $L_{s}[X]$ with tag $M_s^{P}.tagvec[X]$. From the protocol, we note that the first such point of addition is either on receiving a write in line \ref{line:writeaddtolist} in Algorithm \ref{alg:inputactions_clients}, or an apply action in line \ref{line:applycondition} in Algorithm \ref{alg:internal_actions}.
For any tag, each node receives either an apply message, or a message from a write with that tag, and such a message is received only once. Therefore, we note that an element with tag $M_s^{P}.tagvec[X]$ is never added after point $P$ to $L_s[X]$ via \ref{line:writeaddtolist} in Algorithm \ref{alg:inputactions_clients}, or an $\texttt{Apply\_Inqueue}$ action in line \ref{line:applycondition} in Algorithm \ref{alg:internal_actions}. 

To complete the proof, we consider the case where for any point $\overline{Q}$ of a fair execution $\beta$ that is after point $P$, $M_s^{\overline{Q}}.tagvec[X] = M_s^{P}.tagvec[X]$. We aim to show that   that an element with tag $M_s^{P}.tagvec[X]$ is not added after point $P$ to $L_s[X]$ on behalf of lines \ref{line:valrespaddtolist}, or line \ref{line:valrespencoded_addtuple} in Algorithm \ref{alg:input_actions}. Note that these lines are executed at server $s$ on receipt of a $\texttt{val\_resp}$ or a $\texttt{val\_resp\_encoded}$ message from some server $s',$ with parameters $clientid=\textrm{localhost},$ object $X$, and a tag vector $tvec$ with $tvec[X] = M_s^{P}.tagvec[X].$ Note that server $s'$  sent such a message on receipt of a $\texttt{val\_inq}$ message from server $s$. Based on the code, the sending of such a $\texttt{val\_inq}$ from server $s$ with $clientid=\textrm{localhost}$ occurs on line \ref{line:valinq_send2}. However, this line is executed if there exists an element $(t,v)$ in $L_s[X]$ with $t > M_s^{P}.tagvec[X].$ Lemma \ref{lem:Mexceedslisteventually}, there is a point $Q$ after $P$ where  $M_s^{Q}.tagvec[X] \geq t >  M_s^{P}.tagvec[X]$. This is a contradiction. Therefore, that an element with tag $M_s^{P}.tagvec[X]$ is not added after point $P$ to $L_s[X]$ in $\beta$. 
\end{proof}

\begin{lemma}
Consider an infinite execution $\beta$ of $CausalEC$, and consider  point $P$ of $\beta$,  server $s$ and object $X$. 
For a tag $t < M_s^{P}.tagvec[X]$, then eventually, there is a point $Q$ such that, after $Q$ no element $(t,v)$ is added to $L_s[X]$ for any value $v$.
\label{lem:listisnotaddedbeyond2}
\end{lemma}
\begin{proof}

For tag $t$, at node $s$, an element $(t,v)$ is added to $L_s[X]$ at most one time in $\beta$ via line \ref{line:writeaddtolist} in Algorithm \ref{alg:inputactions_clients}, or line \ref{line:applycondition} in Algorithm \ref{alg:internal_actions}. We aim to show that  that an element with tag $t$ is not added eventually, after a point $Q$, to $L_s[X]$ on behalf of lines \ref{line:valrespaddtolist}, or line \ref{line:valrespencoded_addtuple} in Algorithm \ref{alg:input_actions}. From the code, any element   $\langle clientid,opid,X, tvec, \overline{w}\rangle$ that is added to $ReadL_s$ at point $Q$ has $tvec=M_s^Q.tagvec.$ If $Q$ is after $P$, then from Lemma \ref{lemma:tagsalwaysincrease}, we have $tagvec[X] = M_s^Q.tagvec[X] \geq M_s^P.tagvec[X] > t$. Therefore, an element  \newline  $\langle \texttt{localhost},opid,X, tvec, \overline{w}\rangle$ with $tvec[X] = t$ is added at most a finite number of times to $ReadL_s$ in $\beta$ in line \ref{line:readlentryencoding}. Correspondingly, a finite number of $\texttt{val\_inq}$ messages are sent on line \ref{line:valinq_send2} with a parameter tag vector $tvec$ and object $X$ satisfying $tvec[X] = t.$ Since lines \ref{line:valrespaddtolist}, or line \ref{line:valrespencoded_addtuple} in Algorithm \ref{alg:input_actions} are executed on $\texttt{val\_resp}$ and $\texttt{val\_resp\_encoded}$ messages that are responses to $\texttt{val\_inq}$ messages, these lines are executed a finite number of times with parameters $tvec$ and object $X$ satisfying $tvec[X]=t.$ Therefore, eventually, there is a point $Q$ after which an element with tag $t$ is not added to $L_s[X]$ on behalf of lines \ref{line:valrespaddtolist}, or line \ref{line:valrespencoded_addtuple} in Algorithm \ref{alg:input_actions}. This completes the proof.
\end{proof}

\remove{
\begin{lemma}
Consider a fair execution $\beta$ where every server is non-halting. Let $\pi$ be a write to object $X$ that completes in $\beta$ with  value $v$. Then, there is a point $P$ in $\beta$ after which, for any server $s,$ there is no entry of the form $(tag(\pi),val)$ in $L_s[X]$ for any value $val.$
\label{lem:finalstorage}
\end{lemma}
}

\subsection{Proof of Lemma \ref{lem:storagecostkey}}
\begin{proof}
We consider two cases (I) $\pi$ is the highest tagged write to object $X$ in $\beta,$ and (II) there is a write in $\beta$ with a tag larger than $tag(\pi)$ in $\beta$.

\underline{Case (I) $\pi$ is the highest tagged write to object $X$ in $\beta.$}

Consider any server $s$. From Lemma \ref{lem:Meventuallyexceedseverywrite}, we know that there is a point $Q$ such that: 
\begin{itemize}
    \item server $s$ has, for every server $s' \in \mathcal{N}$, an element $(t_{s'},s')$ in $DelL_{s}^{Q}[X]$ with $t_{s'} \geq tag(\pi).$ 
    \item $M_s^{Q}.tagvec[X] \geq tag(\pi).$
\end{itemize}
Furthermore, since there is no write to object $X$ with a tag larger than $tag(\pi)$, so we have $t_{s'} = tag(\pi), \forall s'$ at $Q$, and $M_s^{Q}.tagvec[X] = tag(\pi).$ In fact, because of Lemma \ref{lemma:tagsalwaysincrease}, at every point after $Q$, we have $M_s.tagvec[X] = tag(\pi).$

Consider the first $\texttt{Garbage\_Collection}$ action after ${Q}$. We claim that line \ref{line:GCcondition1} returns true. Since there is an element $(tag(\pi),s')$ for every $s' \in \mathcal{N}$ in $DelL_s[X],$ the set $S$ in line \ref{line:sset}  and the set $\overline{S}$ in line \ref{line:s1set} both contain $tag(\pi).$ Therefore, $t_{max}[X] = tag(\pi)$ in line \ref{line:tmax}. Further, if $L_s[X]$ is non-empty, since there is no tag larger than $tag(\pi)$ in $\beta$ for a write to object $X$, $L_s[X].Highesttagged.tag \leq t_{max}[X] = M_s.tagvec[X].$
Therefore, line \ref{line:GCcondition1} returns true and \ref{line:GC1} is executed. Further, because $tag(\pi) = M_s.tagvec[X],$ it does not belong to set ${T}$ identified in line \ref{line:pendingreads}. Therefore, line \ref{line:GC1} removes any element with tag $tag(\pi)$ from $L_s[X].$ Because $tag(\pi)$ is the largest tag in $\beta$ for objrct $X$, note that statement $(2)$ of Lemma \ref{lem:listisnotaddedbeyond} is not satisfied after the $\texttt{Garbage\_Collection}$. From the lemma statement, we conclude that statement (1) of Lemma \ref{lem:listisnotaddedbeyond} holds. That is, we know that any element with tag $tag(\pi)$ is not added again to $L_s[X]$ after the point of $\texttt{Garbage\_Collection}$ action. Therefore, the lemma holds with point $P$ being the point of the $\texttt{Garbage\_Collection}$ action.

\underline{Case (II) $\pi$ is not the highest tagged write to object $X$ in $\beta.$}
Let $\phi$ be a write to object with tag larger than $tag(\pi).$ Then from Lemma \ref{lem:Meventuallyexceedseverywrite}, we know that 
we know that there is a point $Q'$ such that: 
\begin{itemize}
    \item server $s$ has, for every server $s' \in \mathcal{N}$, an element $(\overline{t}_{s'},s')$ in $DelL_{s}^{Q}[X]$ with $\overline{t}_{s'} \geq tag(\phi) > tag(\pi).$ 
    \item $M_s^{Q}.tagvec[X] \geq tag(\phi) > tag(\pi).$
\end{itemize}

Since the hypothesis of Lemma \ref{lem:listisnotaddedbeyond2} is satisfied with tag $t = tag(\pi)$, we conclude that there is eventually a point $Q''$ such that, after $Q'',$ no element is added with tag $t=tag(\pi)$ to $L_s[X].$  Since entries added into $ReadL_s$ at a point $\overline{P}$ have parameter tag vector $M_s^{\overline{P}}.tagvec,$ no entry is added with parameter tag vector $tvec$ satisfying $tvec[X]=tag(\pi)$ after $Q$ to $ReadL.$  

Therefore, entries with parameter tag vector $tvec$ satisfying $tvec[X]=tag(\pi)$ are added at most a finite number of times to $ReadL_s$ in $\beta.$ From Theorem \ref{lem:read_terminates_f=0}, we note that every pending read in $ReadL_s$ is eventually returned and the entry is cleared. Let $Q'''$ be the point where every read with tag vector $tvec$ satisfying $tvec[X]=tag(\pi)$ is cleared from $ReadL_s.$ 
Let $\overline{Q}$ be the latest of $Q', Q'', Q'''.$

 Consider the first $\texttt{Garbage\_Collection}$ action after $\overline{Q}$, suppose this is performed at point ${P}$. Since $P$ is no earlier than $Q'$ there is an element $(t_{s'},s')$ for every $s' \in \mathcal{N}$ with $t_{s'} > tag(\phi)$ in $DelL_s^{P}[X],$ the set $S$ in line \ref{line:sset}  contains $tag(\pi).$ Therefore, $tmax_{s}^{P}[X] > tag(\pi)$ in line \ref{line:tmax}.  Further, since the action is performed no sooner than $Q''',$ $tag(\pi) \notin T,$ where $T$ is determined in line \ref{line:pendingreads}. Therefore, the $\texttt{Garbage\_Collection}$ action executes \ref{line:GC1}, \ref{line:GC3} or  \ref{line:GC2}, which all remove any element with tag $tag(\pi)$ from $L_s^{{P}}[X].$

Further, since ${P}$ is after $Q'',$ an element with tag $tag(\pi)$ is not added after ${P}$ to $L_s[X].$ Therefore, the lemma holds, that is,  after $P$ there is no entry with tag $tag(\pi)$ in $L_s^{Q}[X].$ 
\end{proof}
\subsection{Proof of Theorem \ref{thm:storagecost}}
\begin{proof}

\underline{Proof of (a):} 
By the theorem hypothesis, there are a finite number of write operations in $\beta.$
From Lemma \ref{lem:storagecostkey}, for every object $X$, for every server $s$, for every write operation $\pi$ there is a point $P_{\pi,s,X}$ in $\beta$ such that $L_{s}[X]$ does not contain an element $(tag(\pi),val)$ for any val in $L_s[X].$ Let $P$ be the latest of the point $P_{\pi,s,X}$ over all write operations $\pi,$ servers $s$ and objects $X$; note that $P$ exists in $\beta$ because the number of write operations, servers and objects are all finite. At any point $\overline{P}$ after $P$, there is no element of the form $(tag(\pi),val)$ in $L_s^{\overline{P}}[X]$ for any object $X$  in $\beta$, for any server $s$. 
From Lemma \ref{lem:valuesarelegitimate} (proved in Appendix \ref{app:safety}), we know that an element $(t,v)$ exists in a list $L_s[X]$ for some server $s$ for some object $X$ only if $t$ is that of some write operation to object $X$, and $v$ is the value of that write operation. Therefore, 
we infer that $L_s[X]$ is empty after $P$.

\underline{Proof of (b):} 
It readily follows from Lemma \ref{lem:eventually_apply}

\underline{Proof of (c):} 
The hypothesis implies that there are a finite number of read operations in $\beta$. From lemma \ref{lem:read-liveness}, there is no entry $\langle clientd, opid, X, tags, \overline{w}\rangle$ in $ReadL_s$  with $clientid \neq \texttt{localhost}$, for any server $s$.

 An entry $\langle \texttt{clientid}, opid, X, tags, \overline{w}\rangle$  with $clientid \neq \texttt{localhost}$ is added to $ReadL_{s}$ only on receiving a $\texttt{read}$ message from a client at server $s$. Because the number of read operations is finite, and because lemma \ref{lem:read-liveness} implies that every read operation terminates. Further, from the protocol code, the sending of a $\texttt{read-return}$ message to operation $opid$ in line \ref{line:read-remove1} in Algorithm \ref{alg:inputactions_clients}  lines  \ref{line:read-remove2}, \ref{line:read-remove3} in Algorithm \ref{alg:input_actions} and line \ref{line:read-remove4} in Algorithm \ref{alg:internal_actions} are accompanied by removal of the corresponding entry with operation id $opid$ from $ReadL_s.$ 
Thus, eventually there is a point $Q$ after which $ReadL_s$ is does not contain any entry with $clientid \neq \texttt{localhost}.$

Because of $(a)$ there is eventually a point after which line \ref{line:Forloopupdatetag1} in Algorithm \ref{alg:internal_actions} is not satisfied for any object $X$ for any server $s$ as the list $L_s[X]$ is empty.
Since an entry $\langle \texttt{localhost}, opid, X, tags, \overline{w}\rangle$ is added to $ReadL_s$ on line \ref{line:readlentryencoding}, which is executed only if line \ref{line:Forloopupdatetag1} is true, there is a point after which such an entry is not added to $ReadL_s.$ Furthermore, such an entry is eventually cleared by executing line \ref{line:read-remove2} or \ref{line:read-remove3} in Algorithm 
\ref{alg:input_actions}. 

Therefore, there is a point $Q$ after which $ReadL_s$ is empty in $\beta.$

\end{proof}

%% file: AppendixC.tex
\section{Low-Cost Variant of \CausalEC}
\label{app:variant}
We describe a low-cost variant of \CausalEC~that is used for cost analysis in Sec. \ref{sec:communication}. The low cost variant uses Lamport timestamps and an Eventual-Broadcast primitive described next.

\underline{(i) Use of Lamport timestamps}
For any write operation $\pi$, it has a Lamport timestamp defined as follows.
\begin{definition}
For any tag $t$ its Lamport-timestamp denoted by $lt(t)$ is defined as:
$$lt(t)=(s,t.ts[s])$$
where $s$ is the unique server such that client with identifier $t.id$ belongs to $\mathcal{C}_{s}$. For any write operation $\pi,$ its Lamport-timestamp denoted (with slight abuse of notation) as $lt(\pi)$ is equal to $lt(tag(\pi)).$
\end{definition}
We consider a modification of the algorithm that uses Lamport timestamps (instead of tags) for $\texttt{val\_inquiry},\texttt{val\_response},\texttt{val\_response\_encoded}$ and $\texttt{Del}$ messages. The use of Lamport timestamps is justified by Lemma \ref{lem:uniquetag1}. 

Let $\mathcal{LT}=\mathbb{N}^{2}$ denote the set of all possible Lamport timestamps. In our algorithm descriptions in Algorithms \ref{alg:input_actions_var}, \ref{alg:inputactions_clients_var},
\ref{alg:internal_actions_var}, for a tag vector $tagvec \in \mathcal{T}^{\mathcal{X}},$ we denote by $lt(tagvec) \in \mathcal{LT}^{\mathcal{X}}$ as the vector where each component of $tagvec$ is replaced by its Lamport timestamp.

\underline{(ii) Using $O(N)$ delete message}
As stated, Algorithms \ref{alg:inputactions_clients}, \ref{alg:input_actions} \ref{alg:internal_actions} translate to $O(N^{2})$ delete messages being sent per write operation. We explain an optimization that reduces this complexity to $O(N).$ Consider the primitive $\texttt{Eventual-Broadcast}(m,S)$ where $m$ is a message of the form $(\overline{m},id,t)$ where $t$ comes from a totally ordered set $\mathcal{T}$, and $S$ is set of nodes. The primitive satisfies following correctness  (liveness) properties:
\begin{enumerate}
    \item If node $i$  in $S$ issues $\texttt{Eventual-Broadcast}((\langle\overline{m}\rangle,t),S)$ at point $P$ of a fair execution $\beta$ where no node in $S$ fails, then every node in $S$ eventually outputs\newline  $\texttt{Eventual-Broadcast-Receive}((\langle\overline{m}, \rangle,i,t'))$ where $t'= t$,  or node $i$ issues \newline $\texttt{Eventual-Broadcast}((\langle\overline{m}\rangle,t),S)$ after point $P$ of $\beta.$
    \item If node $i$ outputs $\texttt{Eventual-Broadcast-Receive}((\overline{m},i,t))$, then the processor with node $i$ issued 
    $\texttt{Eventual-Broadcast}((\overline{m},t),S)$ before the receipt of the message where $S$ is a set that contains node $i$.
\end{enumerate}
Note that message duplications are permitted at the receiver and the messages can be delivered in any order. The only requirement is eventual delivery in failure-free executions.

Consider a modification to Algorithm \ref{alg:internal_actions} in lines \ref{line:Delete_send1}, \ref{line:Delete_send2}, \ref{line:Delete_send3} where we  use\newline $\texttt{Eventual-Broadcast}(\langle m\rangle,S)$, with appropriate choices of message $m$ and $S$. For example, in line \ref{line:Delete_send1}, we have   $m=(\texttt{del},X),clientid,L[X].Highesttagged.tag)$ and $S=R.$ Similarly, lines \ref{line:Delete_send2}, \ref{line:Delete_send3} use $S= \mathcal{N}$ with appropriate tags in the messages. Correspondingly, line \ref{line:delmessage} corresponds to $\texttt{Eventual-broadcast-receive}$ of the corresponding message. With the above correctness conditions for the Eventual broadcast primitive, it is instructive to observe that all our proofs apply. In particular, Lemma \ref{lem:Meventuallyexceedseverywrite} in Appendix \ref{app:storage} holds and therefore, the delete messages serve their purpose of clearing lists via garbage collection.

Importantly, $\texttt{Eventual-Broadcast}$ can be readily implemented with $O(N)$ messages. For every set $S$, identify a leader node $\ell_{S}\in\mathcal{N}.$ The Eventual broadcast primitive simply sends a the message to the leader $\ell_S,$ which waits for messages from all nodes in $S$. On receipt of messages from all nodes in $S$, leader $\ell_{S}$ sends one message with the collected information back to all the nodes in $S$.  This justifies our write communication cost. 

\subsection*{Formal Description of variant of \CausalEC}

The state variables of the variant are the same, with two differences.
First, there is an additional state variable:
$$lttotag \subset \mathcal{LT} \times \mathcal{T}.$$

Second, the delete list data type is modified as:
$$Del:\mathcal{X} \rightarrow \mathcal{LT} $$

Since the variant sends (several) messages with only Lamport-timestamps, $lttotag$ is used to map Lamport-timestamps to tags.

\begin{figure}[ht] 
\begin{algorithm*}[H]
\caption{Protocol of low-cost variant of  \CausalEC: Transitions for input actions corresponding to messages from clients at node $s$} 
\label{alg:inputactions_clients_var}

\footnotesize
 \underline{On  {$receive\langle\texttt{write}, opid, X,v\rangle$}  from client $id$}:

$vc[s]\leftarrow vc[s]+1$
\label{line:vcincrement_var}

$t\leftarrow (vc,id)$

$L[X] \leftarrow L[X] \cup \{(t,v)\}$
\label{line:writeaddtolist_var}

$send\langle\texttt{write-return-ack}, opid\rangle$ to client $id$
\label{line:writeresponse_var}

$send\langle\texttt{app},X,v,t\rangle$ to all other nodes, $j\neq i$ 
\label{line:writesendapply_var}

for every $opid$ such that there exists $clientid \neq \texttt{localhost}, \bar{v},tags$ such that $(clientid, opid, X, tags, \bar{v})\in ReadL$:


\hspace{1cm}$send\langle \texttt{read-return}, opid,v \rangle$ to client $clientid$
\label{line:sendtoread1_var}

\hspace{1cm}$ReadL\leftarrow ReadL - \{clientid, opid, X, tvec, \bar{v})\}$
\label{line:read-remove1_var}

\underline{On $receive\langle\texttt{read},opid, X\rangle$ from client $id$}:

if $L[X] \neq \{\}$ and $L[X].Highesttagged.tag \geq M.tagvec[X]$: \label{line:readimmediatereturncheck1_var}

\hspace{1cm}$send\langle \texttt{read-return}, opid, L[X].HighestTagged.val \rangle$ to client $id$    \hspace{10pt} \# return locally
\label{line:readimmediatereturn_var}

else if  $\{s\} \in \mathcal{R}_X$ \label{line:readreturnimmediatecheck2_var}

\hspace{1.5cm}$v\leftarrow \Psi_{\{s\}}^{ObjectIndex(X)}( M.val)$

\hspace{1.5cm} send $\langle \texttt{read-return}, opid, v \rangle$ to client $id$ 

\label{line:send-to-read1_var}
 
else:  \hspace{10pt} \# contact other nodes

\hspace{1cm}$ReadL \leftarrow ReadL\cup (id, opid, X, M.tagvec ,w_1, w_2, \ldots, w_N)$, where $w_{i} = \begin{cases} M.val & \textrm{if }i = s\\ \bot & \textrm{otherwise} \end{cases}$ 
\label{line:addtoreadl_var}

\hspace{1cm}send$\langle\texttt{val\_inq},id, opid, X, lt(M.tagvec)\rangle$ to every node $j$ where $j\neq i$
\label{line:valinq_send_var}
\end{algorithm*}

\end{figure}

\begin{figure}[htp!] 
\begin{algorithm}[H]
\footnotesize

\underline{On $\texttt{Eventual-Broadcast-receive}(\langle\texttt{del},X), j, lt)$}:
 
 $DelL[X] \leftarrow DelL[X] \cup \{(lt,j)\}$ 
\label{line:delmessage_var}

\underline{On $receive\langle\texttt{val\_inq},clientid, opid,\overline{X}, wantedlts\rangle$ from node $j$}:
 
if there exists $v$ such that $(tag,v) \in L[\bar{X}]$ and $lt(wantedlts[\overline{X}]) = lt(tag) $ :

\hspace{0.5cm}$send\langle\texttt{val\_resp},clientid, opid \overline{X},v,wantedlts[\overline{X}]\rangle$ to node $j$

else

\hspace{0.5cm} $ResponsetoValInq \leftarrow (M.val,lt(M.tagvec))$
\label{line:Responsetovalinqinit_Var}

\hspace{0.5 cm} For all $X \in \mathcal{X}_{s}$
\label{line:forloop1_var}

\hspace{1 cm} If $lt(M.tagvec[X]) \neq wantedlts[X]$: 

\hspace{1.5 cm} If there exists unique $tag,v$ such that $(tag,v) \in L[X]$ and $lt(tag)=wantedlts[X]$
\label{line:responsetovalinqcondition1_var}

 \hspace{2 cm}$ResponsetoValInq.val \leftarrow \Gamma_{s,ObjectIndex(X)}
 (ResponsetoValInq.val,v,\mathbf{0})$
\label{line:ResponsetoValInqvalupdate1_var} 

  \hspace{2 cm}$ResponsetoValInq.lts[X] \leftarrow \mathbf{0}$.
\label{line:ResponsetoValInqtagupdate1_var} 

\hspace{2 cm}If there exists unique $tag,w$ such that $(tag,w) \in L[X]$ and $lt(tag)=wantedlts[X]$
\label{line:responsetovalinqcondition2_var}

\hspace{2.5 cm}$ResponsetoValInq.val \leftarrow \Gamma_{s,ObjectIndex(X)} (ResponsetoValInq.val,\mathbf{0},w)$
\label{line:ResponsetoValInqvalupdate2_var} 
  \hspace{2.5 cm}$ResponsetoValInq.lts[X] \leftarrow wantedlts[X]$.
\label{line:ResponsetoValInqtagupdate2_var} 
\hspace{0.5cm} $send \langle \texttt{val\_resp\_encoded}, ResponsetoValInq,clientid,opid, \overline{X}, wantedlts \rangle$ to node $j$.

\underline{On $receive\langle \texttt{val\_resp\_encoded},\overline{M},clientid,opid, \overline{X}, requestedlts \rangle$ from node $j$}:

For $X \in \mathcal{X}$

\hspace{0.5 cm}
$Error1[X],Error2[X] \leftarrow 0,$

$Modified\_codeword \leftarrow \overline{M}.val$

\label{line:modifiedcodeword_init_var}

If there exists a tuple $(clientid,opid, \overline{X},requestedtags,\overline{w}) \in ReadL$ for some $\bar{w} \in \overline{\mathcal{W}}$ such that $lt(requestedtags)=requestedlts$ \label{line:valrespencoded_readcheck_var}

For all $X \in \mathcal{X}_{j}$,   \label{line:valrespencoded_objectforloop_var}

\hspace{0.5 cm} If $requestedlts[X] \neq \overline{M}.lts[X]$
\label{line:metaerrorcondition_var}

\hspace{1.0 cm} If $\overline{M}.lts[X] \neq \mathbf{0}$ and there exists unique $(tag,w)\in L[X]$ such that $lt(tag)=\overline{M}.lts[X]$
\label{line:metaerrorcondition1_var}

\hspace{1.5 cm}   $Modified\_codeword \leftarrow \Gamma_{s,ObjectIndex(X)}
 (Modified\_codeword,w, \mathbf{0})$
 \label{line:Modifiedcodeword1_var}

\hspace{1.0 cm} else if $\overline{M}.lts[X] \neq 0$:  $Error1[X] \leftarrow 1.$
\label{line:seterror1_var}

 \hspace{1.0 cm} If $Error1[X] \neq 1$ and there exists unique $v$ such that $(tag,v)\in L[X]$ and $lt(tag) = requestedlts[X]$
 \label{line:metaerrorcondition2_var}

  \hspace{1.5 cm}  $Modified\_codeword \leftarrow \Gamma_{s,ObjectIndex(X)}
 (Modified\_codeword,\mathbf{0},v)$
 \label{line:Modifiedcodeword2_var}

 \hspace{1.0 cm} else:  \label{line:seterror2_var}$Error2[X] \leftarrow 1$

    If $\bigwedge_{X \in \mathcal{X}} \bigg(Error1[X] \wedge Error2[X]=0\bigg)$
  \label{line:errorcondition_var}  
  
\hspace{0.5cm}  $ReadL\leftarrow ReadL-\{(clientid,opid,X,requestedtags,\overline{w})\} \cup \{(clientid,opid,X,requestedtags,\overline{v})\}$, where $\overline{v} \in \overline{W}$ is generated so that $\Pi_{i}(\overline{v}) \leftarrow  \begin{cases} \Pi_i(\overline{w}) & \textrm{if  }i \neq j \\ Modified\_codeword & \textrm{if } i = j \end{cases}$ 
\label{line:valrespencoded_addtuple_var}

\hspace{1cm}$S\leftarrow \{i  ~|~\Pi_i(\overline{v}) \neq \bot\}$, $\ell \leftarrow ObjectIndex[X]$

\hspace{1cm}if there exists $T \in \mathcal{R}_\ell$ such that $T \subseteq S$:
\label{line:readreturncheck_var}

\hspace{1.5cm}$v\leftarrow \Psi_{T}^{(\ell)} \left(\Pi_{T}(\overline{v})\right)$
\label{line:valresp_handle_decoding_var}

\hspace{1.5cm}if $clientid =\texttt{localhost}$:

\hspace{2cm}$L[X]\leftarrow L[X] \cup \{( requestedtags[X], v)\}$
\label{line:readaddtolist_var}

\hspace{1.5cm}else: $send\langle \texttt{read-return}, opid, v \rangle $ to client $clientid$
\label{line:read-send-val_resp_encoded_var}

\hspace{1.5cm}$ReadL\leftarrow ReadL- \{(clientid,opid, {X}, requestedtags, \bar{v}): \exists \bar{v}  \textrm{ s. t. }(clientid,opid,{X}, requestedtags, \bar{v}) \in ReadL \}$
\label{line:read-remove2_var}

\underline{On $receive\langle \texttt{val\_resp},X,v,clientid, opid, requestedlts\rangle$ from node $j$}:

if there exists $\overline{v}$ such that $(clientid, opid, X, requestedtags,\overline{v})\in ReadL$ and $lt(reqestedtags)=requestedlts)$: 
\label{line:valresp_handle_condition_var}

\hspace{0.5 cm}if $clientid = \texttt{localhost}$ 

\hspace{1cm}$L[X] \leftarrow L[X] \cup  (requestedtags[X],v)$
\label{line:valrespaddtolist_var}

\hspace{0.5cm}else: $send \langle \texttt{read-return}, opid, v \rangle$ to client $clientid$ 
\label{line:valresp_handle_var}

\hspace{0.5 cm}$ReadL\leftarrow ReadL - \{ (clientid, opid, X, requestedtags, \overline{v})\}$
\label{line:read-remove3_var}

\underline{On receipt of \emph{app$(X,v, t)$} from node $j$}:

$InQueue \leftarrow InQueue \cup \{(j,X,v,t)\}$
\caption{Server protocol of low-cost variant of \CausalEC: Transitions for input actions.}
\label{alg:input_actions_var}
\end{algorithm}

\end{figure}

\begin{figure}[htp!] 

\begin{algorithm}[H]
\footnotesize

$\texttt{Apply}\_\texttt{InQueue}$: precondition: $InQueue \neq \{\}$ 

effect:

$(j,X,v,t) \leftarrow InQueue.Head$

If $t.ts[p]\leq vc[p]$ for all $p\neq j$, and $t.ts[j]=vc[j]+1$:
\label{line:applycondition_var}

\hspace{1cm}$InQueue \leftarrow InQueue - \{ (j,X,v,t)\}$, $vc[j]\leftarrow t.ts[j]$, $L[X] \leftarrow L[X] \cup (t,v)$ 
\label{line:list_applyaction_var}
\label{line:vcincrement_apply_var}\label{line:inqueue_remove_var}

\hspace{1cm} $lttotag \leftarrow lttotag \cup (lt(t),t)$ 
\label{line:lt-to-tag}

\hspace{1cm}for every $opid$ such that $(clientid,opid, X, tvec, \bar{v})\in ReadL, clientid \neq \texttt{localhost}, tvec[X] \leq t$: 
\label{line:readreturn_apply_condition_var}

\hspace{1.5cm}$send \langle \texttt{read-return}, opid, v \rangle$ to client with id $clientid$
\label{line:readreturn_apply_var}

\hspace{1.5cm}$ReadL\leftarrow ReadL - \{ (clientid, opid, X, tvec, \bar{v})\}$
\label{line:read-remove4_var}

\hspace{1cm}for every $opid$ such that $(clientid,opid, X, tvec, \bar{v})\in ReadL, clientid = \texttt{localhost}, tvec[X] = t$: 
\label{line:readreturn_apply_condition2_var}

\hspace{1.5cm}$ReadL\leftarrow ReadL - \{ (clientid, opid, X, tvec, \bar{v})\}$
\label{line:read-remove5_var}

$\texttt{Encoding}$: precondition: none

effect:

For every $X \in \mathcal{X}_{s}$ such that $L[X] \neq \{\}$ and $L[X].Highesttagged.tag > M.tagvec[X]$
\label{line:Forloopupdatetag1_var}

\hspace{0.5 cm} If there exists a tuple $(M.tagvec[X],val) \in L[X]$
\label{line:conditionforupdate_var}

\hspace{1cm} Let $v$ be a value such that $(L[X].Highesttagged, v) \in L[X].$

\hspace{1cm} $M.val \leftarrow \Gamma_{s,k}(M.val, val,v)$, $M.tagvec[X] \leftarrow  L[X].Highesttagged.tag$
\label{line:updatetag1_var}

\hspace{1 cm} $R \leftarrow \{i \in \mathcal{N}: X \in \mathcal{X}_{i}\}$

\hspace{1 cm} $\texttt{Eventual-broadcast}(\langle \texttt{del},X \rangle,lt(L[X].Highesttagged.tag),R)$
\label{line:Delete_send1_var}

\hspace{1 cm} $DelL[X] \leftarrow DelL[X] \cup \{L[X].Highesttagged.tag, s)\}$
\label{line:Delete_local1_var}

\hspace{0.5 cm} else if there exists no tuple $(\texttt{localhost}, \overline{opid}, X, tvec, \overline{w})$ in $ReadL$ with $tvec[X]=M.tagvec[X]$

\hspace{1 cm} Generate a unique operation identifier $opid.$

\hspace{1 cm} $ReadL \leftarrow ReadL \cup \{ (\texttt{localhost}, opid, X, M.tagvec, w_1, w_2, \ldots, w_N, ) \}$, where $w_i = \begin{cases} M.val & \textrm{if }i=s \\ \bot & \textrm{otherwise}\end{cases}$
\label{line:readlentryencoding_var}

\hspace{1 cm}$send \langle \texttt{val\_inq},\texttt{localhost}, opid, X, lt(M.tagvec)\rangle$ to every node $j\neq s$
\label{line:valinq_send2_var}

For every $X \notin \mathcal{X}_{s}$ such that $L[X] \neq \{\}$ and $L[X].Highesttagged.tag > M.tagvec[X]$
\label{line:Forloopupdatetag2_var}

\hspace{0.5 cm} $R \leftarrow \{i \in \mathcal{N}: X \in \mathcal{X}_{i}\}$
\label{line:Delete2_Rset_var}

\hspace{0.5cm}  Let $U$ be the set of all tags $t$ such that $ R \subseteq \{i \in \mathcal{N}: \exists \hat{t} \textrm{such that} (\hat{t},\hat{lt}) \in lttotag, (\hat{lt},i) \in DelL[X], \hat{t} \geq t\}$
\label{line:Delete2_Uset_var}

\hspace{0.5cm}  Let $\overline{U} = \{t \in \mathcal{T}: \exists  val, (t,val) \in L[X], t > M.tagvec[X]\}$
\label{line:Delete2_Ubarset_var}

\hspace{0.5 cm} {If $U \cap \overline{U} \neq \{\}$}
\label{line:Delete2_condition_var}

\hspace{1 cm} $M.tagvec[X] \leftarrow \max(U \cap \overline{U})$, $DelL[X] \leftarrow DelL[X] \cup \{lt(\max(U\cap \overline{U}), s)\}$
\label{line:Delete_local2_var}, \label{line:updatetag2_var}

\hspace{1 cm}  $\texttt{Eventual-broadcast}(\langle \texttt{del},X\rangle,\max(U\cap \overline{U},\mathcal{N})$
\label{line:Delete_send2_var}

$\texttt{Garbage}\_\texttt{Collection}$: precondition: None

effect: 

For $X \in \mathcal{X}$

\hspace{0.5cm} {Let $S$ be the set of all tags $t$ such that $\{i: \exists \hat{t} \textrm{such that} (lt(\hat{t}),\hat{t}) \in lttotag, (lt(\hat{t}),i) \in DelL[X], \hat{t} \geq t \} = \mathcal{N}$}
\label{line:sset_var}

\hspace{.5 cm} $t{max}[X] \leftarrow \max (S)$
\label{line:tmax_var}

\hspace{0.5 cm} Let $\overline{S}$ be the set of all tags $t$ such that $\{i: \exists (lt(t),t) \in lttotag, (lt(t),i) \in DelL[X] \}  = \mathcal{N}$
\label{line:s1set_var}

\hspace{0.5cm} $T \leftarrow \{tvec[X]: \exists clientid, opid, i, \overline{w},\bar{X} \textrm{ s.t. } (clientid, opid, \bar{X}, i, tvec, \overline{w}) \in ReadL { \textrm{ and } tvec[X] < M.tagvec[X]}\}$ 
\label{line:pendingreads_var}

\hspace{0.5cm} If $ \big(t{max}[X]=M.tagvec[X]\big) \textrm{AND~} \big(M.tagvec[X] \in \overline{S}) \big) {\textrm{ AND~ } \big(L[X] = \{\} \textrm{ OR } L[X].Highesttagged.tag \leq M.tagvec[X]\big)}$
\label{line:GCcondition1_var}

\hspace{1 cm}$L[X] \leftarrow L[X] - \{(tag,v): \exists (tag,v) \in L[X] \textrm{ s.t. } tag \leq t{max}[X], tag \notin T\}$
\label{line:GC1_var}

\hspace{0.5cm} else if $\big(t{max}[X] < M.tagvec[X]\big) \textrm{AND~} \big(X \notin \mathcal{X}_{s}\big)$
\label{line:GCcondition3_var}

\hspace{1 cm}$L[X] \leftarrow L[X] - \{(tag,v): \exists (tag,v) \in L[X] \textrm{ s.t. } tag \leq t{max}[X], tag \notin T\}$
\label{line:GC3_var}

\hspace{0.5cm} else: $L[X] \leftarrow L[X] - \{(tag,v): \exists (tag,v) \in L[X] \textrm{ s.t. } tag  < t{max}[X], tag \notin T\}$
\label{line:GC2_ar}

\hspace{0.5 cm} $R \leftarrow \{i \in \mathcal{N}: X \in \mathcal{X}_{i}\}$
\label{line:Delete3_Rset_var}

\hspace{0.5 cm}  Let $U$ be the set of all tags $t$ such that $ R \subseteq \{i \in \mathcal{N}: \exists \hat{t} \geq t \textrm{ such that } (lt(\hat{t}),\hat{t}) \in lttotag, (lt(\hat{t}),i) \in DelL[X], \hat{t} \geq t\}$
\label{line:Delete3_Uset_var}

\hspace{0.5 cm} If $U \neq \{\}$ and $X \in \mathcal{X}_{s}$
\label{line:Delete_cond3_var}

\hspace{1 cm} $\texttt{Eventual-broadcast}\langle \texttt{del},X\rangle,max(U),\mathcal{N})$ 
\label{line:Delete_send3_var}

\caption{Server protocol for  low-cost variant of \CausalEC: Transitions for internal actions of node $s$}
\label{alg:internal_actions_var}

\end{algorithm}

\end{figure}

%

\newpage
\section{Transient Storage Overheads Analysis}
\label{app:transient}

We provide some details of our storage overhead calculations. 
Assume that for any object $X$, the arrival rate is a renewal random process with expected inter-arrival time $1/\rho_{w,X}.$ The arrival processes for each object is statistically independent of the other objects.
Suppose that that an $\texttt{Encoding}$ action is performed for every object version (i.e., a write or $app$ message) for every server, and that the write arrival rate $\rho_{w,X}$ is much smaller than $1/T$, where $T$ is the the maximum round-trip time in the network. That is network delay $T$ of propagation of $app$ and encoding messages can be neglected. 

We assume that each periodically performs a garbage collection action every $T_{gc} \gg T $ seconds. More precisely, for ease of statistical analysis, we assume that a node $i$ performs garbage collection action at times $, T_{gc},2T_{gc}, 3T_{gc}\ldots.$ Upto two garbage collection actions can be required for removing an object version at the server. To see this, observe that nodes in $R=\{i: X \in \mathcal{X}_{i}\}$ send $\texttt{Del}$ messages to each other as part of \texttt{Encoding} actions in line \ref{line:Delete_send1} in Algorithm \ref{alg:internal_actions}. Once a version arrives and these messages are sent, at a time $k T_{gc},$ the nodes in $R$ send $\texttt{Del}$ to all other nodes in line \ref{line:Delete_send3} of Algorithm \ref{alg:internal_actions}. Assuming $T_{gc} \gg T,$ all the nodes in $\mathcal{N}-\mathcal{R}$ send $\texttt{Del}$ messages at time $(k+1)T_{gc}$ in line \ref{line:Delete_send2}. These messages are processed at time $(k+2)T_{gc}$ at which time the corresponding version is removed from $L[X]$. Thus any version of the object that arrives at time $t$ stays at the server until time $\overline{t},$ where $\overline{t}$ is the smallest time such that $\overline{t} > t + 2 T_{gc}, \overline{t}=k T_{gc}$ for some integer $k \geq 2.$

Hence, for any node $i,$ for any object $X,$ the overhead of the lists at time $t$ is equal to $S_X(t)B,$ where  $S_X(t) = N_X(t)-N_X\left(\lfloor \frac{t-2T_{gc}}{ T_{gc}}\rfloor T_{gc}\right)$ where $N_X(t)$ is the number of arrivals for object $X$ (via write or app messages) by time $t$.  Note that we can write:
$$E\left[S_X(t)\right] \leq E\left[N_X(t)-N_X(t-3T_{gc})\right] \stackrel{t \to \infty}{\longrightarrow} {3T_{gc} \rho_{w,X}} $$
where we have used Blackwell's renewal theorem in the final step above\cite{gallager2013stochastic}.

We now apply this for the $YCSB$ workload parameters, with $120 \times 10^{6}$ objects. Assume for simplicity of analysis that all objects are $B$ bits each. We label the objects as $\mathcal{X}=\{X_1,X_2,\ldots,X_{120 \times 10^{6}}\}$, where we denote $X_{\ell}$ as the object with the $\ell$th largest arrival rate; thus $X_{1}$ is the most popular object with the highest arrival rate.    With a Zipf distribution with parameter $0.99,$ for an arrival rate of $10^5$ writes per second to the overall data store, the  arrival rate of $X_{\ell}$ is
$$\rho_{w,X_{\ell}}=10^{5} \times \frac{\frac{1}{\ell^{0.99}}}{\sum_{n=1}^{120\times 10^{6}}\frac{1}{n^{0.99}}}  .$$
If erasure coding is used for $95\%$ of the objects and a garbage collection performed every $T_{gc}= 120s$ seconds, then  storage overhead due to history lists \emph{per object} is:
 \begin{align*}&B \times \frac{1}{0.95 \times 120\times 10^{6}} \times   \sum_{\ell=120\times 10^{6}\times 0.05}^{120 \times 10^{6}}E[S_{X_{\ell}}(t)]
 \\&\stackrel{t \to \infty}{\longrightarrow} B \times  \frac{3 T_{gc}}{0.95 \times 120\times 10^{6}} \times   \sum_{\ell=120\times 10^{6}\times 0.05}^{120 \times 10^{6}} \left(10^{5} \times  \frac{\frac{1}{\ell^{0.99}}}{\sum_{n=1}^{120\times 10^{6}}\frac{1}{n^{0.99}}} \right)
 \\&\approx 0.05 B
 \end{align*}